\documentclass[12pt]{article}
\usepackage[utf8]{inputenc}
\usepackage{amsmath}
\usepackage{amsfonts}
\usepackage[english]{babel}
\usepackage{amsthm}
\usepackage{url}
\usepackage{hyperref}
\usepackage{float, graphicx}
\usepackage{subcaption}
\usepackage{natbib}
\usepackage{xcolor}
\usepackage{multirow}
\usepackage{amsfonts}
\usepackage{amssymb}
\usepackage{dsfont}

\newcommand{\indep}{\perp \!\!\! \perp}

\newtheorem{theorem}{Theorem}
\newtheorem{corollary}{Corollary}
\newtheorem{lemma}{Lemma}
\newtheorem{definition}{Definition}
\newtheorem{proposition}{Proposition}
\newtheorem{assumption}{Assumption}

\newtheorem{example}{Example}

\def\lo{_}
\def\hi{^}
\def\real{\mathbb{R}}
\def\ca#1{{\cal{#1}}}
\def\of{\circ}
\def\inv{^{-1}}
\def\ka{\kappa}
\def\tr{\mathrm{tr}}
\def\trans{^\mathsf{T}}

\allowdisplaybreaks

\title{A Reproducing-Kernel-Based Nonparametric Test for Conditional Independence of Functional Data}
\author{Yin Tang and Bing Li}

\begin{document}

\maketitle

\begin{abstract}
Conditional independence is a fundamental concept in many areas of
statistical research, including, for example, sufficient dimension
reduction, causal inference, and statistical graphical models. In many
modern applications, data arise in the form of random functions, making
it important to determine whether two random functions are conditionally
independent given a third. However, to the best of our knowledge,
existing conditional independence tests in the literature apply only to
multivariate data, and extensions to the functional setting are not
available. To fill this gap, we develop a reproducing-kernel-based test for conditional
independence of random functions based on the conjoined conditional covariance operator (CCCO). We rigorously derive the asymptotic
distribution of the CCCO estimator using a recently established
sharpened convergence rate for the regression operator
\citep{choi2026sharpened}. Based on this result, we construct a test
statistic using the spectral decomposition of the operator appearing in
the asymptotic distribution. The proposed method is illustrated through
applications to an activity and biometrics dataset and a macroeconomic
dataset.
    
\end{abstract}


\def\lo{_}
\def\hi{^}

\section{Introduction}
Conditional independence is a fundamental concept underlying a wide range of multivariate data analysis methods. For instance, in Sufficient Dimension Reduction, we need to find $\beta$ such that $Y\indep X|\beta^\mathsf{T}X$ \citep{li2018sufficient}. In graphical models, nodes $i$ and $j$ are disconnected if and only if 
\begin{align*}
    X \lo i \indep X \lo j | \{X \lo 1, \ldots, X \lo p \} \setminus \{X \lo i, X \lo j\}, 
\end{align*}
where $X \lo 1, \ldots, X \lo p$ represent the random variables on nodes $1, \ldots, p$. See \cite{lauritzen1996graphical} and \cite{koller2009probabilistic}. In causal inference, the common cause principle indicates that two dependent random variables $X$ and $Y$ may have a common cause $Z$ such that $X\indep Y|Z$ \citep{peters2017elements}. Conditional independence plays a central role in all three applications. 

It is thus of great importance to test  whether two random variables are conditionally independent given a third. \cite{li2020nonparametric} provides a comprehensive review of nonparametric conditional independence (CI) tests for continuous variables.
For example, under the Gaussian assumption, partial correlation and conditional correlation can be used to characterize the conditional dependency between two random variables. In fact, the Gaussian assumption can be relaxed as indicated in \cite{baba2004partial}. Representative works in this area include the discretization-based test \citep{huang2010testing}, metric-based tests \citep{su2008nonparametric, huang2016flexible}, the permutation-based two-sample test \citep{doran2014permutation}, and the mutual independence test under transformations \citep{cai2022distribution}.
Also, some regression-based methods can be employed for this purpose under some additivity or linearity assumptions (see, for example, \cite{zhang2017causal}).

To further relax modeling assumptions, {reproducing-kernel-based} methods have recently
gained prominence by embedding probability distributions into feature spaces
induced by kernels. 
Within the reproducing-kernel Hilbert space (RKHS) framework, conditional independence can be characterized via the
Conjoined Conditional Covariance Operator (CCCO); see, for example,
\cite{fukumizu2004dimensionality,fukumizu2007kernel,lee2016additive}.
Building on this characterization, \cite{zhang2012kernelbased} proposed a
kernel-based conditional independence test whose statistic is closely related
to the Hilbert--Schmidt norm of the CCCO introduced in
\cite{fukumizu2007kernel}.
Related developments include \cite{huang2022kernel}, which introduces a kernel
partial correlation coefficient as a measure of conditional independence, and
\cite{sheng2023distance}, which investigates the connection between
distance-based and {reproducing-kernel-based} measures of conditional independence.
As a sidenote, we emphasize that the reproducing kernels considered in this paper are fundamentally different from the smoothing kernels used in functional nonparametric regression and estimation \citep{ling2018nonparametric}.

{Our conditional independence test for functional data has wide applications in nonparametric Functional Data Analysis. For example, it could be used for variable selection in the function-on-function nonparametric regression setting (see, for example, \cite{sang2026nonlinear}), where we can drop a predictor $X_2$ if $Y \indep X_2 \mid X_1$. Here, $X_1$, $X_2$ and $Y$ can all be random functions or vectors of random functions.
Another example is the functional graphical model \citep{qiao2019functional,li2018nonparametric}, where we can determine the number of edges in a statistical graphical model, with the vertex variables being random functions.}

{Although a rich literature exists on conditional independence testing for multivariate data, analogous methods for random functions have not, to the best of our knowledge, been developed. This paper fills this gap by developing a reproducing-kernel-based conditional independence test for functional data, where the three random elements involved in the test are random functions or vectors of random functions.} Such functional observations are increasingly prevalent in modern data analysis {\citep{goia2016introduction,aneiros2019recent,aneiros2022functional}}, particularly with the rapid growth of applications involving smart wearable devices. See, for example, \cite{ramsay2007applied} and \cite{acar2025functional}. A flexible conditional independence testing procedure would constitute a
valuable addition to the toolbox of functional data analysis, enabling the
study of complex interconnections among functional variables. Such
interconnections arise in a wide range of problems, including
function-on-function regression \citep{sun2018optimal,sang2026nonlinear},
functional graphical models \citep{qiao2019functional,li2018nonparametric},
and functional sufficient dimension reduction
\citep{ferre2003functional,li2017nonlinear,li2022weak}.

\def\lo#1{_{#1}}
\def\inv{^{-1}}

In addition to the importance of conditional independence testing for
functional data, we would like to highlight a particularly important
theoretical contribution of our extension. As the pioneering work on
{reproducing-kernel-based} conditional independence testing,
\citet{zhang2012kernelbased} provided only a sketch of the proof for the
asymptotic null distribution, without a detailed list of conditions under
which the asymptotic distribution holds.

{
A key difficulty arises from the fact that the test statistic involves
the inverse of a covariance operator. In practice, the inverse must be
estimated using Tikhonov regularization, which slows down the
convergence rate of a related regression operator. Until recently, the
best known convergence rate for this operator was
\begin{align}\label{eq:epsilon_n_beta_wedge}
\epsilon_n^{\beta \wedge 1} + n^{-1/2}\epsilon_n^{-1},
\end{align}
where $n$ is the sample size, $\epsilon_n$ is the Tikhonov regularization
parameter, and $\beta>0$ characterizes the smoothness of the regression
relationship between the two random variables involved. See, for
example, \citet{li2017nonlinear}. The optimal rate implied by
\eqref{eq:epsilon_n_beta_wedge} is $n^{-1/4}$, which can be achieved when
$\beta\ge1$ with the optimal tuning $\epsilon_n \asymp n^{-1/4}$.
However, our calculations (see Section~\ref{sec-ran-func}) indicate that this rate is not
sufficiently fast to guarantee the asymptotic distribution claimed in
\citet{zhang2012kernelbased}. 
Recently, under a  slightly different setting, \cite{choi2026sharpened} sharpened the convergence rate of a similar regression operator to 
\begin{align}\label{eq:n-1-2-eps-beta}
n^{-1/2}\epsilon_n^{(\beta\wedge1)-1}
+\epsilon_n^{\beta\wedge1}
+n^{-1}\epsilon_n^{-(3\alpha+1)/(2\alpha)}
+n^{-1/2}\epsilon_n^{-(\alpha+1)/(2\alpha)},
\end{align}
where $\alpha>1$ characterizes the decay rate of the eigenvalues of the covariance operator.
It turns out that we can use that approach to sharpen the convergence rate of the regression operator in the present context, and with the sharpened rate, we can rigorously establish the asymptotic distribution of our test statistics parallel to that claimed in \cite{zhang2012kernelbased} in the multivariate setting. 
}

Although our analysis is carried out in the functional setting, the above
result also applies to the multivariate setting. Consequently,
Theorem~\ref{thm-clt-sigmahat} provides the first rigorous justification
of the asymptotic null distribution for the {reproducing-kernel-based} conditional
independence test. Given the widespread use of this test, establishing
its asymptotic validity is both important and significant.

The rest of the paper is organized as follows.
In Section~\ref{sec-rkhs}, we review  the background on RKHS and several linear operators needed for constructing the test statistic at the population level. In Section~\ref{sec-test-ci}, we introduce the test statistic
for conditional independence. In Section~\ref{sec-asymp-ccco-cpco}, we rigorously develop the asymptotic distribution of the proposed {reproducing-kernel-based} test statistic for conditional independence for functional data. We also carefully discuss the key conditions needed in the asymptotic development. 
In Section~\ref{sec-ran-func}, we introduce 
practical methods for approximating random functions involved in constructing the test statistic. Section~\ref{sec-implementation}
presents the sample-level implementation, including estimation of the
eigenvalues appearing in the asymptotic distribution using an
acceleration method, selection of tuning parameters, and the resulting
algorithms. Section~\ref{sec-simulation} reports results from simulation
studies, and Section~\ref{sec-application} illustrates the proposed
method through an application to the WISDM dataset. Due to the page limit, another application to the WDI dataset \citep{WDI} is placed in the Supplementary Material. Some illustrative examples and proofs of most results are also deferred to Appendix~\ref{app-example} and Appendix~\ref{sec-proof} in the Supplementary Material to avoid interrupting the main exposition.

\def\real{\mathbb{R}}
\def\ca#1{{\cal{#1}}}
\def\of{\circ}
\def\inv{^{-1}}

\section{Background of RKHS}\label{sec-rkhs}

\subsection{Construction of nested Hilbert spaces}

Let $(\Omega,\mathcal{F},P)$ be a probability space, let $T$ be an interval in $\real$ which, without loss of generality, is assumed to be $[0, 1]$, and let $\mathcal{H}_X, \ca H \lo Y, \ca H \lo Z$ be a separable Hilbert spaces of functions defined on $T$. Let   $\ca F_X, \ca F \lo Y, \ca F \lo Z$ be the Borel $\sigma$-fields generated by the open sets in $\mathcal{H}_X, \ca H \lo Y, \ca H \lo Z$, respectively. Let $\ca F \lo X \times \ca F \lo Y \times \ca F \lo Z$ be the product $\sigma$-field and let  $(\ca H \lo X \times \ca H \lo Y \times \ca H \lo Z, \ca F \lo X \times \ca F \lo Y \times \ca F \lo Z) $ be the product measurable space. Furthermore, let $(X, Y, Z): \Omega \to \Omega \lo X \times \Omega \lo Y \times \Omega \lo Z$ be a random element measurable with respect to $\ca F / ( \ca F \lo X \times \ca F \lo Y \times \ca F \lo Z)$, and let $P \lo X = P \of X \inv $, $P \lo Y = P \of Y \inv$,   $P \lo Z = P \of Z \inv$, and $P \lo {XYZ} = P \of (X,Y,Z)\inv$   denote the distributions of $X$, $Y$, $Z$, and $(X,Y,Z)$, respectively. Our goal is to construct a test for whether $X \indep Y | Z$ holds.

\def\ka{\kappa}

Having defined the random elements $X,Y,Z$, we then introduce positive definite kernels 
\begin{align*}
    \ka \lo X : \ca H \lo X \times \ca H \lo X \to \real, \quad 
    \ka \lo Y : \ca H \lo Y \times \ca H \lo Y \to \real, \quad 
    \ka \lo Z : \ca H \lo Z \times \ca H \lo Z \to \real.  \quad 
\end{align*}
Let $\mathcal{G}_X, \ca G \lo Y, \ca G \lo Z$ be the RKHS's generated by $\ka \lo X, \ka \lo Y, \ka \lo Z$, respectively. For further information about RKHS, see, for example, \cite{aronszajn1950theory,berlinet2004reproducing}. We assume  that $\kappa_X$ is uniquely determined  by the inner product in $\mathcal{H}_X$; that is, for any $f,g \in \mathcal{H}_X$, $\kappa_X(f,g)$ is a function of $\langle f,f \rangle_{\mathcal{H}_X}$, $\langle g,g \rangle_{\mathcal{H}_X}$ and $\langle f,g \rangle_{\mathcal{H}_X}$. \cite{li2017nonlinear,li2018nonparametric} refer to this special structure as nested Hilbert spaces:  $\mathcal{H}_X$ is nested in $\mathcal{G}_X$ in the sense that the inner product in $\ca H \lo X$ determines the kernel in $\ca G \lo X$.  We make the same assumptions on  $\kappa_Y$, $\kappa_Z$, $\mathcal{G}_Y$, $\mathcal{G}_Z$ as well.

\subsection{Mean element and covariance operator}

Before introducing the mean element and covariance operator, we first make the following assumption.

\begin{assumption}\label{ass-moment}
    The kernels $\kappa_X$, $\kappa_Y$ and $\kappa_Z$ satisfies $E[\kappa_X(X,X)]<\infty$, $E[\kappa_Y(Y,Y)]<\infty$ and $E[\kappa_Z(Z,Z)]<\infty$.
\end{assumption}

\def\E{\mathbb{E}}

With the reproducing kernel $\kappa_X$ specified, we define the mean element
of $X$ in the RKHS $\mathcal{G}_X$. Under Assumption~\ref{ass-moment}, there is a unique element $\mu \lo X \in \ca G \lo X$  satisfying
\begin{align*}
    \langle \mu_X, f \rangle_{\mathcal{G}_X}
    = E[ f(X) ], \qquad \forall f \in \mathcal{G}_X .
\end{align*}
This 
element  is defined as the mean element of $X$.
We denote this element by 
$ 
E [ \kappa_X(\cdot, X)  ] 
 $, so that the identity 
\begin{align*}
    \langle E [ \kappa_X(\cdot, X)  ], f \rangle \lo {\ca G \lo X} =    E   \langle  \kappa_X(\cdot, X) , f \rangle \lo {\ca G \lo X} 
    = E[f(X)]
\end{align*}
is satisfied. 
We define the mean elements
$\mu_Y = E[ \kappa_Y(\cdot, Y) ]$ and
$\mu_Z = E[ \kappa_Z(\cdot, Z) ]$ in the same way.

We also need the definition of characteristic kernels, {also known as}  probability-determining kernels (see \cite{fukumizu2004dimensionality,zhang2012kernelbased}).
\begin{definition}[characteristic kernel]
Let $(A, \ca F \lo A )$ be a measurable space, and
  $\kappa:A\times A \to\mathbb{R}$ be a reproducing kernel. Let $X: \Omega \to A$ and $Y: \Omega \to A$ be random elements that are  measurable with respect to $\ca F / \ca F \lo A$. We say that $\ka$ is characteristic if  
\begin{align*}
E [\ka (\cdot, X )]   =  E[ \ka (\cdot, Y ) ]\ \Rightarrow\  P \of X\inv = P \of Y\inv.
\end{align*}
\end{definition}  

\def\ali{&}

According to Proposition 2 of \cite{zhang2024dimension}, the Gaussian kernel and Laplace kernel in $L^2(T)$, defined by
\begin{align*}
    \kappa_1(\cdot,z)=  \exp (-\gamma_1\left\|\cdot-z\right\|^2),\quad\gamma_1>0, \quad \text{and} \quad  
    \kappa_2(\cdot,z)= \exp\left(-\gamma_2\left\|\cdot-z\right\|\right),\quad\gamma_2>0,
\end{align*}
respectively,  are both characteristic kernels on their corresponding RKHS's.

We then define the covariance operator as follows. Under Assumption \ref{ass-moment}, there is a unique  bounded linear operator $\Sigma_{XX}:\mathcal{G}_X \to \mathcal{G}_X$ satisfying
\begin{align*}
    \left\langle f,\Sigma_{XX}g\right\rangle_{\mathcal{G}_X}=\mathrm{cov}\left[f(X),g(X)\right]
\end{align*}
for all $f,g\in\mathcal{G}_X$. This operator is called the covariance operator of $X$. It can be equivalently expressed as
\begin{align}
    \Sigma_{XX}=E\left[\left(\kappa_X(\cdot,X)-\mu_X\right)\otimes \left(\kappa_X(\cdot,X)-\mu_X\right)\right].\label{eq-pop-cov}
\end{align}
It can be shown that $\Sigma_{XX}$ is a Hilbert-Schmidt operator on $\mathcal{G}_X$ under Assumption \ref{ass-moment}; see, for example, \cite{fukumizu2007statistical}. Linear operators $\Sigma_{YY}$ and $\Sigma_{ZZ}$ are defined similarly.
Letting $\tilde \ka$ be  centralized kernel $
    \tilde{\kappa}_X(\cdot,x)=\kappa_X(\cdot,x)-\mu_X $, 
  we can represent \eqref{eq-pop-cov} as
\begin{align}
\Sigma_{XX}=E\left[\tilde{\kappa}_X(\cdot,X)\otimes\tilde{\kappa}_X(\cdot,X)\right].\label{eq-pop-cov-rep}
\end{align}
Similarly, we can rewrite $\Sigma_{YY}$ and $\Sigma_{ZZ}$ as 
\begin{align*}
\Sigma_{YY}=E\left[\tilde{\kappa}_Y(\cdot,Y)\otimes\tilde{\kappa}_Y(\cdot,Y)\right], \quad \text{and} \quad 
\Sigma_{ZZ}=E\left[\tilde{\kappa}_Z(\cdot,Z)\otimes\tilde{\kappa}_Z(\cdot,Z)\right].
\end{align*}
We further define the (cross-)covariance operator between $(X,Y)$,  $\Sigma_{XY}: \mathcal{G}_Y \to \mathcal{G}_X$, as 
\begin{align}
\Sigma_{XY}=E\left[\left(\kappa_X(\cdot,X)-\mu_X\right)\otimes\left(\kappa_Y(\cdot,Y)-\mu_Y\right)\right]
=E\left[\tilde{\kappa}_X(\cdot,X)\otimes\tilde{\kappa}_Y(\cdot,Y)\right].\label{eq-def-cross-cov-op}
\end{align}
Equivalently, for all $f\in\mathcal{G}_X$ and $g\in\mathcal{G}_Y$, we have
\begin{align*}
    \left\langle f,\Sigma_{XY}g\right\rangle=\mathrm{cov}\left[f(X),g(Y)\right].
\end{align*}
A well-known fact is that, if $\kappa_X$ and $\kappa_Y$ are characteristic, then $X$ and $Y$ are independent if and only if $\Sigma_{XY}=0$. See, for example, \cite{fukumizu2007kernel}. This fact is often used to construct tests for independence \citep{gretton2005measuring,sejdinovic2013equivalence}. We define $\Sigma_{XZ}$, $\Sigma_{YZ}$, $\Sigma_{YX}$, $\Sigma_{ZX}$ and $\Sigma_{ZY}$ similarly.

\subsection{Conditional Covariance Operator (CCO)}

Based on the definitions of the mean element and the covariance operator, 
we define the conditional covariance operator (CCO) of $(X,Y)$ given $Z$, 
denoted by $\Sigma_{XY | Z} : \mathcal{G}_Y \to \mathcal{G}_X$, as follows:
\begin{align}
    \Sigma_{XY|Z}=\Sigma_{XY}-\Sigma_{XZ}\Sigma_{ZZ}^{\dagger}\Sigma_{ZY},\label{eq-def-cond-cov-op}
\end{align}
where $\Sigma_{ZZ}^{\dagger}$ is the Moore-Penrose inverse of $\Sigma_{ZZ}$ (see, for example, \cite{li2017linear,hsing2015theoretical}). Further define
\[
B_{X | Z} = \Sigma_{XZ} \Sigma_{ZZ}^{\dagger}, 
\qquad 
B_{Y | Z} = \Sigma_{YZ} \Sigma_{ZZ}^{\dagger}.
\]
Note that $B_{X | Z}$ and $B_{Y | Z}$ are precisely the adjoint operators 
of the regression operators introduced in Chapter 13 of \cite{li2018sufficient} 
and in \cite{li2017linear}. We adopt a different notation here (than, say, $R \lo {ZX} \hi *$ and $R \lo {ZY} \hi *$)  for 
convenience in subsequent developments.

To establish an equivalent representation of $\Sigma_{YX | Z}$, 
we first present the following lemma and corollary, 
which connect the operators $B_{X | Z}$ and $B_{Y | Z}$ 
to the regression between two centralized kernels. 
We begin by stating an additional technical assumption.

\begin{assumption}\label{ass-cond-exp} \ \ 
\begin{enumerate}
    \item[(i)] For all $g \in \mathcal{G}_X$, the function $E[g(X)|Z=\cdot]$ is an element of $\mathcal{G}_Z$;
    \item[(ii)] For all $g \in \mathcal{G}_Y$, the function $E[g(Y)|Z=\cdot]$ is an element of $\mathcal{G}_Z$.
\end{enumerate}
\end{assumption}

\begin{lemma}\label{lemma-reg-form}
If Assumptions \ref{ass-moment} and \ref{ass-cond-exp} {are satisfied}, then, for all $f \in \ker(\Sigma_{ZZ})^\perp \subset \mathcal{G}_Z$, $g \in \mathcal{G}_X$ and $h \in \mathcal{G}_Y$, we have
\begin{align}\label{eq-regression-inner}
    \langle B_{X|Z} f, g \rangle_{\mathcal{H}_X} = \langle f, E[g(X) | Z=\cdot] \rangle_{\mathcal{H}_Z}, \quad  \langle B_{Y|Z} f, h \rangle_{\mathcal{H}_Y} = \langle f, E[h(Y) | Z=\cdot] \rangle_{\mathcal{H}_Z}.
\end{align}    
\end{lemma}
\noindent
\begin{proof}
    We will only show the first part of \eqref{eq-regression-inner}, as the second part is the same statement with $X$ replaced by $Y$ and $g$ replaced by $h$. By the definition of an adjoint operator,  
    \begin{align*}
        \langle B_{X|Z} f, g \rangle_{\mathcal{H}_X} 
        = \langle \Sigma_{XZ}\Sigma_{ZZ}^{\dagger} f, g \rangle_{\mathcal{H}_X} = \langle  f, \Sigma_{ZZ}^{\dagger} \Sigma_{ZX} g \rangle_{\mathcal{H}_Z}.
    \end{align*} 
    The rest of the proof follows from Theorem 2 and Corollary 3 of \cite{fukumizu2004dimensionality}.
\end{proof}

\begin{corollary}\label{cor-reg-form}
Under Assumptions \ref{ass-moment} and \ref{ass-cond-exp}, we have
\begin{align}\label{eq-regression-form}
    E[\tilde{\kappa}_X(\cdot,X) | Z=z] = B_{X|Z} \tilde{\kappa}_Z(\cdot,z), \quad E[\tilde{\kappa}_Y(\cdot,Y) | Z=z] = B_{Y|Z} \tilde{\kappa}_Z(\cdot,z).
\end{align}
\end{corollary}

The proof of Corollary \ref{cor-reg-form} is given in Appendix \ref{cor-reg-form-proof}. As we can see from \eqref{eq-regression-form}, $B_{X|Z}$ and $B_{Y|Z}$ can be viewed as the regression coefficients of $\tilde{\kappa}_X(\cdot,X)$ and $\tilde{\kappa}_Y(\cdot,Y)$ on $\tilde{\kappa}_Z(\cdot,Z)$, which is the motivation for choosing $X|Z$ and $Y|Z$ as the subscripts of $B$.
Also, from  Corollary \ref{cor-reg-form}, we derive the equivalent representation of $\Sigma_{XY|Z}$.

\begin{proposition}\label{prop-rep-cond-cov}
Let $\Sigma_{XY|Z}$ be defined by \eqref{eq-def-cond-cov-op}. Under Assumptions \ref{ass-moment} and \ref{ass-cond-exp}, 
\begin{align}
    \Sigma_{XY|Z}=E\left[\left(\tilde{\kappa}_X(\cdot,X)-E\left[\tilde{\kappa}_X(\cdot,X)|Z\right]\right)\otimes\left(\tilde{\kappa}_Y(\cdot,Y)-E\left[\tilde{\kappa}_Y(\cdot,Y)|Z\right]\right)\right].\label{eq-rep-cond-cov-op}
\end{align}
\end{proposition}

The proof of Proposition \ref{prop-rep-cond-cov} is given in Appendix \ref{prop-rep-cond-cov-proof}. Note that the $\mu_X$ and $\mu_Y$ terms in $\tilde{\kappa}_X(\cdot,x)$ and $\tilde{\kappa}_Y(\cdot,y)$ are nonrandom. Two direct corollaries of Proposition \ref{prop-rep-cond-cov} are as follows.
\begin{corollary}
Under Assumptions \ref{ass-moment} and \ref{ass-cond-exp}, 
\begin{align}
    \Sigma_{XY|Z}=E\left[\left(\kappa_X(\cdot,X)-E\left[\kappa_X(\cdot,X)|Z\right]\right)\otimes\left(\kappa_Y(\cdot,Y)-E\left[\kappa_Y(\cdot,Y)|Z\right]\right)\right].\label{eq-rep-cond-cov-op-2}
\end{align}
\end{corollary}
\begin{corollary}\label{cor-sigma-xy-z}
Under Assumptions \ref{ass-moment} and \ref{ass-cond-exp}, for any $f \in \mathcal{G}_X$ and $g \in \mathcal{G}_Y$, we have
\begin{align*}
    \langle f, \Sigma_{XY|Z} g \rangle_{\mathcal{G}_X}=E\left\{ \mathrm{cov} [ f(X), g(Y) | Z] \right\}.
\end{align*}
\end{corollary}

For notational convenience, define the conditional centered kernel as
\begin{align}\label{eq-k-tilde-xz}
    \tilde{\kappa}_{X| Z}(\cdot, x | z)
    = \tilde{\kappa}_X(\cdot, x)
    - E\!\left[\tilde{\kappa}_X(\cdot, X)| Z = z\right]
    = \tilde{\kappa}_X(\cdot, x)
    - B_{X| Z}\tilde{\kappa}_Z(\cdot, z).
\end{align}
By construction, $\tilde{\kappa}_{X| Z}(\cdot, x | z) \in \mathcal{H}_X$.  
Analogously, we define $\tilde{\kappa}_{Y| Z}(\cdot, y | z) \in \mathcal{H}_Y$.
With these notations, the representation in \eqref{eq-rep-cond-cov-op} can be written more compactly as
\begin{align}
    \Sigma_{XY| Z}
    =
    E\!\left[
        \tilde{\kappa}_{X|Z}(\cdot, X|Z)
        \otimes
        \tilde{\kappa}_{Y|Z}(\cdot, Y|Z)
    \right].
    \label{eq-rep-ccco-cond}
\end{align}

\section{Test Construction for Conditional Independence}\label{sec-test-ci}

\subsection{Construction of Test Statistic}

\def\cov{\mathrm{cov}}

The conditional covariance operator $\Sigma \lo {XY|Z}$ is inadequate for testing conditional independence because, by Corollary \ref{cor-sigma-xy-z}, $\Sigma \lo {XY|Z} = 0$ if and only if, for all $f \in \ca G \lo X$ and $g \in \ca G \lo Y$, 
\begin{align*}
    E \{\cov [f(X), g (Y) | Z]\} = 0. 
\end{align*}
Clearly, this does not imply conditional independence between $X$ and $Y$ given $Z$. Instead, conditional independence is implied by 
\begin{align}\label{eq:cov f g Z}
 \cov [f(X), g (Y) | Z]  = 0. 
\end{align}
For this reason, we need to make a modification of $\Sigma \lo {XY|Z}$ so that the modified version corresponding to (\ref{eq:cov f g Z}).

Let $\Ddot{X}=(X,Z)$  and $\kappa_{\Ddot{X}}=\kappa_X\kappa_Z$, and let   $\mathcal{G}_{\Ddot{X}}$ be the RKHS generated by $\kappa_{\Ddot{X}}$. We define the modified conditional covariance operator as 
\begin{align}
    \Sigma_{\Ddot{X}Y|Z}=\Sigma_{\Ddot{X}Y}-\Sigma_{\Ddot{X}Z}\Sigma_{ZZ}^{\dagger}\Sigma_{ZY},\label{eq-def-ccco}
\end{align}
This operator was first introduced by \cite{fukumizu2008kernel}. 
{To distinguish this operator from $\Sigma \lo {XY|Z}$ (without the dots), we refer to the operator in \eqref{eq-def-ccco} as the conjoined conditional covariance operator and abbreviate it by CCCO  (see \cite{li2024sufficient}). }
The next theorem shows that the CCCO   can characterize conditional independence. Before stating the theorem, we need to present a modified version of Assumptions \ref{ass-moment} and \ref{ass-cond-exp} {using $\Ddot{X}$}.

\begin{assumption}\label{ass-moment-ddot}
    The kernels $\kappa_{\Ddot{X}}$, $\kappa_{Y}$ and $\kappa_Z$ satisfies $E[\kappa_{\Ddot{X}}(\Ddot{X},\Ddot{X})]<\infty$, $E[\kappa_{Y}(Y,Y)]<\infty$ and $E[\kappa_Z(Z,Z)]<\infty$.
\end{assumption}

\begin{assumption}\label{ass-cond-exp-ddot}
The following two parts hold:
\begin{enumerate}
    \item[(i)] For all $g \in \mathcal{G}_{\Ddot{X}}$, the function $E[g(\Ddot{X})|Z=\cdot]$ is an element of $\mathcal{G}_Z$; 
    \item[(ii)] For all $g \in \mathcal{G}_{Y}$, the function $E[g(Y)|Z=\cdot]$ is an element of $\mathcal{G}_Z$.
\end{enumerate}
\end{assumption}

\begin{theorem}\label{thm-ci}
If Assumptions \ref{ass-moment-ddot} and \ref{ass-cond-exp-ddot} hold and $\kappa_{\Ddot{X}} \otimes \kappa_{Y}$ is characteristic, then,
\begin{align}
    \Sigma_{\Ddot{X}Y|Z}=0\quad\Leftrightarrow\quad X\indep Y|Z.\label{eq-equiv-cond-indep}
\end{align}
\end{theorem}

This theorem was first established in Theorem 8 and Corollary 9 of \cite{fukumizu2004dimensionality},
where $X$, $Y$, $Z$ are random vectors.
On the other hand, when $X$, $Y$, $Z$ are random functions, although the proof is similar, the definitions of their corresponding Borel $\sigma$-fields $\mathcal{F}_X$, $\mathcal{F}_Y$, $\mathcal{F}_Z$ are more complicated, and we refer to Sections 1.1-1.2 of \cite{bosq2012linear} and Section~2.2, items 5-8 of \cite{shiryaev2016probability1}.
{Because of the similarity, we omit the proof of this result in the functional setting. }

By  \eqref{eq-equiv-cond-indep},   testing  the hypothesis
$    H_0:  X\indep Y|Z$ is equivalent to testing   $H \lo 0: \Sigma \lo {\ddot X Y |Z} =0$. This motivates us to use the statistic 
\begin{align*}
 T \lo n = n\|\hat{\Sigma}_{\Ddot{X}Y|Z}\|_{\mathrm{HS}}^2, 
\end{align*}
as the test statistic, where $\hat \Sigma \lo {\ddot X Y |X}$ is the sample estimate of $\Sigma \lo {\ddot X Y | Z}$, which will be constructed in the next subsection, and $\| \cdot \| \lo {\mathrm{HS}}$ is the Hilbert-Schmidt norm. Under $H \lo 0$, we expect $T \lo n$ to be small, whereas under  the alternative hypothesis $H \lo 1: \Sigma \lo {\ddot X Y|Z} \ne 0$,   $T \lo n$ should tend to be  large. We therefore reject $H \lo 0$ for sufficiently large values of $T \lo n$. 
In this paper, we will focus on the test statistic of the form  $n\|\hat{\Sigma}_{\Ddot{X}Y|Z}\|_{\mathrm{HS}}^2$.

\subsection{Estimating CCCO and its Hilbert-Schmidt norm}

We now construct an estimator of 
$\Sigma_{\Ddot{X}Y| Z}$. 
The construction proceeds in two parts.
In the first part, we derive the estimator 
$\hat{\Sigma}_{\Ddot{X}Y| Z}$ 
and therefrom the corresponding test statistic
$ 
T_n
=
n  \|
\hat{\Sigma}_{\Ddot{X}Y| Z}
 \|_{\mathrm{HS}}^2.
 $
In the second part, we develop an implementable numerical procedure 
based on coordinate representations.

Henceforth, we will use $E \lo n (\cdots)$ to denote the sample average. For example, for a generic kernel $\ka$ and a sample of random variables $X \lo 1, \ldots, X \lo n$, we use $E \lo n [\ka (\cdot, X)]$ is the function $n \inv \sum \lo {i=1} \hi n \ka ( \cdot, X \lo i)$.  We let $\hat \ka (\cdot, x)$ denote the  centered version 
\begin{align*}
    \hat \ka ( \cdot, x) = \ka (\cdot, x) - E \lo n [\ka (\cdot, X)]. 
\end{align*}
Recall that we have used $\tilde \ka (\cdot, x)$ to represent the version centered by the true mean $E [\ka (\cdot, X)]$. Thus, $\hat \ka$ can be regarded as the estimate of $\tilde \ka$. 

We begin the construction by estimating $B_{\Ddot{X}|Z}=\Sigma_{\Ddot{X}Z}\Sigma_{ZZ}^{\dagger}$. 
Each covariance operator can be estimated by replacing the population expectations with sample means. Specifically, the empirical estimators of $\Sigma_{\Ddot{X}Z}$ and $\Sigma \lo {ZZ}$ are given by
\begin{align*}
    \hat{\Sigma}_{\Ddot{X}Z}=E_n[\hat{\kappa}_{\Ddot{X}}(\cdot,\Ddot{X})\otimes\hat{\kappa}_Z(\cdot,Z)], \quad \hat{\Sigma}_{ZZ}=E_n[\hat{\kappa}_{Z}(\cdot,Z)\otimes\hat{\kappa}_Z(\cdot,Z)]
\end{align*}
Since $\hat{\Sigma}_{ZZ}$ is not  invertible, we apply Tikhonov regularization   when estimating $\hat{B}_{X|Z}$. That is, we can estimate $B_{\Ddot{X}|Z}$ by
\begin{align*}
    \hat{B}_{\Ddot{X}|Z}=\hat{\Sigma}_{\Ddot{X}Z}(\hat{\Sigma}_{ZZ}+\epsilon_nI)^{-1},
\end{align*}
where $\epsilon \lo n > 0$ is a tuning constant. 
We further estimate $\tilde{\kappa}_{\Ddot{X}|Z}(\cdot, \ddot x|z)$ by
\begin{align}
    \hat{\kappa}_{\Ddot{X}|Z}(\cdot,\Ddot{x}|z)=\hat{\kappa}_{\Ddot{X}}(\cdot,\Ddot{x})-\hat{B}_{\Ddot{X}|Z}\hat{\kappa}_Z(\cdot,z).\label{eq-est-cond-ker}
\end{align}
We perform  exactly the same sequence of operations to obtain the estimates of $B \lo {Y|Z}$ and $\tilde{\kappa}_{Y|Z} (\cdot, y | z)$ as
\begin{align*}
    \hat{B} \lo {Y|Z} = \hat{\Sigma} \lo {Y|Z}(\hat{\Sigma}_{ZZ}+\epsilon_nI)^{-1}, \quad \hat \ka \lo {Y|Z} (\cdot, y | z) = \hat \ka \lo {Y} (\cdot, y) - \hat B \lo {Y|Z} \hat \ka \lo Z ( \cdot, z). 
\end{align*}

We construct an estimator of $\Sigma_{\Ddot{X}Y|Z}$  by simply replacing the expectation in \eqref{eq-rep-ccco-cond} by sample mean and replacing $\tilde{\kappa}_{\Ddot{X}|Z}$ and $\tilde{\kappa}_{Y|Z}$ by $\hat{\kappa}_{\Ddot{X}|Z}$ and $\hat{\kappa}_{Y|Z}$. That is, 
\begin{align}
\hat{\Sigma}_{\Ddot{X}Y|Z}=E_n[\hat{\kappa}_{\Ddot{X}|Z}(\cdot,\Ddot{X}|Z)\otimes\hat{\kappa}_{Y|Z}(\cdot,Y|Z)].\label{eq-est-ccco}
\end{align}
Substituting \eqref{eq-est-cond-ker} into \eqref{eq-est-ccco}, we have
\begin{align}
\begin{split}
\hat{\Sigma}_{\Ddot{X}Y|Z}=&E_n\{[\hat{\kappa}_{\Ddot{X}}(\cdot,\Ddot{X})-\hat{B}_{\Ddot{X}|Z}\hat{\kappa}_Z(\cdot,Z)]\otimes[\hat{\kappa}_{Y}(\cdot,Y)-\hat{B}_{Y|Z}\hat{\kappa}_Z(\cdot,Z)]\}\\
=&E_n[\hat{\kappa}_{\Ddot{X}}(\cdot,\Ddot{X})\otimes\hat{\kappa}_{Y}(\cdot,Y)]
-E_n[\hat{B}_{\Ddot{X}|Z}\hat{\kappa}_Z(\cdot,Z)\otimes\hat{\kappa}_{Y}(\cdot,Y)]\\
&\quad -E_n[\hat{\kappa}_{\Ddot{X}}(\cdot,\Ddot{X})\otimes\hat{B}_{Y|Z}\hat{\kappa}_Z(\cdot,Z)]
+E_n[\hat{B}_{\Ddot{X}|Z}\hat{\kappa}_Z(\cdot,Z)\otimes\hat{B}_{Y|Z}\hat{\kappa}_Z(\cdot,Z)]\\
=&\hat{\Sigma}_{\Ddot{X}Y}-\hat{B}_{\Ddot{X}|Z}\hat{\Sigma}_{ZY}-\hat{\Sigma}_{\Ddot{X}Z}\hat{B}_{Y|Z}^*+\hat{B}_{\Ddot{X}|Z}\hat{\Sigma}_{ZZ}\hat{B}_{Y|Z}^*, 
\end{split}\label{eq-ccco-est-exp}
\end{align}
where, for a linear operator $A$, $A \hi *$ stands for the adjoint operator of $A$. 
The squared Hilbert-Schmidt norm of $\hat{\Sigma}_{\Ddot{X}Y|Z}$ is computed as follows: 
\begin{align*}
\begin{split}
\|\hat{\Sigma}_{\Ddot{X}Y|Z}\|_{\mathrm{HS}}^2=&\left\|E_n\left\{[\hat{\kappa}_{\Ddot{X}}(\cdot,\Ddot{X})-\hat{B}_{\Ddot{X}|Z}\hat{\kappa}_Z(\cdot,Z)]\otimes[\hat{\kappa}_{Y}(\cdot,Y)-\hat{B}_{Y|Z}\hat{\kappa}_Z(\cdot,Z)]\right\}\right\|_{\mathrm{HS}}^2\\
=&E_n\left[\langle\hat{\kappa}_{\Ddot{X}}(\cdot,\Ddot{X})-\hat{B}_{\Ddot{X}|Z}\hat{\kappa}_Z(\cdot,Z),\hat{\kappa}_{\Ddot{X}}(\cdot,\Ddot{X}')-\hat{B}_{\Ddot{X}|Z}\hat{\kappa}_Z(\cdot,Z')\rangle_{\mathcal{G}_{\Ddot{X}}}\right.\times\\
&\quad\quad\quad\left.\langle\hat{\kappa}_{Y}(\cdot,Y)-\hat{B}_{Y|Z}\hat{\kappa}_Z(\cdot,Z),\hat{\kappa}_{Y}(\cdot,Y')-\hat{B}_{Y|Z}\hat{\kappa}_Z(\cdot,Z')\rangle_{\mathcal{G}_{Y}}\right]\\
=&\frac{1}{n^2}\sum_{i=1}^n\sum_{j=1}^n\left[\langle\hat{\kappa}_{\Ddot{X}}(\cdot,\Ddot{X}_i)-\hat{B}_{\Ddot{X}|Z}\hat{\kappa}_Z(\cdot,Z_i),\hat{\kappa}_{\Ddot{X}}(\cdot,\Ddot{X}_j)-\hat{B}_{\Ddot{X}|Z}\hat{\kappa}_Z(\cdot,Z_j)\rangle_{\mathcal{G}_{\Ddot{X}}}\right.\times\\
&\quad\quad\quad\quad\quad\left.\langle\hat{\kappa}_{Y}(\cdot,Y_i)-\hat{B}_{Y|Z}\hat{\kappa}_Z(\cdot,Z_i),\hat{\kappa}_{Y}(\cdot,Y_j)-\hat{B}_{Y|Z}\hat{\kappa}_Z(\cdot,Z_j)\rangle_{\mathcal{G}_{Y}}\right].
\end{split}
\end{align*}

The above expression is at the operator level. In order to implement it numerically we need to turn operators to matrices via  coordinate mapping (see, for example, Section~12.3 of \cite{li2018sufficient}).    The subspace ${\mathrm{ran}}(\hat{\Sigma}_{\Ddot{X}\Ddot{X}})$ is spanned by
\begin{align*}
    \mathcal{B}_{\Ddot{X}}=\{\kappa_{\Ddot{X}}(\cdot,\Ddot{X}_i)-E_n[\kappa_{\Ddot{X}}(\cdot,\Ddot{X})]:i\in\{1,\dots,n\}\}=\{\hat{\kappa}_{\Ddot{X}}(\cdot,\Ddot{X}_i):i\in\{1,\dots,n\}\}.
\end{align*}
Similarly,  $ {\mathrm{ran}}(\hat{\Sigma}_{YY})$ and $ {\mathrm{ran}}(\hat{\Sigma}_{ZZ})$ are spanned by 
\begin{align*}
\mathcal{B}_{Y}=\{\hat{\kappa}_{Y}(\cdot,Y_i):i\in\{1,\dots,n\}\}, \quad \mathcal{B}_Z=\{\hat{\kappa}_Z(\cdot,Z_i):i\in\{1,\dots,n\}\},
\end{align*}
respectively. Let $K_X\in\mathbb{R}^{n\times n}$ be the Gram matrix of $\kappa_X$ acting on $X_1,\dots,X_n$; that is, $\left({K}_X\right)_{ij}=\kappa_{X}(X_i,X_j)$. Similarly, we define $K_Y$ and $K_Z$ as the Gram matrices for $\kappa_Y$ and $\kappa_Z$. 
Since $\kappa_{\Ddot{X}}=\kappa_X\kappa_Z$, the Gram matrix $K_{\Ddot{X}}$ of $\kappa_{\Ddot{X}}$ can be represented as
\begin{align}\label{eq-k-ddot}
    K_{\Ddot{X}} = K_X \odot K_Z,
\end{align}
where $\odot$ represents the Hadamard product between matrices; that is, for example, $(K_{\Ddot{X}})_{ij} = (K_X)_{ij} (K_Z)_{ij}$ for $i,j=1,\dots,n$. 
Furthermore, let $\tilde{K}_{\Ddot{X}} = HK_{\Ddot{X}}H$ be the centralized Gram matrix where $H=I \lo n -\frac{1}{n} 1_n1_n^T$, where $I \lo n$ is the $n \times n$ identity matrix, and $1 \lo n$ is the $n$-dimensional vector whose entries are 1. Equivalently, $(\tilde{K}_{\Ddot{X}})_{ij}=\langle \hat{\kappa}_{\Ddot X}(\cdot,\Ddot X_i), \hat{\kappa}_{\Ddot X}(\cdot,\Ddot X_j)\rangle_{\mathcal{G}_{\Ddot X}}$. Similarly, $\tilde{K}_{Y} = HK_{Y}H$ and $\tilde{K}_Z = HK_ZH$. We now give the coordinate mappings for some related operators as follows (in the notations described in, for example, \cite{li2018sufficient}, \cite{li2018nonparametric}):
\begin{align*}
    &\,_{\mathcal{B}_{\Ddot{X}}}[\hat{\Sigma}_{\Ddot{X}\Ddot{X}}]_{\mathcal{B}_{\Ddot{X}}}=n^{-1}\tilde{K}_{\Ddot{X}},\quad 
    \,_{\mathcal{B}_{Y}}[\hat{\Sigma}_{YY}]_{\mathcal{B}_{Y}}=n^{-1}\tilde{K}_{Y}, \quad
    \,_{\mathcal{B}_{Z}}[\hat{\Sigma}_{ZZ}]_{\mathcal{B}_{ZZ}}=n^{-1}\tilde{K}_{Z}, \\
    &\,_{\mathcal{B}_{\Ddot{X}}} [\hat{\Sigma}_{\Ddot{X}Y}]_{\mathcal{B}_{Y}}=n^{-1}\tilde{K}_{Y},\quad 
    \,_{\mathcal{B}_{Y}} [\hat{\Sigma}_{Y\Ddot{X}}]_{\mathcal{B}_{\Ddot{X}}}=n^{-1}\tilde{K}_{\Ddot{X}},\quad 
    \,_{\mathcal{B}_{\Ddot{X}}} [\hat{\Sigma}_{\Ddot{X}Z}]_{\mathcal{B}_{Z}}=n^{-1}\tilde{K}_{Z}, \\ 
    &\,_{\mathcal{B}_{Z}}[\hat{\Sigma}_{Z\Ddot{X}}]_{\mathcal{B}_{\Ddot{X}}}=n^{-1}\tilde{K}_{\Ddot{X}},\quad 
    \,_{\mathcal{B}_{Y}} [\hat{\Sigma}_{YZ}]_{\mathcal{B}_{Z}}=n^{-1}\tilde{K}_{Z}, \quad
    \,_{\mathcal{B}_{Z}}[\hat{\Sigma}_{ZY}]_{\mathcal{B}_{Y}}=n^{-1}\tilde{K}_{Y}.
\end{align*}
Using the above results, and applying the rules for coordinate mapping described in \cite{li2018nonparametric}, we derive the coordinate mappings of $\hat{B}_{\Ddot{X}|Z}$ and $\hat{B}_{Y|Z}^*$ as follows: 
\begin{align*}
    \,_{\mathcal{B}_{\Ddot{X}}}[\hat{B}_{\Ddot{X}|Z}]_{\mathcal{B}_Z}
    &=\, _{\mathcal{B}_{\Ddot{X}}}[\hat{\Sigma}_{\Ddot{X}Z}]_{\mathcal{B}_Z}\, _{\mathcal{B}_Z}[(\hat{\Sigma}_{ZZ}+\epsilon_nI)^{-1}]_{\mathcal{B}_Z}=
n^{-1}\tilde{K}_Z(n^{-1}\tilde{K}_Z+\epsilon_nI)^{-1}, \\
    \,_{\mathcal{B}_Z}[\hat{B}_{Y|Z}^*]_{\mathcal{B}_{Y}}
    &=\, _{\mathcal{B}_Z}[(\hat{\Sigma}_{ZZ}+\epsilon_nI)^{-1}]_{\mathcal{B}_Z}\, _{\mathcal{B}_Z}[\hat{\Sigma}_{ZY}]_{\mathcal{B}_{Y}}=(n^{-1}\tilde{K}_Z+\epsilon_nI)^{-1}n^{-1}\tilde{K}_{Y}.
\end{align*}
Therefore, the coordinate mapping of $\hat{\Sigma}_{\Ddot{X}Y|Z}$ is
\begin{align*}
\begin{split}
\,_{\mathcal{B}_{\Ddot{X}}}[\hat{\Sigma}_{\Ddot{X}Y|Z}]_{\mathcal{B}_{Y}}=&\, _{\mathcal{B}_{\Ddot{X}}}[\hat{\Sigma}_{\Ddot{X}Y}-\hat{B}_{\Ddot{X}|Z}\hat{\Sigma}_{ZY}-\hat{\Sigma}_{\Ddot{X}Z}\hat{B}_{Y|Z}^*+\hat{B}_{\Ddot{X}|Z}\hat{\Sigma}_{ZZ}\hat{B}_{Y|Z}^*]_{\mathcal{B}_{Y}}\\
=&\,_{\mathcal{B}_{\Ddot{X}}}[\hat{\Sigma}_{\Ddot{X}Y}]_{\mathcal{B}_{Y}}
-\,_{\mathcal{B}_{\Ddot{X}}}[\hat{B}_{\Ddot{X}|Z}]_{\mathcal{B}_Z}\,_{\mathcal{B}_Z}[\hat{\Sigma}_{ZY}]_{\mathcal{B}_{Y}}
-\,_{\mathcal{B}_{\Ddot{X}}}[\hat{\Sigma}_{\Ddot{X}Z}]_{\mathcal{B}_Z}\,_{\mathcal{B}_Z}[\hat{B}_{Y|Z}^*]_{\mathcal{B}_{Y}}\\
&\quad\quad+\,_{\mathcal{B}_{\Ddot{X}}}[\hat{B}_{\Ddot{X}|Z}]_{\mathcal{B}_Z}\,_{\mathcal{B}_Z}[\hat{\Sigma}_{ZZ}]_{\mathcal{B}_Z}\,_{\mathcal{B}_Z}[\hat{B}_{Y|Z}^*]_{\mathcal{B}_{Y}}\\
=&n^{-1}\tilde{K}_{Y}
-\tilde{K}_Z(\tilde{K}_Z+n\epsilon_nI)^{-1}n^{-1}\tilde{K}_{Y}
-n^{-1}\tilde{K}_Z(\tilde{K}_Z+n\epsilon_nI)^{-1}\tilde{K}_{Y}\\
&\quad\quad+\tilde{K}_Z(\tilde{K}_Z+n\epsilon_nI)^{-1}n^{-1}\tilde{K}_Z(\tilde{K}_Z+n\epsilon_nI)^{-1}\tilde{K}_{Y}\\
=&n^{-1}[I-\tilde{K}_Z(\tilde{K}_Z+n\epsilon_nI)^{-1}]^2\tilde{K}_{Y}.
\end{split}
\end{align*}

Next, we derive the expression for $\| \hat \Sigma \lo {\ddot X Y |Z} \| \lo {\mathrm{HS}} \hi 2 $. Let
\begin{align}
    R_Z=I-\tilde{K}_Z(\tilde{K}_Z+n\epsilon_nI)^{-1}=n\epsilon_n(\tilde{K}_Z+n\epsilon_nI)^{-1}.  \label{eq-rz}
\end{align}
Then $\,_{\mathcal{B}_{\Ddot{X}}}[\hat{\Sigma}_{\Ddot{X}Y|Z}]_{\mathcal{B}_{Y}}$ can be abbreviated as 
$n^{-1}R_Z^2\tilde{K}_{Y}$.
By Theorem 8 (part 4) of \cite{li2018nonparametric}, 
\begin{align*}
\begin{split}
\|\hat{\Sigma}_{\Ddot{X}Y|Z}\|_{\mathrm{HS}}^2
=&\left\|\tilde{K}_{\Ddot{X}}^{1/2}\,_{\mathcal{B}_{\Ddot{X}}}[\hat{\Sigma}_{\Ddot{X}Y|Z}]_{\mathcal{B}_{Y}}\tilde{K}_{Y}^{-1/2}\right\|_F^2
=\left\|n^{-1}\tilde{K}_{\Ddot{X}}^{1/2}R_Z^2\tilde{K}_{Y}^{1/2}\right\|_F^2\\
=&n^{-2}\tr \left(\tilde{K}_{\Ddot{X}}^{1/2}R_Z^2\tilde{K}_{Y}^{1/2} \tilde{K}_{Y}^{1/2}R_Z^2\tilde{K}_{\Ddot{X}}^{1/2}\right)
=n^{-2}\tr \left(R_Z\tilde{K}_{\Ddot{X}}R_Z  R_Z\tilde{K}_{Y}R_Z\right),
\end{split}
\end{align*}
{where $\tr$ denotes the trace of a matrix.} Let
\begin{align}
    \tilde{K}_{\Ddot{X}|Z}=R_Z\tilde{K}_{\Ddot{X}}R_Z, \quad \tilde{K}_{Y|Z}=R_Z\tilde{K}_{Y}R_Z. \label{eq-kxz-kyz}
\end{align} 
Then, $\|\hat{\Sigma}_{\Ddot{X}Y|Z}\|_{\mathrm{HS}}^2$ can be further abbreviated as 
\begin{align}\label{eq-sigma-hat-hs}
\|\hat{\Sigma}_{\Ddot{X}Y|Z}\|_{\mathrm{HS}}^2=n^{-2}\tr (\tilde{K}_{\Ddot{X}|Z}\tilde{K}_{Y|Z}).
\end{align}
The test statistic is then expressed as $T \lo n =n^{-1}\tr (\tilde{K}_{\Ddot{X}|Z}\tilde{K}_{Y|Z})$, which coincides the form given in equation (13) of \cite{zhang2012kernelbased}.

\def\ka{\kappa}

\section{Asymptotic Distribution of $T \lo n$}\label{sec-asymp-ccco-cpco}

\subsection{Asymptotic properties under null distribution}
In this section we derive the asymptotic distribution of $T \lo n$. We make the following additional assumptions. 

\begin{assumption}\label{ass-moment2}
    The kernels $\kappa_X$ and $\kappa_Y$ are bounded, and the kernel $\kappa_Z$ satisfies $E[\kappa_Z^2(Z,Z)]<\infty$.
\end{assumption}

This assumption is satisfied, for example, by the Gaussian radial basis function (RBF) kernel and the Laplace kernel. 
{Before stating the next assumption, we define an intermediate operator
\begin{align*}
\Lambda_{\Ddot{X}Z} = E[\tilde{\kappa}_X(\cdot,X)\kappa_Z(\cdot,Z) \otimes \tilde{\kappa}_Z(\cdot,Z)].
\end{align*}
Note that this operator is different from $\Sigma_{\Ddot{X}Z}$, which can be represented by 
\begin{align*}
\Sigma_{\Ddot{X}Z} = E[\kappa_X(\cdot,X)\kappa_Z(\cdot,Z) \otimes \tilde{\kappa}_Z(\cdot,Z)].
\end{align*}
Clearly, we have
\begin{align}\label{eq-sigma-lambda}
\Sigma_{\Ddot{X}Z} = E\{ [\mu_X + \tilde\kappa_X(\cdot,X)]\kappa_Z(\cdot,Z) \otimes \tilde{\kappa}_Z(\cdot,Z)\}
= \mu_X \Sigma_{ZZ} + \Lambda_{\Ddot{X}Z},
\end{align}
where $\mu_X\Sigma_{ZZ}$ is a linear operator from $\mathcal G_Z$ to $\mathcal G_{\Ddot{X}}$ defined by 
\begin{align*}
   ( \mu \lo X \Sigma \lo {ZZ} ) f = \mu \lo X (\Sigma \lo {ZZ} f) \in \ca G \lo {\ddot X}. 
\end{align*}
The following smoothness assumption is imposed in terms of $\Lambda_{\Ddot{X}Z}$ and $\Sigma_{YZ}$.}
\begin{assumption}\label{ass-beta}
    There exists some $\beta>0$ such that {$\Lambda_{\Ddot{X}Z}=S_{\Ddot{X}Z}\Sigma_{ZZ}^{1+\beta}$} and $\Sigma_{YZ}=S_{YZ}\Sigma_{ZZ}^{1+\beta}$ for some bounded linear operators $S_{\Ddot{X}Z}: \mathcal{G}_Z \to \mathcal{G}_{\Ddot{X}}$ and $S_{YZ}: \mathcal{G}_Z \to \mathcal{G}_{Y}$.
\end{assumption}

Furthermore, let 
\begin{align}\label{eq-eigen}
    \Sigma_{ZZ} = \sum_{k=1}^\infty \gamma_k (\phi_k \otimes \phi_k), \quad
    {\Sigma_{XX} = \sum_{k=1}^\infty \eta_k (\chi_k \otimes \chi_k), \quad 
    \Sigma_{YY} = \sum_{k=1}^\infty \mu_k (\psi_k \otimes \psi_k)}
\end{align}
be the {spectral decompositions of $\Sigma_{ZZ}$, $\Sigma_{XX}$ and $\Sigma_{YY}$, where, for example, } $\gamma_1 \ge \gamma_2 \ge \cdots$ are the eigenvalues of $\Sigma_{ZZ}$, and $\phi_1, \phi \lo 2, \ldots$  are the corresponding eigenfunctions. 
The Karhunen-Lo\'eve {expansions of $\kappa_Z(\cdot,Z)$, $\ka \lo X (\cdot, X)$ and $\ka \lo Y (\cdot, Y)$  are}
\begin{align}
\begin{split}
&\kappa_Z(\cdot,Z) = \mu_Z + \sum_{k=1}^\infty {\sqrt{\gamma_k}}\xi_k\phi_k,  \quad \xi_k = \frac{\langle \kappa_Z(\cdot,Z) - \mu_Z,  \phi_k \rangle_{\mathcal{G}_Z}}{{\sqrt{\gamma_k}}},\label{eq-kl-expansion}\\
&{\kappa_X(\cdot,X) = \mu_X + \sum_{k=1}^\infty \sqrt{\eta_k}\omega_k\chi_k,  \quad \omega_k = \frac{\langle \kappa_X(\cdot,X) - \mu_X,  \chi_k \rangle_{\mathcal{G}_X}}{\sqrt{\eta_k}},} \\
&{\kappa_Y(\cdot,Y) = \mu_Y + \sum_{k=1}^\infty \sqrt{\mu_k}\theta_k\psi_k,  \quad \theta_k = \frac{\langle \kappa_Y(\cdot,Y) - \mu_Y,  \psi_k \rangle_{\mathcal{G}_Y}}{\sqrt{\mu_k}}.} 
\end{split}
\end{align}
Note that $\xi \lo k$'s are uncorrelated random variables with mean 0 and variance {1}, {and same for $\omega_k$'s and $\theta_k$'s}.
We {make} an additional assumption on the decay rate of eigenvalues {of $\Sigma_{ZZ}$} in \eqref{eq-eigen}. {Henceforth, for two sequences of positive real numbers $(a \lo j)$ and $(b \lo j)$, we write $a \lo j \preceq b \lo j$ if there exists a constant $C$ such that $a \lo j \le Cb \lo j$ for all $j$, and write $a \lo j \asymp b \lo j$ if $a \lo j \preceq b \lo j$ and $b\lo j \preceq a \lo j$.}

\begin{assumption}\label{ass-alpha}
    {$\gamma_j \preceq j^{-\alpha}$} for some $\alpha>1$, and for all $j\in\{1,2,\dots\}$.
\end{assumption}

{We now give some equivalent conditions and concrete examples for Assumptions \ref{ass-beta} and \ref{ass-alpha}. }

{
\begin{proposition}\label{proposition:equivalence}
{Suppose $\beta$ is any positive constant,} and $\mathrm{rank}(\Sigma_{ZZ}) = K$ for some $K \in \mathbb{N} \cup \{\infty\}$.
\begin{enumerate}
\item[(i)] $\Sigma_{YZ}=S_{YZ}\Sigma_{ZZ}^{1+\beta}$ holds for some bounded linear operator $S_{YZ}: \mathcal{G}_Z \to \mathcal{G}_{Y}$ if and only if 
\begin{align}\label{eq-yz-beta}
\sup_{j,k \in \mathbb{N}, {k \le K}} \mu_j^{1/2} \gamma_k^{-1/2-\beta} |E(\theta_j \xi_k)| < \infty.
\end{align}
\item[(ii)] $\Lambda_{\Ddot{X}Z}=S_{\Ddot{X}Z}\Sigma_{ZZ}^{1+\beta}$ holds for some bounded linear operator $S_{\Ddot{X}Z}: \mathcal{G}_Z \to \mathcal{G}_{\Ddot{X}}$ if and only if
\begin{align}\label{eq-xddotz-beta}
\sup_{j,k,l \in \mathbb{N}, {k,l \le K}} \eta_j^{1/2} \gamma_k^{1/2} \gamma_ l ^{-1/2-\beta} |E(\omega_j \xi_k \xi_ l )| < \infty,
\quad 
\sup_{j,k \in \mathbb{N}, {k \le K}} \eta_j^{1/2} \gamma_k^{-1/2-\beta} |E(\omega_j \xi_k)| < \infty.
\end{align}
\end{enumerate}
\end{proposition}
}

\def\eop{\hfill $\Box$}

{The proof of Proposition \ref{proposition:equivalence} is given in Appendix \ref{proposition:equivalence-proof}. The equivalences in Proposition \ref{proposition:equivalence} show that  Assumption \ref{ass-beta} in effect imposes a certain degree of smoothness between   $\kappa_Z(\cdot,Z)$ and $\kappa_X(\cdot,X)$ and between $\kappa_Z(\cdot,Z)$ and $\kappa_Y(\cdot,Y)$. That is,
the dependence between them  occurs mostly in the low-frequency parts of the spectra, whereas high-frequency dependence is   ignorable. This point was also discussed  in \cite{li2018sufficient} and \cite{li2017linear}. To provide more insights into the nature of Assumption \ref{ass-beta}, we  give  two illustrative examples where it is  satisfied in Appendix \ref{app-example}.}

{In particular, the $\Sigma_{YZ}$ part of Assumption \ref{ass-beta} is commonly used in RKHS regression or dimension reduction literature, such as \cite{caponnetto2007optimal,li2017nonlinear}. However, the similar condition is only imposed on $\Lambda_{\Ddot{X}Z}$ rather than $\Sigma_{\Ddot{X}Z}$ because their difference term, given by \eqref{eq-sigma-lambda}, $\mu_X\Sigma_{ZZ}$ cannot be as smooth as $\Sigma_{ZZ}^{1+\beta}$ unless the eigenvalues of $\Sigma_{ZZ}$ do not decay, which contradicts Assumption \ref{ass-alpha}. Therefore, when deriving the asymptotic properties, the $\mu_X\Sigma_{ZZ}$ term will be handled separately.}

{Assumption \ref{ass-alpha} regulates the behavior of the distribution of the random function $Z$. Instead of assuming $\gamma_j \asymp j^{-\alpha}$, which is common in functional data analysis \citep{cai2006prediction,sang2026nonlinear}, our Assumption \ref{ass-alpha} only requires a weaker condition that $\gamma_j$ can be either faster than or of the same order as $j^{-\alpha}$. Based on our Assumption \ref{ass-alpha}, we prove a similar lemma as Lemma 8 of \cite{sang2026nonlinear} as follows.}

{
\begin{lemma}\label{lem-alpha}
Under Assumption \ref{ass-alpha}, if $\epsilon_n \to 0$, then we have 
\begin{align*}
\sum_{j=1}^{\infty}
\frac{\gamma_j}{(\gamma_j+\epsilon_n)^2}
=
O\left(\epsilon_n^{-(\alpha+1)/\alpha}\right),
\qquad 
{\sum_{j=1}^{\infty}
\frac{\gamma_j^2}{(\gamma_j+\epsilon_n)^2}
=
O\left(\epsilon_n^{-1/\alpha}\right)}.
\end{align*}
\end{lemma}}
{The proof of Lemma \ref{lem-alpha} is presented in Appendix \ref{lem-alpha-proof}. The next Proposition gives a concrete example to illustrate Assumption~\ref{ass-alpha}, and its proof is given in Appendix \ref{prop-brownian-proof}.}
{
\begin{proposition}[Brownian bridge with Gaussian kernel]\label{prop-brownian}
Let $\mathcal H_Z = L^2[0,1]$, $Z$ be a Brownian bridge on $[0,1]$, and $\kappa_Z$ be the Gaussian kernel
\begin{align*}
\kappa_Z(z,z')
=
\exp\{-\eta\|z-z'\|_{L^2}^2\},
\qquad  z,z' \in \mathcal H_Z.
\end{align*}
If \ $0<\eta\le \pi^2/8$, then Assumption \ref{ass-alpha}  is satisfied for  every $\alpha \in (1,2)$. {Furthermore, if $Z$ is replaced by any Gaussian process whose eigenvalues of covariance operator, $\rho_1\ge\rho_2\ge\cdots$ decays at the rate $\rho_m \preceq m^{-r}$ for some $r>1$, then Assumption \ref{ass-alpha} is satisfied for every $\alpha \in (1,r)$.}
\end{proposition}}

{
We introduce the following operator, which will contribute to the bias term in $\hat{B}_{\Ddot{X}|Z} - B_{\Ddot{X}|Z}$:
\begin{align}\label{eq-r-def}
R 
= \mu_X\hat{\Sigma}_{ZZ}(\hat{\Sigma}_{ZZ}+\epsilon_nI)^{-1} - \mu_X.
\end{align}
Similar to the notation used in  (\ref{eq-sigma-lambda}), here,   $\mu_X\hat{\Sigma}_{ZZ}(\hat{\Sigma}_{ZZ}+\epsilon_nI)^{-1} $ stands for the operator from $\ca G \lo Z$ to $\ca G \lo {\ddot X}$ defined by 
\begin{align*}
 [ \mu_X\hat{\Sigma}_{ZZ}(\hat{\Sigma}_{ZZ}+\epsilon_nI)^{-1}  ] f =  \mu_X [\hat{\Sigma}_{ZZ}(\hat{\Sigma}_{ZZ}+\epsilon_nI)^{-1}  f ]. 
\end{align*}
Henceforth, we use $\|\cdot\|_{\mathrm{OP}}$ to denote the operator norm.
The next lemma develops some properties of   $R$ which will be used later, and its proof is given in Appendix \ref{lem-R-rate-proof}. 
}

{
\begin{lemma}\label{lem-R-rate}
Under Assumption \ref{ass-moment2}, we have
\begin{enumerate}
\item[(i)] $\|R\|_{\mathrm{OP}} \le \|\mu_X\|_{\mathcal G_X} < \infty$, i.e., $R$ is bounded;
\item[(ii)] for any $\theta>0$, we have
\begin{align}\label{eq-r-sigma-theta}
\|R \Sigma_{ZZ}^\theta\|_{\mathrm{OP}}
= O_P(n^{-1/2}\epsilon_n^{\theta\wedge 1 -1} + \epsilon_n^{\theta\wedge 1});
\end{align}
\item[(iii)] {in particular,
\begin{align}\label{eq-r-sigma-hs}
\|R \Sigma_{ZZ}\|_{\mathrm{HS}}= O_P(n^{-1/2} + \epsilon_n^{1-1/(2\alpha)}).
\end{align}}
\end{enumerate}
\end{lemma}}

We {next} state a lemma regarding the convergence {rates  of} $\hat{B}_{\Ddot{X}|Z}-B_{\Ddot{X}|Z}{-R}$ and $\hat{B}_{Y|Z}-B_{Y|Z}$, which {were} established in a recent work by \cite{choi2026sharpened}.

\begin{lemma}\label{lemma-regul-conv}
{Suppose} Assumptions \ref{ass-cond-exp-ddot}--\ref{ass-alpha} {are satisfied} for some $\beta>0$ and $\alpha>1$. If $\epsilon_n\to 0$ as $n \to \infty$, then
\begin{align*}
    \|\hat{B}_{\Ddot{X}|Z}-B_{\Ddot{X}|Z}{-R}\|_{\mathrm{OP}}=  O_P(\epsilon_n^{\beta\wedge 1} + n^{-1/2} \epsilon_n^{\beta\wedge 1 -1} + n^{-1} \epsilon_n^{-(3\alpha+1)/(2\alpha)} + n^{-1/2} \epsilon_n^{-(\alpha+1)/(2\alpha)})
\end{align*}
and
\begin{align*}
    \|\hat{B}_{Y|Z}-B_{Y|Z}\|_{\mathrm{OP}}=O_P(\epsilon_n^{\beta\wedge 1} + n^{-1/2} \epsilon_n^{\beta\wedge 1 -1} + n^{-1} \epsilon_n^{-(3\alpha+1)/(2\alpha)} + n^{-1/2} \epsilon_n^{-(\alpha+1)/(2\alpha)}).
\end{align*}
\end{lemma}

The proof of Lemma \ref{lemma-regul-conv} mimics that of Theorem 3.4 of \cite{choi2026sharpened} and Theorem 9 of \cite{sang2026nonlinear}, and we place the full proof in our setting in Appendix~\ref{sec-proof} for completeness. 

Furthermore, Theorem 3.7 of \cite{choi2026sharpened} provides the optimal choice of $\epsilon_n$, which yields the convergence rates for
$\|\hat{B} \lo {\Ddot{X}| Z} - B \lo {\Ddot{X}| Z}{-R}\|\lo {\mathrm{OP}}$ and
$\|\hat{B} \lo {Y| Z} - B_{Y| Z}\|_{\mathrm{OP}}$ {under the optimal choice of the tuning parameter}.
We restate that result below in the form of a corollary tailored to our setting.

\begin{corollary}\label{cor-opt-eps}
Suppose the conditions in Lemma \ref{lemma-regul-conv} hold and $\epsilon_n\asymp n^{-\delta}$ for some $\delta>0$. 
\begin{enumerate}
    \item[(i)] If $\beta > \frac{\alpha-1}{2\alpha}$, then the optimal choice of $\delta$ is $\frac{\alpha}{2\alpha(\beta\wedge 1)+\alpha+1}$, and the convergence rates of $\|\hat{B}_{\Ddot{X}|Z}-B_{\Ddot{X}|Z}{-R}\|_{\mathrm{OP}}$ and $\|\hat{B}_{Y|Z}-B_{Y|Z}\|_{\mathrm{OP}}$ {under the optimal choice of the tuning parameter} are $O_P(n^{-\frac{\alpha(\beta\wedge 1)}{2\alpha(\beta\wedge 1)+\alpha+1}})$;
    \item[(ii)] If $\beta \le \frac{\alpha-1}{2\alpha}$, then the optimal choice of $\delta$ is $\frac{1}{2}$, and the convergence rates of $\|\hat{B}_{\Ddot{X}|Z}-B_{\Ddot{X}|Z}{-R}\|_{\mathrm{OP}}$ and $\|\hat{B}_{Y|Z}-B_{Y|Z}\|_{\mathrm{OP}}$ {under the optimal choice of the tuning parameter} are $O_P(n^{-\frac{\beta}{2}})$.
\end{enumerate}    
\end{corollary}

We now state the asymptotic normality of $\hat{\Sigma}_{\Ddot{X}Y|Z}$ based on the {convergence rate} of $\hat{B}_{\Ddot{X}|Z}$ and $\hat{B}_{Y|Z}$ {under the optimal choice of the tuning parameter} in Corollary \ref{cor-opt-eps}.

\begin{theorem}\label{thm-clt-sigmahat}
Suppose that Assumptions \ref{ass-cond-exp}--\ref{ass-alpha} hold for some $\beta>0$ and {$\alpha>2$}. Further suppose that $\beta > \frac{\alpha-1}{2\alpha}$ and $\frac{\alpha(\beta\wedge 1)}{2\alpha(\beta\wedge 1)+\alpha+1} > \frac{1}{4}$. Let $\epsilon_n \asymp n^{-\frac{\alpha(\beta\wedge 1)}{2\alpha(\beta\wedge 1)+\alpha+1}}$. Then, we have
\begin{align}
    \sqrt{n}(\hat{\Sigma}_{\Ddot{X}Y|Z}-\Sigma_{\Ddot{X}Y|Z})\xrightarrow{\mathcal{D}}N(0,\Gamma_{\Ddot{X}Y|Z}),\label{eq-asymp-dist-ccco}
\end{align}
where $\Gamma_{\Ddot{X}Y|Z} : \mathcal{G}_{\Ddot{X}}\otimes\mathcal{G}_{Y} \to \mathcal{G}_{\Ddot{X}}\otimes\mathcal{G}_{Y}$ is the linear operator defined by
\begin{align}
\begin{split}
    \Gamma_{\Ddot{X}Y|Z}=&E\left\{[\tilde{\kappa}_{\Ddot{X}|Z}(\cdot,\Ddot{X}|Z)\otimes\tilde{\kappa}_{Y|Z}(\cdot,Y|Z)-\Sigma_{\Ddot{X}Y|Z}]\otimes
    [\tilde{\kappa}_{\Ddot{X}|Z}(\cdot,\Ddot{X}|Z)\otimes\tilde{\kappa}_{Y|Z}(\cdot,Y|Z)-\Sigma_{\Ddot{X}Y|Z}]\right\}.\label{eq-asymp-var-ccco}
\end{split}
\end{align}
\end{theorem}

{A key step in the proof of Theorem \ref{thm-clt-sigmahat} is to show that the three terms in 
\begin{align*} 
&\sqrt{n}(B_{\Ddot{X}|Z}-\hat{B}_{\Ddot{X}|Z})E_n[\tilde{\kappa}_Z(\cdot,Z)\otimes\epsilon_{Y}]+\sqrt{n}E_n[\epsilon_{\Ddot{X}}\otimes\tilde{\kappa}_Z(\cdot,Z)](B_{Y|Z}-\hat{B}_{Y|Z})^*\\
&\qquad +\sqrt{n}(B_{\Ddot{X}|Z}-\hat{B}_{\Ddot{X}|Z})E_n[\tilde{\kappa}_Z(\cdot,Z)\otimes\tilde{\kappa}_Z(\cdot,Z)](B_{Y|Z}-\hat{B}_{Y|Z})^*
\end{align*} 
are all $o_P(1)$, which, as mentioned in the Introduction, is ensured by applying Corollary~\ref{cor-opt-eps} and Lemma~\ref{lem-R-rate}. Without these conditions --- that
is, if we only use the optimal rate implied by \eqref{eq:epsilon_n_beta_wedge} --- the above quantity
would have the same order of magnitude as the leading term
\begin{align*}
\sqrt{n} \left( E_n\left\{[\tilde{\kappa}_{\Ddot{X}}(\cdot,\Ddot{X})-B_{\Ddot{X}|Z}\tilde{\kappa}_Z(\cdot,Z)]\otimes[\tilde{\kappa}_{Y}(\cdot,Y)-B_{Y|Z}\tilde{\kappa}_Z(\cdot,Z)]\right\} - \Sigma_{\Ddot{X}Y|Z} \right),
\end{align*}
and consequently the asymptotic distribution in \eqref{eq-asymp-dist-ccco} and \eqref{eq-asymp-var-ccco} would fail to hold. The full proof of Theorem \ref{thm-clt-sigmahat} is placed in Appendix \ref{thm-clt-sigmahat-proof}.
}

By Theorem \ref{thm-clt-sigmahat} and continuous mapping theorem, we have the following asymptotic distribution of our test statistic.

\begin{corollary}\label{corollary:weighted chisquare}
Suppose that $\Gamma_{\Ddot{X}Y|Z}$ defined in \eqref{eq-asymp-var-ccco} is trace class. Under all assumptions in Theorem \ref{thm-clt-sigmahat}, when $\Sigma_{\Ddot{X}Y|Z}=0$, we have
\begin{align}
n\|\hat{\Sigma}_{\Ddot{X}Y|Z}\|_{\mathrm{HS}}^2\xrightarrow{\mathcal{D}}\sum_{k=1}^\infty\lambda_kZ_k^2,\label{eq-asymp-ccco-chisq}
\end{align}
where $\lambda_1\geq\lambda_2\geq\cdots$ are eigenvalues of $\Gamma_{\Ddot{X}Y|Z}$ in \eqref{eq-asymp-var-ccco} without the $\Sigma_{\Ddot{X}Y|Z}$ term, and $Z_k$'s are independent standard normal random variables.
\end{corollary}

\subsection{Welch-Satterthwaite approximation}\label{sec-ws-approx}

{Following the idea of \cite{kokoszka2017inference} and \cite{zhang2014analysis}, we use the Welch-Satterwaite approximation to the asymptotic distribution in \eqref{eq-asymp-ccco-chisq}. 
Let $V = \sum_{i=1}^\infty \lambda_i Z_i^2$, where $\lambda_i$'s and $Z_i$'s are as defined in Corollary \ref{corollary:weighted chisquare}.
Our goal is to approximate the distribution of $V$ using 
\begin{align}\label{eq-sw-approx}
\tilde{V} \sim \zeta \chi^2(\nu),
\end{align} 
where $\zeta$ and $\nu$ are calibrated to capture the first two moments of $V$. 
Since $Z_1,Z_2,\dots$ are i.i.d. $N(0,1)$, we have
\begin{align*}
\mu_V = E(V) = \sum_{i=1}^\infty \lambda_i E(Z_i^2) = \sum_{i=1}^\infty \lambda_i = \tr (\Gamma_{\Ddot{X}Y|Z})
\end{align*}
and
\begin{align*}
\sigma_V^2 = \mathrm{var}(V) = \sum_{i=1}^\infty \lambda_i^2 \mathrm{var}(Z_i^2) 
= 2 \sum_{i=1}^\infty \lambda_i^2 
= \tr (\Gamma_{\Ddot{X}Y|Z}^2).
\end{align*}
By matching the first two moments of $\zeta \chi (\nu)$ and $V$, as was done in  equation (14) of \cite{kokoszka2017inference}, we have 
\begin{align}\label{eq-zeta-nu-pop}
\zeta = \frac{\sigma_V^2}{2\mu_V} 
= \frac{\tr (\Gamma_{\Ddot{X}Y|Z}^2)}{2\,\tr (\Gamma_{\Ddot{X}Y|Z})}, \qquad
\nu = \frac{2\mu_V^2}{\sigma_V^2}
= \frac{2\, [\tr (\Gamma_{\Ddot{X}Y|Z})]^2}{\tr (\Gamma_{\Ddot{X}Y|Z}^2)}.
\end{align}}

\subsection{{Local power analysis}}\label{sec-local-power}

{In this subsection, we derive the asymptotic distribution under the local alternative distribution using the technique in Theorem 3 of \cite{tang2024nonparametric} and Theorem 13 of \cite{gretton2012kernel}, with the detailed proof placed in Appendix \ref{thm:local power-proof}.}

{\begin{theorem} \label{thm:local power}
Suppose all assumptions in Theorem \ref{thm-clt-sigmahat} are satisfied and suppose 
\begin{enumerate}
\item[(i)] $\Gamma \lo {\ddot X Y |Z}$ has spectral decomposition $\sum \lo {j=1} \hi \infty \lambda \lo j (v \lo j \otimes v \lo j)$, where $\lambda \lo 1 \ge \lambda \lo 2 \ge \dots$ are eigenvalues of $\Gamma \lo {\ddot X Y |Z}$, and $v \lo 1, v \lo 2, \ldots$ are eigenfunctions of $\Gamma \lo {\ddot X Y |Z}$, which form an orthonormal basis in $\mathcal G \lo {\Ddot X} \otimes \mathcal G \lo {Y}$;
\item[(ii)] $\Sigma \lo 1$ is a fixed linear operator in $\mathcal G \lo {\Ddot X} \otimes \mathcal G \lo {Y}$ with expansion $\sum \lo {j=1} \hi \infty \sigma \lo j  v \lo j$ and $\| \Sigma \lo 1 \| \lo {\mathrm{HS}} = c > 0$.
\end{enumerate}
Then, under the local alternative hypothesis
$
H \lo 1 \hi {(n)}: \Sigma \lo {\Ddot X Y | Z} = n \hi {-1/2} \Sigma \lo 1,
$
the asymptotic distribution of $\|\hat{\Sigma} \lo {\Ddot X Y | Z}  \| \lo {\mathrm{HS}}$ is
\begin{align*}
n\|\hat{\Sigma} \lo {\Ddot X Y | Z}  \| \lo {\mathrm{HS}}  \hi 2  \xrightarrow{\mathcal{D}} \sum \lo {j=1} \hi {\infty} \lambda \lo j \tilde{Z} \lo j \hi 2
\end{align*}
where $\tilde{Z} \lo j$ are independent $N(\sigma \lo j / \sqrt{\lambda \lo j}, 1) $ random variables. 
Thus, the local power of the test is
\begin{align*}
P \left( n\|\hat{\Sigma} \lo {\Ddot X Y | Z} \| \lo {\mathrm{HS}}  \hi 2 > s \right) \to P \left( \sum \lo {j=1} \hi {\infty} \lambda \lo j \tilde{Z} \lo j \hi 2 > s  \right).
\end{align*}
\end{theorem}}

\section{Approximation for Random Functions}\label{sec-ran-func}

\subsection{Approximating trajectories by {reproducing-kernel-based} smoother}

In practice, the processes $X(t)$, $Y(t)$, and $Z(t)$ are not observed 
continuously over $T$; instead, they are observed only at a finite 
collection of discrete time points. The underlying curves must 
therefore be reconstructed from these discrete observations. 
In the following, we describe the procedure for estimating these 
curves, using $X(t)$ as an illustration.  Let $X_1,\dots,X_n$ be i.i.d. samples from   $X$. Assume, for sample $i$,   we can only observe $X_i(t)$ at time points in $T_i=\{t_{i1},\dots,t_{im_i}\}\subseteq T$.  We use an RKHS smoother to approximate the functions (see, for example, Section~6.2 of \cite{li2017nonlinear} or Section~7.2 of \cite{kokoszka2017introduction}). Let $\kappa_T:T\times T\to\mathbb{R}$ be a positive definite kernel, and let $\mathcal{H}$ be a Hilbert space spanned by $\mathcal{L}=\{\kappa_T(\cdot,t):t \in T\}$.

For subject $i$, we use the subset 
$\mathcal{L}_i=\{\kappa_T(\cdot,t):t\in T_i\}$ to construct an approximation 
of the function $X_i(t)$, and denote the resulting approximation by 
$\hat{X}_i(t)$. Let $K_T^{(i)}$ be the kernel matrix generated by 
$\kappa_T$ evaluated on $T_i$, that is,
\begin{align}\label{eq-kt-i}
(K_T^{(i)})_{rs}=\kappa_T(t_{ir},t_{is}), \quad r,s\in\{1,\dots,m_i\}.
\end{align}
Furthermore, let 
$
X_i(T_i)=\{X_i(t_{ir}): r\in\{1,\dots,m_i\}\}
$
denote the observed values of $X_i(t)$ at the time points in $T_i$.

We now represent $X_i(t)$ in the coordinate system induced by 
$\mathcal{L}_i$, and denote the corresponding coordinate vector 
by $[X_i]_{\mathcal{L}_i}$. By the definition of the coordinate mapping, 
for $s\in\{1,\dots,m_i\}$,
\[
X_i(t_{is})=\sum_{r=1}^{m_i}([X_i]_{\mathcal{L}_i})_r\,
\kappa_T(t_{is},t_{ir}).
\]
Stacking the observations $X_i(t_{is})$ together yields
\begin{align*}
X_i(T_i)=K_T^{(i)}[X_i]_{\mathcal{L}_i}.
\end{align*}
To solve for $[X_i]_{\mathcal{L}_i}$, we account for the possible 
singularity of $K_T^{(i)}$ by applying Tikhonov regularization. 
The regularized coordinate mapping is given by
\begin{align}\label{eq-xhat-coord}
[\hat{X}_i]_{\mathcal{L}_i}=(K_T^{(i)}+\delta_n I)^{-1}X_i(T_i),
\end{align}
where $\delta_n>0$ is a regularization tuning parameter. 
We use the notation $[\hat{X}_i]_{\mathcal{L}_i}$ instead of 
$[X_i]_{\mathcal{L}_i}$ to emphasize that this is the coordinate of the estimated curve $\hat X \lo i$. 
Let 
\[
\kappa_T(\cdot,T_i)=\{\kappa_T(\cdot,t_{ir}):r\in\{1,\dots,m_i\}\}.
\]
The coordinate vector $[\hat{X}_i]_{\mathcal{L}_i}$ then yields the 
following representation of the reconstructed curve:
\begin{align*}
\hat{X}_i(\cdot)=[\hat{X}_i]_{\mathcal{L}_i}^{\mathsf T}\kappa_T(\cdot,T_i).
\end{align*}

\subsection{Approximating  inner products and distances}\label{sec-approx-inner-prod}

To construct the second layer kernel  $\ka \lo X$ (again, using $X$ as an illustration), we need to approximate the inner products $\langle X \lo i, X \lo j \rangle$ and the distances $\| X \lo i - X \lo j \|$, both taken in the space $L \hi 2 (T)$. 
To estimate the inner product
$ 
\langle X_i, X_j \rangle = \int_T X_i(t) X_j(t)\,dt,  
 $
we approximate $X_i$ and $X_j$
by their reconstructions $\hat{X}_i$ and $\hat{X}_j$, respectively, and
evaluate the integral numerically.
Specifically, we evaluate $\hat{X}_i$ and $\hat{X}_j$ at $\ell+1$
equally-spaced points in $T$, denoted by
$
U=\{u_0,u_1,\dots,u_\ell\}.
$
For simplicity, we assume that $\ell$ is even and that $u_0$ and $u_\ell$
are the left and right endpoints of $T$.
Applying Simpson's rule to approximate the integral, we obtain
\begin{align}\label{eq-simpson}
\int_{u_0}^{u_\ell}\hat{X}_i(t)\hat{X}_j(t)\,dt
\approx
\hat{X}_i(U)^{\mathsf T}D\hat{X}_j(U),
\end{align}
where $\hat{X}_i(U)$ is the vector $\{\hat{X}_i(u_r):r\in\{0,1,\dots,\ell\}\}$ and same for $\hat{X}_j(U)$, 
\begin{align*}
    D=(h/3)\mathrm{diag}(1,4,2,4,\dots,2,4,1),
\end{align*} and $h=(u_\ell-u_0)/\ell$. Using the coordinate mappings, we have
\begin{align}\label{eq-unbalanced}
    \langle X_i,X_j \rangle \approx  \langle \hat{X}_i,\hat{X}_j \rangle =\int_{u_0}^{u_\ell}\hat{X}_i(t)\hat{X}_j(t)dt\approx[\hat{X}_i]_{\mathcal{L}_i}^{\mathsf{T}}\kappa_T(T_i,U) \, D \, \kappa_T(U,T_j)[\hat{X}_j]_{\mathcal{L}_j}
\end{align}
where $\kappa_T(T_i,U)$ is the $m \lo i \times \ell$  matrix $\{\kappa_T(t_{ir},u_s):r\in\{1,\dots,m_i\},s\in\{0,\dots,\ell\}\}$, and $\kappa_T(U,T_j)$ is the $\ell$ by $m \lo j$ matrix $\{\kappa_T(u_r,t_{js}):r\in\{0,\dots,\ell\},s\in\{1,\dots,m_j\}\}$. 
The distance $\| X \lo i - X \lo j \| $ can then be estimated by 
\begin{align}\label{eq:norm hat X i}
    \| \hat X \lo i - \hat X \lo j \|  = 
    \sqrt{\langle \hat X \lo i , \hat X \lo i \rangle - 2 \langle \hat X \lo i, \hat X \lo j \rangle + \langle \hat X\lo j , \hat X \lo j \rangle}, 
\end{align}
where the inner products on the right-hand side are given the right-hand side of \eqref{eq-unbalanced}.

\subsection{Simplification under balanced observation schedule}\label{sec-simplification}

An important special case of the observation schedule is where all subjects are observed according to the same set of equally-spaced time points, which is known as the  balanced observation scheme. In symbols,   $T_1=\dots=T_n=U$, where $U$ is the common set of time points. In this case we can simplify the estimation procedure by directly applying Simpson's rule \eqref{eq-simpson} to $X_i(U)$ and $X_j(U)$ as an approximation of $\langle X_i, X_j \rangle$. That is,
\begin{align}\label{eq-balanced}
    \langle X_i,X_j \rangle  = \int_{u_0}^{u_l} X_i(t) X_j(t) dt \approx X_i(U)^{\mathsf{T}} D X_j(U).
\end{align}
Based on the inner product, we can calculate the inner product matrix and the pairwise distance matrix needed for constructing the second layer kernel $\ka \lo X$ (as well as $\ka \lo Y$ and $\ka \lo Z$),   Note that we do not need to involve $\kappa_T$ and the tuning parameter, and the calculation is significantly simplified.

\section{Sample-Level Implementation}\label{sec-implementation}

\subsection{Estimation of eigenvalues}

The linear operator $\Gamma_{\Ddot{X}Y|Z}$ defined by \eqref{eq-asymp-var-ccco} depends on the unknown population distribution of $(X,Y,Z)$. So, to apply Theorem \ref{thm-clt-sigmahat} and Corollary \ref{corollary:weighted chisquare} in hypothesis testing of conditional independence we must first approximate this linear operator. 

Since $\Gamma_{\Ddot{X}Y|Z}$ is a self-adjoint linear operator  on $\ca G_{\Ddot{X}\otimes Y}=\mathcal{G}_{\Ddot{X}}\otimes \mathcal{G}_{Y}$, we begin by characterizing the tensor product space $\ca G \lo {\ddot X \otimes Y}$ at the sample level.  By Theorem 2.7.13 of \cite{hsing2015theoretical}, $\kappa_{\Ddot{X}}\kappa_{Y}$ 
is a reproducing kernel of $\mathcal{G}_{\Ddot{X}\otimes Y}$. 

To estimate the eigenvalues of $\Gamma_{\Ddot{X}Y|Z}$, we perform coordinate mapping on its empirical version $\hat{\Gamma}_{\Ddot{X}Y|Z}$ under the following basis: 
\begin{align*}
    \mathcal{B}_{\Ddot{X}Y}=\left\{ {\hat{\kappa}_{\Ddot{X}Y}( \cdot ,  \Ddot{X}_i, \star, Y_j) }:i\in\{1,\dots,n\};j\in\{1,\dots,n\}\right\},
\end{align*}
where 
$\hat{\kappa}_{\Ddot{X}Y}( \cdot,   \Ddot{x},\star,y)=\hat{\kappa}_{\Ddot{X}}(\cdot,\Ddot{x})\hat{\kappa}_{Y}(\star,y)$.

In the operator level, $\hat{\Gamma}_{\Ddot{X}Y|Z}$ is estimated by replacing all expectations in \eqref{eq-asymp-var-ccco} by the corresponding sample means, as follows:
\begin{align*}
    \hat{\Gamma}_{\Ddot{X}Y|Z}=\frac{1}{n}\sum_{k=1}^n[\hat{\kappa}_{\Ddot{X}|Z}(\cdot,\Ddot{X}_k|Z_k)\otimes\hat{\kappa}_{Y|Z}(\cdot,Y_k|Z_k)]\otimes[\hat{\kappa}_{\Ddot{X}|Z}(\cdot,\Ddot{X}_k|Z_k)\otimes\hat{\kappa}_{Y|Z}(\cdot,Y_k|Z_k)].
\end{align*}
We then turn the linear operators in the right-hand above into $n \times n$ matrices by coordinating mapping. Applying  the rules of coordinate mapping described in Lemma 12.3 of \cite{li2018sufficient}, the coordinate of the linear operator $\hat{\Gamma}_{\Ddot{X}Y|Z}$ with respect to the spanning system  $\mathcal{B}_{\Ddot{X}Y}$ is
\begin{align}\label{eq:ca B ddot}
\begin{split}
\,_{\mathcal{B}_{\Ddot{X}Y}}[\hat{\Gamma}_{\Ddot{X}Y|Z}]_{\mathcal{B}_{\Ddot{X}Y}}
=&\frac{1}{n}\sum_{k=1}^n\,_{\mathcal{B}_{\Ddot{X}Y}}[[\hat{\kappa}_{\Ddot{X}|Z}(\cdot,\Ddot{X}_k|Z_k)\otimes\hat{\kappa}_{Y|Z}(\cdot,Y_k|Z_k)]\otimes[\hat{\kappa}_{\Ddot{X}|Z}(\cdot,\Ddot{X}_k|Z_k)\otimes\hat{\kappa}_{Y|Z}(\cdot,Y_k|Z_k)]]_{\mathcal{B}_{\Ddot{X}Y}}\\
=&\frac{1}{n}\sum_{k=1}^n[\hat{\kappa}_{\Ddot{X}|Z}(\cdot,\Ddot{X}_k|Z_k)\otimes\hat{\kappa}_{Y|Z}(\cdot,Y_k|Z_k)]_{\mathcal{B}_{\Ddot{X}Y}}[\hat{\kappa}_{\Ddot{X}|Z}(\cdot,\Ddot{X}_k|Z_k)\otimes\hat{\kappa}_{Y|Z}(\cdot,Y_k|Z_k)]_{\mathcal{B}_{\Ddot{X}Y}}^\mathsf{T}G_{\mathcal{B}_{\Ddot{X}Y}},
\end{split}
\end{align}
where $G_{\mathcal{B}_{\Ddot{X}Y}}$ is the Gram matrix of the set $\mathcal{B}_{\Ddot{X}Y}$; that is, the $((i,j),(i',j'))$-th entry of $G_{\mathcal{B}_{\Ddot{X}Y}}$ is
\begin{align*}
(G_{\mathcal{B}_{\Ddot{X}Y}})_{(ij)(i'j')}=&\langle\hat{\kappa}_{\Ddot{X}}(\cdot,\Ddot{X}_i)\otimes\hat{\kappa}_{Y}(\cdot,Y_j),\hat{\kappa}_{\Ddot{X}}(\cdot,\Ddot{X}_{i'})\otimes\hat{\kappa}_{Y}(\cdot,Y_{j'})\rangle_{\mathcal{H}_{\Ddot{X}\otimes Y}}\\
=&\langle\hat{\kappa}_{\Ddot{X}}(\cdot,\Ddot{X}_i),\hat{\kappa}_{\Ddot{X}}(\cdot,\Ddot{X}_{i'})\rangle_{\mathcal{H}_{\Ddot{X}}}\langle\hat{\kappa}_{Y}(\cdot,Y_j),\hat{\kappa}_{Y}(\cdot,Y_{j'})\rangle_{\mathcal{H}_{Y}}\\
=&\hat{\kappa}_{\Ddot{X}}(\Ddot{X}_i,\Ddot{X}_{i'})\hat{\kappa}_{Y}(Y_j,Y_{j'})
=(\tilde{K}_{\Ddot{X}})_{ii'}(\tilde{K}_{Y})_{jj'}.
\end{align*}
Hence, we have
\begin{align*}
G_{\mathcal{B}_{\Ddot{X}Y}}=\tilde{K}_{\Ddot{X}}\otimes\tilde{K}_{Y},
\end{align*}
where $\otimes$ denotes the Kronecker product between two matrices. Also, on the right-hand side of (\ref{eq:ca B ddot}), 
\begin{align*}
   [\hat{\kappa}_{\Ddot{X}|Z}(\cdot,\Ddot{X}_k|Z_k)\otimes\hat{\kappa}_{Y|Z}(\cdot,Y_k|Z_k)]_{\mathcal{B}_{\Ddot{X}Y}}=[\hat{\kappa}_{\Ddot{X}|Z}(\cdot,\Ddot{X}_k|Z_k)]_{\mathcal{B}_{\Ddot{X}}}\otimes[\hat{\kappa}_{Y|Z}(\cdot,Y_k|Z_k)]_{\mathcal{B}_{Y}},
\end{align*}
where the first term is 
\begin{align*}
[\hat{\kappa}_{\Ddot{X}|Z}(\cdot,\Ddot{X}_k|Z_k)]_{\mathcal{B}_{\Ddot{X}}}
=[\hat{\kappa}_{\Ddot{X}}(\cdot,\Ddot{X}_k)-\hat{B}_{\Ddot{X}|Z}\hat{\kappa}_Z(\cdot,Z_k)]_{\mathcal{B}_{\Ddot{X}}}
=[\hat{\kappa}_{\Ddot{X}}(\cdot,\Ddot{X}_k)]_{\mathcal{B}_{\Ddot{X}}}-\,_{\mathcal{B}_{\Ddot{X}}}[\hat{B}_{\Ddot{X}|Z}]_{\mathcal{B}_{Z}}[\hat{\kappa}_Z(\cdot,Z_k)]_{\mathcal{B}_{Z}}\\
=e_k-\tilde{K}_Z(\tilde{K}_Z+n\epsilon_nI)^{-1}e_k
=[I-\tilde{K}_Z(\tilde{K}_Z+n\epsilon_nI)^{-1}]e_k
=n\epsilon_n(\tilde{K}_Z+n\epsilon_nI)^{-1}e_k
=R_Ze_k,
\end{align*}
and same arguments apply to the second term, which gives
\begin{align*}
    [\hat{\kappa}_{Y|Z}(\cdot,Y_k|Z_k)]_{\mathcal{B}_{Y}}=R_Ze_k,
\end{align*}
where $R_Z$ is defined by \eqref{eq-rz}.
  Therefore, we have
\begin{align*}
   [\hat{\kappa}_{\Ddot{X}|Z}(\cdot,\Ddot{X}_k|Z_k)\otimes\hat{\kappa}_{Y|Z}(\cdot,Y_k|Z_k)]_{\mathcal{B}_{\Ddot{X}Y}}=\left(R_Ze_k\right)\otimes\left(R_Ze_k\right).
\end{align*}
We plug this result into the coordinate mapping of $\hat{\Gamma}_{\Ddot{X}Y|Z}$, and we have
\begin{align}\label{eq-gammahat-raw}
    \,_{\mathcal{B}_{\Ddot{X}Y}}[\hat{\Gamma}_{\Ddot{X}Y|Z}]_{\mathcal{B}_{\Ddot{X}Y}}=\frac{1}{n}\sum_{k=1}^n[(R_Ze_k)\otimes(R_Ze_k)][(R_Ze_k)\otimes(R_Ze_k)]^{\mathsf{T}}(\tilde{K}_{\Ddot{X}}\otimes\tilde{K}_{Y}).
\end{align}
Since the rank of the matrix $\,_{\mathcal{B}_{\Ddot{X}Y}}[\hat{\Gamma}_{\Ddot{X}Y|Z}]_{\mathcal{B}_{\Ddot{X}Y}}$ is at most $n$, we compute its first $n$ eigenvalues $\hat{\lambda}_1\geq\dots\geq\hat{\lambda}_n$ and use them as estimates of the eigenvalues $\lambda_k$'s in \eqref{eq-asymp-ccco-chisq}.

\subsection{Acceleration of computation}\label{sec-acceleration}

Since the coordinate mapping $\,_{\mathcal{B}_{\Ddot{X}Y}}[\hat{\Gamma}_{\Ddot{X}Y|Z}]_{\mathcal{B}_{\Ddot{X}Y}}$  in \eqref{eq-gammahat-raw},  is an $n^2 \times n^2$ matrix, the direct  computational cost for its eigenvalues is very high even for moderate sample sizes. However, since the rank of this matrix is only $n$, the cost can be substantially reduced by a carefully designed algorithm. In the following we adapt the idea in Section~6.3 of \cite{tang2024nonparametric} to accelerate the computation.

Specifically, since $\tilde{K}_{\Ddot{X}}$ and $\tilde{K}_{Y}$ are symmetric, \eqref{eq-gammahat-raw} can be rewritten as
\begin{align*}
\begin{split}
\,_{\mathcal{B}_{\Ddot{X}Y}}[\hat{\Gamma}_{\Ddot{X}Y|Z}]_{\mathcal{B}_{\Ddot{X}Y}}
=&\frac{1}{n}\sum_{k=1}^n(\tilde{K}_{\Ddot{X}}^{1/2}\otimes\tilde{K}_{Y}^{1/2})[(R_Ze_k)\otimes(R_Ze_k)][(R_Ze_k)\otimes(R_Ze_k)]^{\mathsf{T}}(\tilde{K}_{\Ddot{X}}^{1/2}\otimes\tilde{K}_{Y}^{1/2})\\
=&\sum_{k=1}^n\left[(\tilde{K}_{\Ddot{X}}^{1/2}R_Ze_k)\otimes(\tilde{K}_{Y}^{1/2}R_Ze_k)/\sqrt{n}\right]\left[(\tilde{K}_{\Ddot{X}}^{1/2}R_Ze_k)\otimes(\tilde{K}_{Y}^{1/2}R_Ze_k)/\sqrt{n}\right]^{\mathsf{T}}\\
=&\sum_{k=1}^n\left[(L_{\Ddot{X}}e_k)\otimes(L_{Y}e_k)/\sqrt{n}\right]\left[(L_{\Ddot{X}}e_k)\otimes(L_{Y}e_k)/\sqrt{n}\right]^{\mathsf{T}},
\end{split}
\end{align*}
where $L_{\Ddot{X}}=\tilde{K}_{\Ddot{X}}^{1/2}R_Z$ and $L_{Y}=\tilde{K}_{Y}^{1/2}R_Z$. 
Let $L = (L_1,\dots,L_n)$ be the $n^2 \times n$ matrix, whose $k$th column  is $L_k = (L_{\Ddot{X}}e_k)\otimes(L_{Y}e_k)/\sqrt{n}$. Then, \begin{align*}
\,_{\mathcal{B}_{\Ddot{X}Y}}[\hat{\Gamma}_{\Ddot{X}Y|Z}]_{\mathcal{B}_{\Ddot{X}Y}} = \sum_{k=1}^n L_k L_k^{\mathsf{T}} = L L^{\mathsf{T}}.
\end{align*}
Since the nonzero eigenvalues of $L L^{\mathsf{T}}$ are same as those of $L^{\mathsf{T}} L$, we only need to  calculate the eigenvalues of the $n \times n$ matrix  $L^{\mathsf{T}} L$, substantially reducing the computing time.  The  computer memory needed is also   substantially reduced,  because we never need to save the $n \hi2 \times n \hi 2$ matrix $\,_{\mathcal{B}_{\Ddot{X}Y}}[\hat{\Gamma}_{\Ddot{X}Y|Z}]_{\mathcal{B}_{\Ddot{X}Y}}$, and   the largest matrix we need to save is the $n \hi 2 \times n$ matrix $L$.

\subsection{{Implementation of Welch-Satterthwaite approximation}}\label{sec-sw}
{
To implement the  Welch-Satterthwaite approximation  in subsection \ref{sec-ws-approx} at the sample level, we simply substitute    $\hat\Gamma_{\Ddot{X}Y|Z}$ for $\Gamma \lo {\ddot X Y | Z}$  to estimate $\hat\mu_V$ and $\hat\sigma_V^2$, as follows:
\begin{align}\label{eq:muy sigmay}
\hat\mu_V = \tr (\hat\Gamma_{\Ddot{X}Y|Z}) = \tr (\,_{\mathcal{B}_{\Ddot{X}Y}}[\hat{\Gamma}_{\Ddot{X}Y|Z}]_{\mathcal{B}_{\Ddot{X}Y}}), \qquad
\hat\sigma_V^2 = \tr (\hat\Gamma_{\Ddot{X}Y|Z}^2) = \tr (\,_{\mathcal{B}_{\Ddot{X}Y}}[\hat{\Gamma}_{\Ddot{X}Y|Z}]_{\mathcal{B}_{\Ddot{X}Y}}^2).
\end{align}
By \eqref{eq-kxz-kyz}, we have the following identity: for any $k, \ell \in\{ 1, \ldots, n\}$, 
\begin{align}\label{eq:K odot K}
\begin{split}
&[(R_Ze_k)\otimes(R_Ze_k)]^{\mathsf{T}}(\tilde{K}_{\Ddot{X}}\otimes\tilde{K}_{Y})[(R_Ze_\ell )\otimes(R_Ze_\ell )]= (e_k\trans R_Z \tilde{K}_{\Ddot{X}} R_Z e_\ell ) (e_k\trans R_Z \tilde{K}_{Y} R_Z e_\ell )\\
&
\qquad\qquad = (e_k\trans  \tilde{K}_{\Ddot{X}|Z}  e_\ell ) (e_k\trans  \tilde{K}_{Y|Z}  e_\ell )
= (\tilde{K}_{\Ddot{X}|Z})_{k\ell } (\tilde{K}_{Y|Z})_{k\ell }
= (\tilde{K}_{\Ddot{X}|Z} \odot \tilde{K}_{Y|Z})_{k\ell }.
\end{split}
\end{align}
Substituting \eqref{eq-gammahat-raw} into the first relation in \eqref{eq:muy sigmay} and then invoking the above identity, we have 
\begin{align*}
\hat\mu_V 
&= \tr 
\left\{ \frac{1}{n}\sum_{k=1}^n[(R_Ze_k)\otimes(R_Ze_k)][(R_Ze_k)\otimes(R_Ze_k)]^{\mathsf{T}}(\tilde{K}_{\Ddot{X}}\otimes\tilde{K}_{Y}) \right\} \\
&=  \frac{1}{n}\sum_{k=1}^n \tr 
\left\{[(R_Ze_k)\otimes(R_Ze_k)][(R_Ze_k)\otimes(R_Ze_k)]^{\mathsf{T}}(\tilde{K}_{\Ddot{X}}\otimes\tilde{K}_{Y}) \right\} \\
&=  \frac{1}{n}\sum_{k=1}^n [(R_Ze_k)\otimes(R_Ze_k)]^{\mathsf{T}}(\tilde{K}_{\Ddot{X}}\otimes\tilde{K}_{Y})[(R_Ze_k)\otimes(R_Ze_k)]  \\
&= \frac{1}{n}\sum_{k=1}^n (\tilde{K}_{\Ddot{X}|Z} \odot \tilde{K}_{Y|Z})_{kk}
=\frac{1}{n} \tr  (\tilde{K}_{\Ddot{X}|Z} \odot \tilde{K}_{Y|Z}). 
\end{align*}
Similarly,  substituting \eqref{eq-gammahat-raw} into the second relation in \eqref{eq:muy sigmay} and then invoking identity \eqref{eq:K odot K}, we have 
\begin{align*}
\hat\sigma_V^2  
&= \tr 
\left(\left\{ \frac{1}{n}\sum_{k=1}^n[(R_Ze_k)\otimes(R_Ze_k)][(R_Ze_k)\otimes(R_Ze_k)]^{\mathsf{T}}(\tilde{K}_{\Ddot{X}}\otimes\tilde{K}_{Y}) \right\}^2\right) \\
&=  \frac{1}{n^2}\sum_{k=1}^n\sum_{ \ell =1}^n \tr 
\left\{[(R_Ze_k)\otimes(R_Ze_k)][(R_Ze_k)\otimes(R_Ze_k)]^{\mathsf{T}}(\tilde{K}_{\Ddot{X}}\otimes\tilde{K}_{Y}) \right. \\
& \qquad\qquad\qquad\qquad\quad \left.\times [(R_Ze_\ell )\otimes(R_Ze_\ell )][(R_Ze_\ell )\otimes(R_Ze_\ell )]^{\mathsf{T}}(\tilde{K}_{\Ddot{X}}\otimes\tilde{K}_{Y})\right\} \\
&=  \frac{1}{n^2}\sum_{k=1}^n\sum_{\ell =1}^n [(R_Ze_k)\otimes(R_Ze_k)]^{\mathsf{T}}(\tilde{K}_{\Ddot{X}}\otimes\tilde{K}_{Y})[(R_Ze_\ell )\otimes(R_Ze_\ell )] \\
& \qquad\qquad\qquad \times [(R_Ze_\ell )\otimes(R_Ze_\ell )]^{\mathsf{T}}(\tilde{K}_{\Ddot{X}}\otimes\tilde{K}_{Y}) [(R_Ze_k)\otimes(R_Ze_k)]\\
&=  \frac{1}{n^2}\sum_{k=1}^n\sum_{\ell =1}^n (\tilde{K}_{\Ddot{X}|Z} \odot \tilde{K}_{Y|Z})_{k \ell } (\tilde{K}_{\Ddot{X}|Z} \odot \tilde{K}_{Y|Z})_{\ell  k} \\
&=  \frac{1}{n^2}\sum_{k=1}^n\sum_{\ell =1}^n (\tilde{K}_{\Ddot{X}|Z} \odot \tilde{K}_{Y|Z})_{k \ell }^2
= \frac{1}{n^2} \|\tilde{K}_{\Ddot{X}|Z} \odot \tilde{K}_{Y|Z}\|_\mathrm{F}^2,
\end{align*}
where $\|\cdot\|_\mathrm{F}$ denotes the Frobenius norm of a matrix. Then, we can estimate $\zeta$ and $\nu$ in \eqref{eq-zeta-nu-pop} by
\begin{align}\label{eq-zeta-nu}
\hat\zeta = \frac{\hat\sigma_V^2}{2\hat\mu_V}
= \frac{\|\tilde{K}_{\Ddot{X}|Z} \odot \tilde{K}_{Y|Z}\|_\mathrm{F}^2}{2n \, \tr  (\tilde{K}_{\Ddot{X}|Z} \odot \tilde{K}_{Y|Z})}, \qquad
\hat\nu = \frac{2\hat\mu_V^2}{\hat\sigma_V^2}
= \frac{2 [\tr  (\tilde{K}_{\Ddot{X}|Z} \odot \tilde{K}_{Y|Z})]^2}{\|\tilde{K}_{\Ddot{X}|Z} \odot \tilde{K}_{Y|Z}\|_\mathrm{F}^2}.
\end{align}
}

\subsection{Tuning}

In this section, we describe the procedure  for selecting tuning parameters.
Three types of tuning parameters are involved: the kernel bandwidths,
the regularization parameter $\delta_n$ associated with $K_T$, and the
regularization parameter $\epsilon_n$ associated with $K_Z$.
We propose a step-by-step procedure for selecting these parameters sequentially.

\subsubsection{Kernel bandwidths}\label{subsubsection:kernel bandwidths}

For all kernels, we use Gaussian radial basis functions (RBFs) as
reproducing kernels. Each kernel involves a tuning parameter controlling
the bandwidth. Specifically, the kernel on $T$ is defined as
\begin{align}
\kappa_T : T\times T \to \mathbb{R},\qquad
(t_1,t_2) \mapsto \exp\{-\gamma_T (t_1-t_2)^2\},
\label{eq-gaussian}
\end{align}
where $\gamma_T>0$ is a bandwidth parameter.
Similarly, the kernel on $\mathcal{G}_X$ is
\begin{align}
\kappa_X : \mathcal{G}_X \times \mathcal{G}_X \to \mathbb{R},\qquad
(x_1,x_2) \mapsto \exp\{-\gamma_X \|x_1-x_2\|^2\},
\label{eq-gaussian-x}
\end{align}
where $\gamma_X>0$ is the corresponding bandwidth parameter.
The kernels $\kappa_Y$ and $\kappa_Z$ are defined analogously with tuning
parameters $\gamma_Y$ and $\gamma_Z$.

Following the criterion in Section~6.4 of \cite{li2018nonparametric},
we choose the bandwidth parameters according to the average pairwise
distances. Specifically,
\begin{align}
\frac{1}{\sqrt{\gamma_T}}
=
{\ell+1 \choose 2}^{-1}
\sum_{k=0}^{\ell-1}\sum_{r=k+1}^{\ell}
|u_k-u_r|,
\label{eq-gaussian-tuning}
\end{align}
where $U=\{u_0,u_1,\dots,u_\ell\}$ are the $\ell+1$ equally-spaced points on $T$ as defined in Section~\ref{sec-approx-inner-prod}, and
\begin{align}
\frac{1}{\sqrt{\gamma_X}}
=
{n\choose 2}^{-1}
\sum_{a=1}^{n-1}\sum_{b=a+1}^{n}
\|X_a-X_b\|.
\label{eq-gaussian-tuning-x}
\end{align}
The parameters $\gamma_Y$ and $\gamma_Z$ are chosen in the same way.

\subsubsection{Tikhonov  regularization constant $\delta_n$}
To achieve appropriate scaling, we reset $\delta_n$ to $\delta_n\lambda_{\max}(K_T^{(i)})$, where $\lambda_{\max}(K_T^{(i)})$ is the largest eigenvalue of $K_T^{(i)}$. Inspired by the approach in Section~6.4 of \cite{li2018nonparametric}, we employ the generalized cross-validation (GCV) criterion
\begin{align}
\begin{split}
    \mathrm{GCV}_T(\delta_n)=&\sum_{i=1}^n\frac{\|X_i(T_i)-K_T^{(i)}(K_T^{(i)}+\delta_n\lambda_{max}(K_T^{(i)})I_{m_i})^{-1}X_i(T_i)\|_F^2}{\{1-\tr [K_T^{(i)}(K_T^{(i)}+\delta_n\lambda_{max}(K_T^{(i)})I_{m_i})^{-1}]/m_i\}^2}\\
    &+\sum_{i=1}^n\frac{\|Y_i(T_i)-K_T^{(i)}(K_T^{(i)}+\delta_n\lambda_{max}(K_T^{(i)})I_{m_i})^{-1}Y_i(T_i)\|_F^2}{\{1-\tr [K_T^{(i)}(K_T^{(i)}+\delta_n\lambda_{max}(K_T^{(i)})I_{m_i})^{-1}]/m_i\}^2}\\
    &+\sum_{i=1}^n\frac{\|Z_i(T_i)-K_T^{(i)}(K_T^{(i)}+\delta_n\lambda_{max}(K_T^{(i)})I_{m_i})^{-1}Z_i(T_i)\|_F^2}{\{1-\tr [K_T^{(i)}(K_T^{(i)}+\delta_n\lambda_{max}(K_T^{(i)})I_{m_i})^{-1}]/m_i\}^2}.
\end{split}\label{eq-tuning-delta}
\end{align}
In practice, we minimize $GCV_T(\delta_n)$ over a prespecified range $I_T$
to determine the optimal tuning parameter $\delta_n^*$:
\begin{align}
\delta_n^*={\arg\min}_{\delta_n\in I_T} \mathrm{GCV}_T(\delta_n).
\label{eq-tuning-delta-star}
\end{align}

\subsubsection{Tikhonov  regularization constant $\epsilon_n$}\label{subsec-tuning-z}

Similarly, let $\kappa_{XY}$ be a reproducing kernel on the tensor-product space
$\mathcal{G}_{X\otimes Y}=\mathcal{G}_X\otimes\mathcal{G}_Y$, defined by
$
\kappa_{XY}(\cdot,x;\star,y)
=
\kappa_X(\cdot,x)\,\kappa_Y(\star,y).
$
The corresponding Gram matrix for $\mathcal{H}_{X\otimes Y}$ is
$
K_{XY}=K_X\odot K_Y,$
where $\odot$ denotes the Hadamard product. The centered Gram matrix is therefore
$
\tilde K_{XY}=H(K_X\odot K_Y)H.
$
Let $\mathcal{R}_Z=\mathrm{ran}(\hat{\Sigma}_{ZZ})=\mathrm{span}(\mathcal{B}_Z)$ and
$\mathcal{R}_{XY}=\mathrm{span}(\mathcal{B}_{XY})$, where
\begin{align*}
\mathcal{B}_{XY}
=&
\left\{
\hat{\kappa}_{XY}(\cdot,X_i;\star,Y_i): i\in\{1,\dots,n\}
\right\} \\
=&
\left\{
\kappa_{XY}(\cdot,X_i;\star,Y_i) - E_n[\kappa_{XY}(\cdot,X;\star,Y)]: i\in\{1,\dots,n\}
\right\}.
\end{align*}

We predict functions in $\mathcal{R}_{XY}$ using functions in
$\mathcal{R}_Z$, which amounts to projecting
$\hat{\kappa}_{XY}(\cdot,X_i,\star,Y_i)$ onto $\mathcal{R}_Z$.
The empirical projection is
$
\hat{\Sigma}_{ZZ}^{-1}\hat{\Sigma}_{Z(XY)}
\hat{\kappa}_{XY}(\cdot,X_i,\star,Y_i),
$
which minimizes the empirical variance
\begin{align}
\mathrm{var}_n\!\left\{
\hat{\kappa}_{XY}(X,X_i,Y,Y_i)
-
[A
\hat{\kappa}_{XY}(\cdot,X_i,\star,Y_i)](Z)
\right\}.
\label{eq-stoc-diff-z}
\end{align}
among all bounded linear operators $A: \ca G \lo {X \otimes Y} \to \ca G \lo Z$. 
Under Tikhonov regularization, the coordinate representation of
$\hat{\Sigma}_{ZZ}^{-1}\hat{\Sigma}_{Z(XY)}$ with respect to the bases
$\mathcal{B}_Z$ and $\mathcal{B}_{XY}$ is
\[
\,_{\mathcal{B}_Z}[\hat{\Sigma}_{ZZ}^{-1}\hat{\Sigma}_{Z(XY)}]_{\mathcal{B}_{XY}}
=
(\tilde K_Z+n\epsilon_n I)^{-1}\tilde K_{XY}.
\]
Hence the empirical error in \eqref{eq-stoc-diff-z} becomes
\[
\|\tilde K_{XY}-\tilde K_Z(\tilde K_Z+n\epsilon_n I)^{-1}\tilde K_{XY}\|_{\mathrm{F}}^2, 
\]
where $\| \cdot \| \lo {\mathrm{F}}$ stands for the Frobenius matrix norm. 

As with $\delta \lo n$,   we replace $\epsilon_n$ by
$\epsilon_n\lambda_{\max}(\tilde K_Z)$ to ensure appropriate scaling.
The generalized cross-validation (GCV) criterion is
\begin{align}
\mathrm{GCV}_Z(\epsilon_n)=
\frac{
\|\tilde K_{XY}-\tilde K_Z(\tilde K_Z+n\epsilon_n
\lambda_{\max}(\tilde K_Z)I)^{-1}\tilde K_{XY}\|\lo {\mathrm{F}} \hi 2
}{
\{
1-\tr  [
\tilde K_Z(\tilde K_Z+n\epsilon_n
\lambda_{\max}(\tilde K_Z)I)^{-1}
]/n
\}^2
}.
\label{eq-tuning-eps-z}
\end{align}
The optimal parameter is chosen as
\begin{align}
\epsilon_n^*={\arg\min}_{\epsilon_n\in I_Z}\mathrm{GCV}_Z(\epsilon_n).
\label{eq-tuning-eps-z-star}
\end{align}
for some prespecified range $I \lo Z$.

\subsection{Algorithm}
We summarize the procedures developed before as the following algorithm.
\begin{enumerate}
    \item Choose the kernel $\kappa_T$ for $\mathcal{H}$, along with its corresponding bandwidth. For example, if $\kappa_T$ is the Gaussian kernel \eqref{eq-gaussian}, then set the bandwidth by \eqref{eq-gaussian-tuning}.
    \item Compute the matrices $K_T^{(i)}$ based on $\kappa_T$, $i\in\{1,\dots,n\}$, by \eqref{eq-kt-i}.
    \item Choose the regularization parameter $\delta_n$ by minimizing $GCV_T(\delta_n)$ in \eqref{eq-tuning-delta} over a prespecified range.
    \item Based on the observations $X_i(T_i)$, $Y_i(T_i)$, $Z_i(T_i)$, $i\in\{1,\dots,n\}$, compute the coordinate mapping $[\hat{X}_i]_{\mathcal{L}_i}$, $[\hat{Y}_i]_{\mathcal{L}_i}$, $[\hat{Z}_i]_{\mathcal{L}_i}$ by \eqref{eq-xhat-coord}, $i\in\{1,\dots,n\}$.
    \item Set the grid points $U$ and compute the matrix $\kappa_T(T_i,U)$, $i\in\{1,\dots,n\}$. Compute the inner products 
    $ \langle \hat X \lo i, \hat X \lo j \rangle$, $\langle \hat Y \lo i , \hat Y \lo j \rangle$ and $\langle Z \lo i, Z \lo j \rangle $
    by (\ref{eq-unbalanced}), 
    and further compute 
$ \| \hat X \lo i -  \hat X \lo j \|$, $\| \hat Y \lo i -  \hat Y \lo j \|$ and $ \| \hat Z \lo i -  \hat Z \lo j \|  $
    by (\ref{eq:norm hat X i}). 
    \item Choose the bandwidths for the Gaussian kernels  $\kappa_X$, $\kappa_Y$, $\kappa_Z$ by \eqref{eq-gaussian-tuning-x}.
    \item Compute the Gram matrices $K_X$, $K_Y$, $K_Z$ based on the kernels $\kappa_X$, $\kappa_Y$, $\kappa_Z$, and further compute $K_{\Ddot{X}}$ based on \eqref{eq-k-ddot}.
    \item Choose the regularization parameters $\delta \lo n$ and  $\epsilon_n$ by minimizing $\mathrm{GCV} \lo T (\delta \lo n)$ in (\ref{eq-tuning-delta}) and $\mathrm{GCV}_Z(\epsilon_n)$ in \eqref{eq-tuning-eps-z}, respectively,  over their prespecified ranges. 
    \item Compute $R_Z$ using \eqref{eq-rz}, and further compute $\tilde{K}_{\Ddot{X}|Z}$ and $\tilde{K}_{Y|Z}$ using \eqref{eq-kxz-kyz}.
    \item Compute the test statistic $n\|\hat{\Sigma}_{\Ddot{X}Y|Z}\|_{\mathrm{HS}}^2$ via \eqref{eq-sigma-hat-hs}. 
    \item Compute the eigenvalues $\hat{\lambda}_1$, \dots, $\hat{\lambda}_n$ of $L^{\mathsf{T}}L$ as defined in Section~\ref{sec-acceleration}. 
    \item Perform the test using the asymptotic distribution in \eqref{eq-asymp-ccco-chisq}.
\end{enumerate}

One side note  is that the first five steps can be simplified under some special cases, where all observations are on the same set of equally-spaced grid points. The simplified steps are to be carried out according to Section~\ref{sec-simplification}. {Moreover, if Welch-Satterthwaite approximation is applied, the last two steps are modified as follows:
\begin{enumerate}
\item[11'.] Compute $\hat\zeta$ and $\hat\nu$ as given in \eqref{eq-zeta-nu}.
\item[12'.] Perform the test using \eqref{eq-sw-approx} as an approximation of the asymptotic distribution in \eqref{eq-asymp-ccco-chisq}.
\end{enumerate}}

\section{Simulation}\label{sec-simulation}
We first generate the random functions $\epsilon_i(t)$ for $i\in\{1,2,3\}$ by 
\begin{align*}
    \epsilon_i(t)=\sum_{k=1}^r\xi_k\kappa_T(t,t_k),
\end{align*}
where $\ka \lo T (t,s) = \min (t,s)$ is the Brownian motion kernel,  $\xi_1,\dots,\xi_r$ are i.i.d.  $N(0,1)$, $t_1,\dots,t_r$ are i.i.d.  $U(0,1)$. We generate $\epsilon_1,\epsilon_2,\epsilon_3$ independently according to the above model, where we choose $r=50$.
Note that, even though we have used the Brownian motion kernel to generate the functional random errors, we use the Gaussian kernel in our estimation process, as mentioned in Section~\ref{subsubsection:kernel bandwidths}. 

For all models, the random functions $X$ and $Z$ are generated from
\[
Z(t)=\epsilon_1(t), \qquad X(t)=2Z(t)+\epsilon_2(t).
\]
The dependence structures are introduced through different
representations of $Y(t)$ in the five models described below.
The models are summarized in Table \ref{tab:models}.

\begin{table}[htbp]
\centering
\caption{Data-generating models used in the simulations. For all models, 
$Z(t)=\epsilon_1(t)$ and $X(t)=2Z(t)+\epsilon_2(t)$, where $\epsilon_1,\epsilon_2,\epsilon_3$ are mutually independent 
functional errors. Model 1 satisfies $X\indep Y\mid Z$; Models 2–5 represent quadratic, 
linear, logarithmic, and multiplicative alternatives, respectively.}\label{tab:models}
\centering

\smallskip
    \begin{tabular}{cl}
        \hline
         Model &  Representation of $Y(t)$ \\
         \hline
         1 & $Y(t)=Z(t)+\epsilon_3(t)$ \\
         2 & $Y(t)=Z(t)+X^2(t)+\epsilon_3(t)$\\
         3 & $Y(t)=Z(t)+X(t)+\epsilon_3(t)$\\
         4 & $Y(t)=Z(t)+4\log(|X(t)|+1)+\epsilon_3(t)$\\
         5 & $Y(t)=Z(t)X(t)+\epsilon_3(t)$\\
         \hline
    \end{tabular}

\smallskip
\end{table}

In Table~\ref{tab:models}, only Model~1 corresponds to the case where
$X(t)$ and $Y(t)$ are independent given $Z(t)$.
In Models~2--4, $Y(t)$ depends on both $Z(t)$ and $X(t)$ through
additive structures with different transformations applied to $X(t)$.
In Model~5, we consider a multiplicative interaction between
$Z(t)$ and $X(t)$.

In the simulation study, the sample sizes are set to be  $n\in\{50,100,200,500\}$, and the number of simulation replications is set to
$n_{\mathrm{sim}}=100$. We examine both balanced and unbalanced
observation schemes with $m=200$.
In the balanced case, the observation times $t_{i0},\dots,t_{im}$ are
fixed as
$t_{ij}=j/200$, for $j\in\{0,1,\dots,200\}$.
In the unbalanced case, for the $i$th subject, the observation times
$t_{i1},\dots,t_{im}$ are obtained by randomly sampling $200$ points
without replacement from
$ 
U=\{j/2000 : j\in\{1,\dots,2000\}\},
 $
and we set $t_{i0}=0$.
We then use \eqref{eq-gaussian-tuning} with grid points $\{0\} \cup U$ to calculate $\gamma_T$, and use the GCV criterion \eqref{eq-tuning-delta-star} to determine
the tuning parameter $\delta_n$, where the prespecified search range
$I_T$ is given by
\begin{align*}
    I_T = \{ 0.001, 0.002, \dots, 0.009, 0.01, 0.02, \dots, 0.09, 0.1 \}.
\end{align*}
Furthermore, for selecting the regularization parameter $\epsilon_n$,
we adopt the GCV criterion \eqref{eq-tuning-eps-z-star} with the
prespecified search range $I_Z$ given by
\[
I_Z = \{0.001/n, 0.002/n, \dots, 0.009/n, 0.01/n, 0.02/n, \dots, 0.09/n, 0.1/n\}.
\]
The resulting $p$-values for all models under the balanced and
unbalanced settings are displayed in
Figures~\ref{fig-res-1}--\ref{fig-res-5}, {where the test is performed either using the original weighted sum of $\chi^2_{(1)}$ distributions in \eqref{eq-asymp-ccco-chisq} or using the Welch-Satterthwaite approximation \eqref{eq-sw-approx}}.

\begin{figure}[!htb] 
\centering 
\begin{subfigure}{0.48\textwidth}
\includegraphics[width=\linewidth]{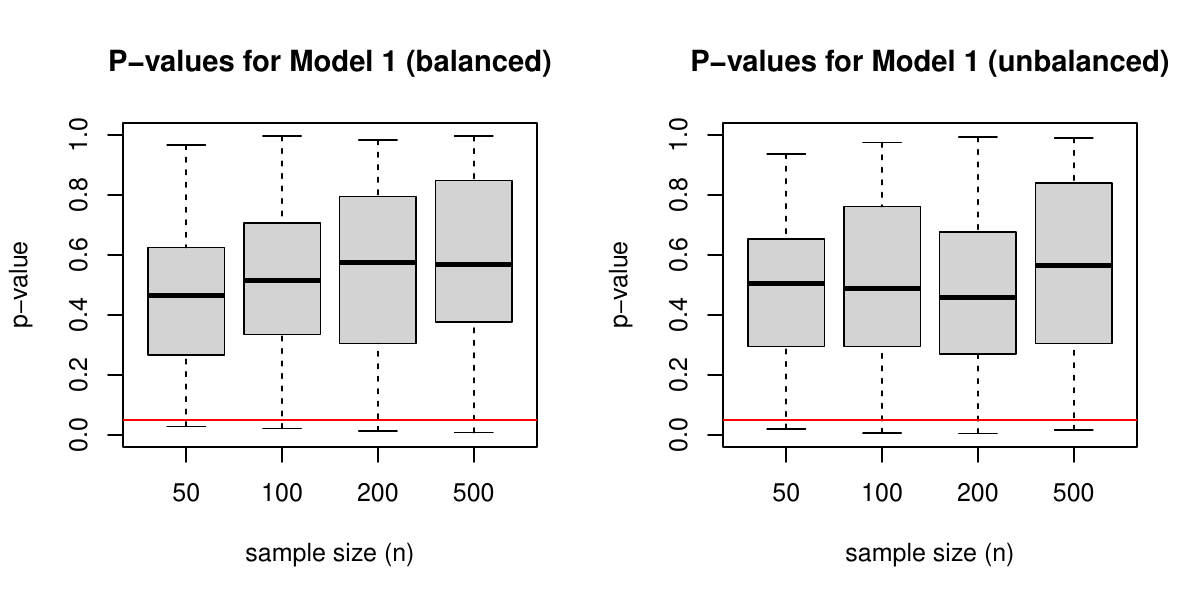}
\caption{Results via weighted sum of $\chi^2_{(1)}$ distributions.}
\end{subfigure}
\begin{subfigure}{0.48\textwidth}
\includegraphics[width=\linewidth]{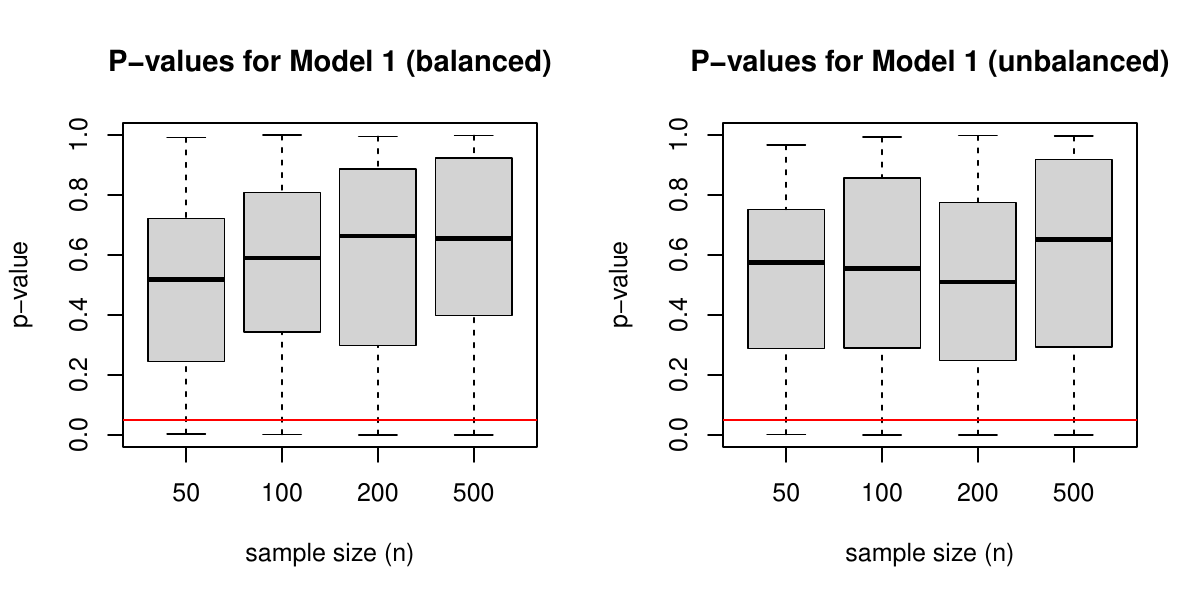}
\caption{Results with Welch-Satterthwaite approximation.}
\end{subfigure}
\caption{Simulation $p$-values for Model 1, where $Y(t) = Z(t) + \epsilon_3(t)$. Each boxplot summarizes 100 simulation replications. The sample sizes are $n\in\{50, 100, 200, 500\}$. The left and right plots within each panel show results for the balanced and unbalanced observation schedules. Panel (a) uses the weighted sum of $\chi^2_{(1)}$ distributions, and panel (b) uses the Welch–Satterthwaite approximation. The red line marks the significance level $\alpha=0.05$. Model 1 satisfies $X\indep Y\mid Z$, and most of the tests correctly do not reject this null hypothesis.}
\label{fig-res-1} 
\end{figure}

\begin{figure}[!htb] 
\centering 
\begin{subfigure}{0.48\textwidth}
\includegraphics[width=\linewidth]{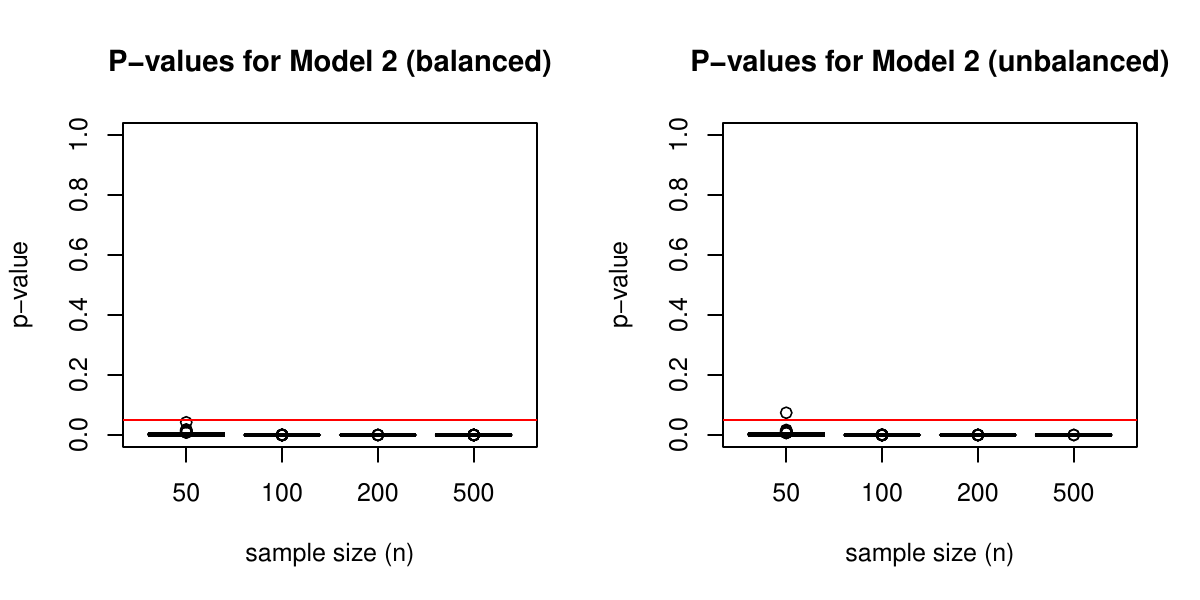}
\caption{Results via weighted sum of $\chi^2_{(1)}$ distributions.}
\end{subfigure}
\begin{subfigure}{0.48\textwidth}
\includegraphics[width=\linewidth]{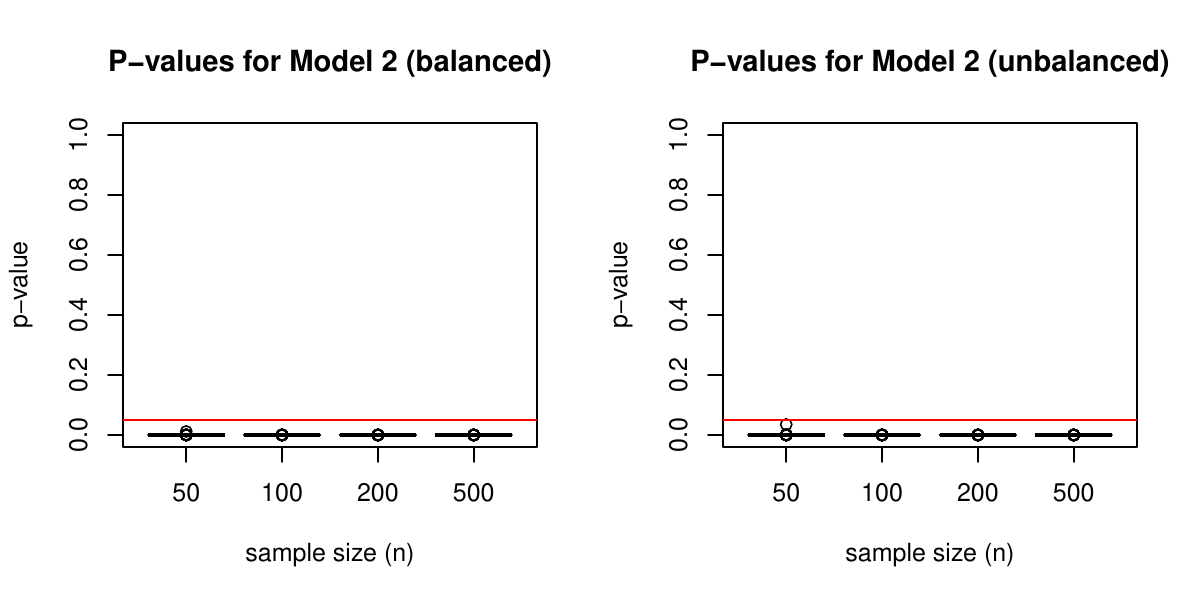}
\caption{Results with Welch-Satterthwaite approximation.}
\end{subfigure}
\caption{Simulation $p$-values for Model 2, where $Y(t)=Z(t)+X^2(t)+\epsilon_3(t)$. The panel arrangement, sample sizes, observation schedules, methods for computing the $p$-values, and significance reference line are the same as in Fig. \ref{fig-res-1}. Both methods successfully detect this conditional dependence.}
\label{fig-res-2} 
\end{figure}

\begin{figure}[!htb] 
\centering 
\begin{subfigure}{0.48\textwidth}
\includegraphics[width=\linewidth]{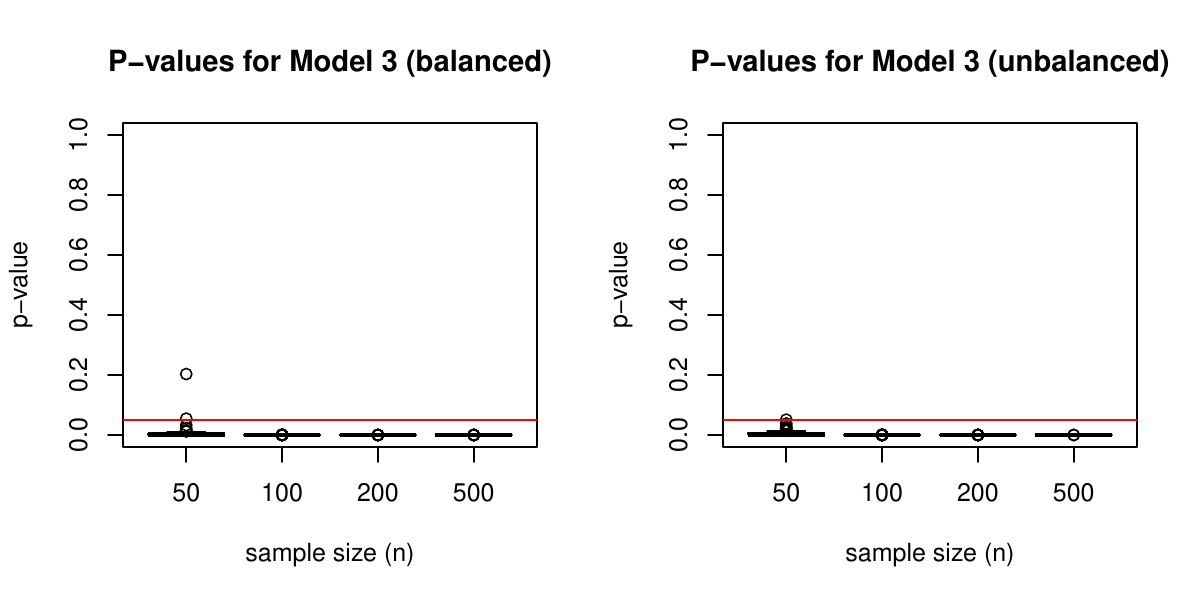}
\caption{Results via weighted sum of $\chi^2_{(1)}$ distributions.}
\end{subfigure}
\begin{subfigure}{0.48\textwidth}
\includegraphics[width=\linewidth]{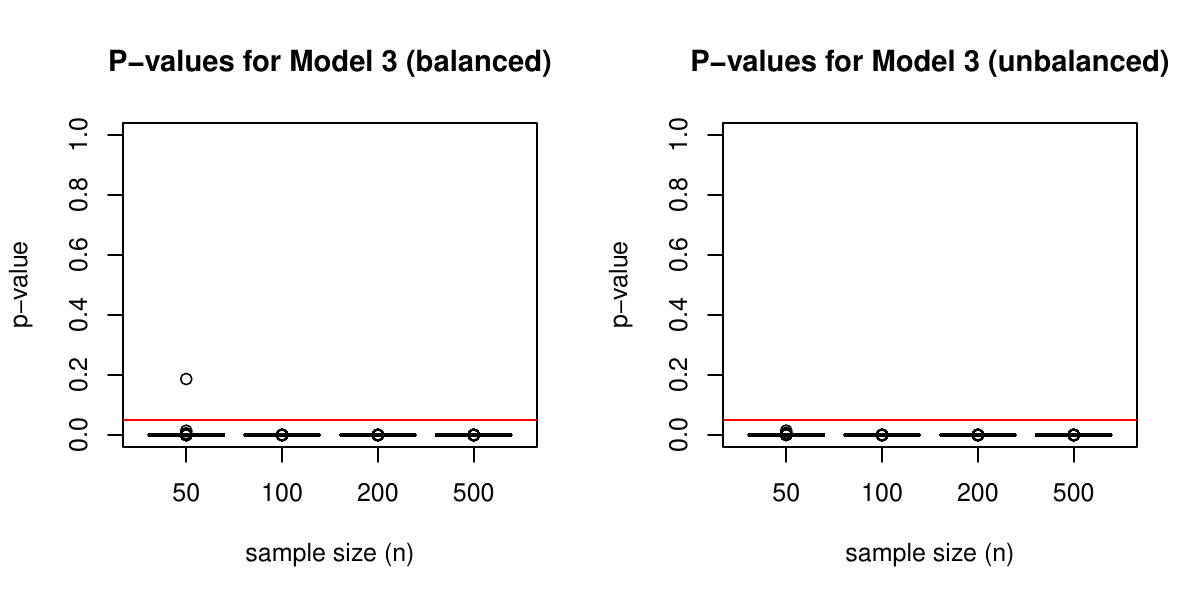}
\caption{Results with Welch-Satterthwaite approximation.}
\end{subfigure}
\caption{Simulation $p$-values for Model 3, where $Y(t)=Z(t)+X(t)+\epsilon_3(t)$. The panel arrangement, sample sizes, observation schedules, methods for computing the $p$-values, and significance reference line are the same as in Fig. \ref{fig-res-1}. Both methods successfully detect this conditional dependence.}
\label{fig-res-3} 
\end{figure}

\begin{figure}[!htb] 
\centering 
\begin{subfigure}{0.48\textwidth}
\includegraphics[width=\linewidth]{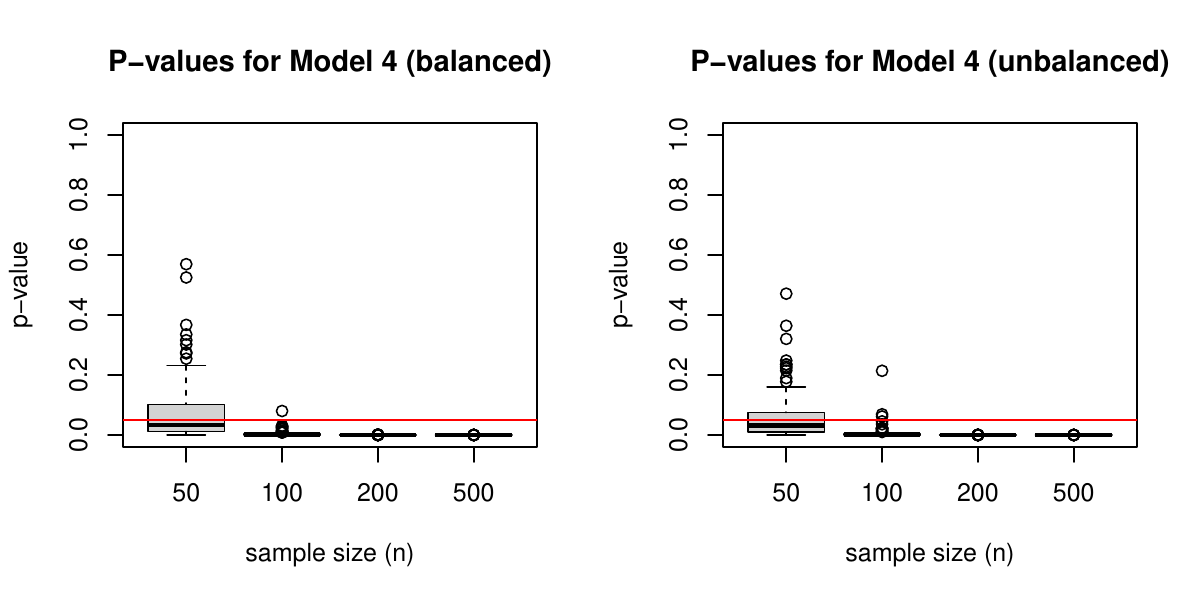}
\caption{Results via weighted sum of $\chi^2_{(1)}$ distributions.}
\end{subfigure}
\begin{subfigure}{0.48\textwidth}
\includegraphics[width=\linewidth]{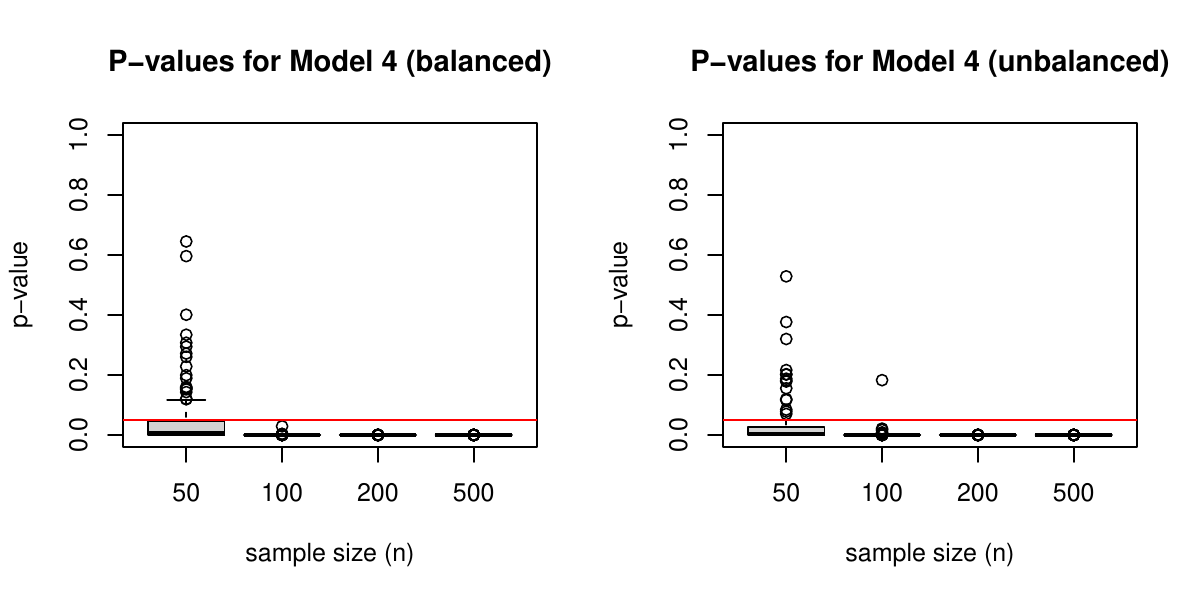}
\caption{Results with Welch-Satterthwaite approximation.}
\end{subfigure}
\caption{Simulation $p$-values for Model 4, where $Y(t)=Z(t)+4\log(|X(t)|+1)+\epsilon_3(t)$. The panel arrangement, sample sizes, observation schedules, methods for computing the $p$-values, and significance reference line are the same as in Fig. \ref{fig-res-1}. This dependence is hard to detect at $n=50$, but when sample size increases, both methods successfully detect this conditional dependence.}
\label{fig-res-4} 
\end{figure}

\begin{figure}[!htb] 
\centering 
\begin{subfigure}{0.48\textwidth}
\includegraphics[width=\linewidth]{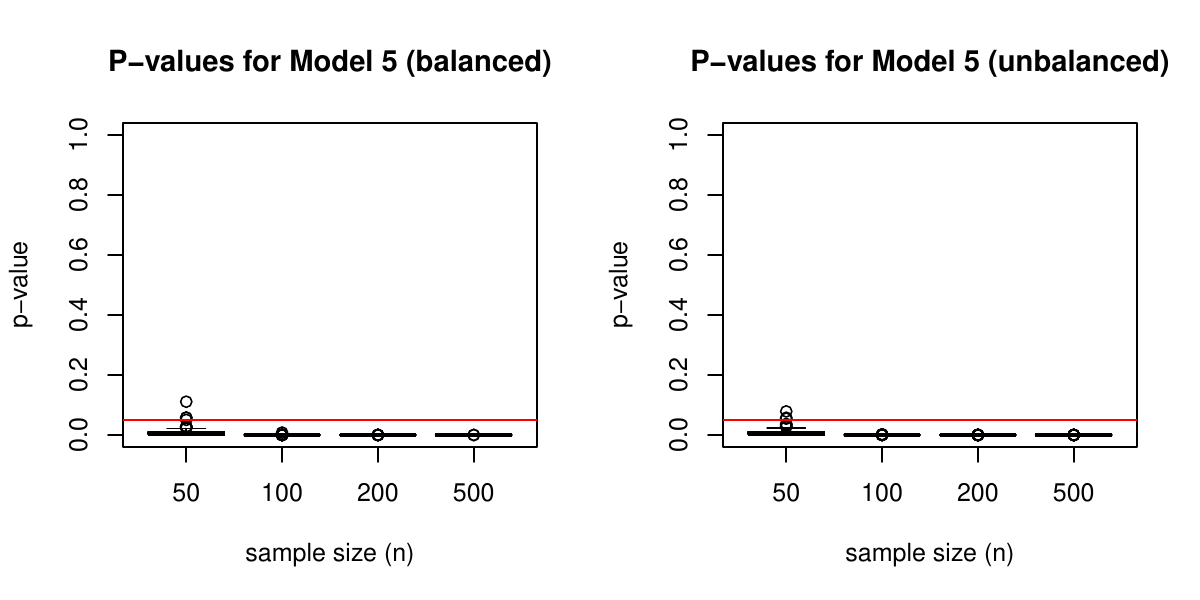}
\caption{Results via weighted sum of $\chi^2_{(1)}$ distributions.}
\end{subfigure}
\begin{subfigure}{0.48\textwidth}
\includegraphics[width=\linewidth]{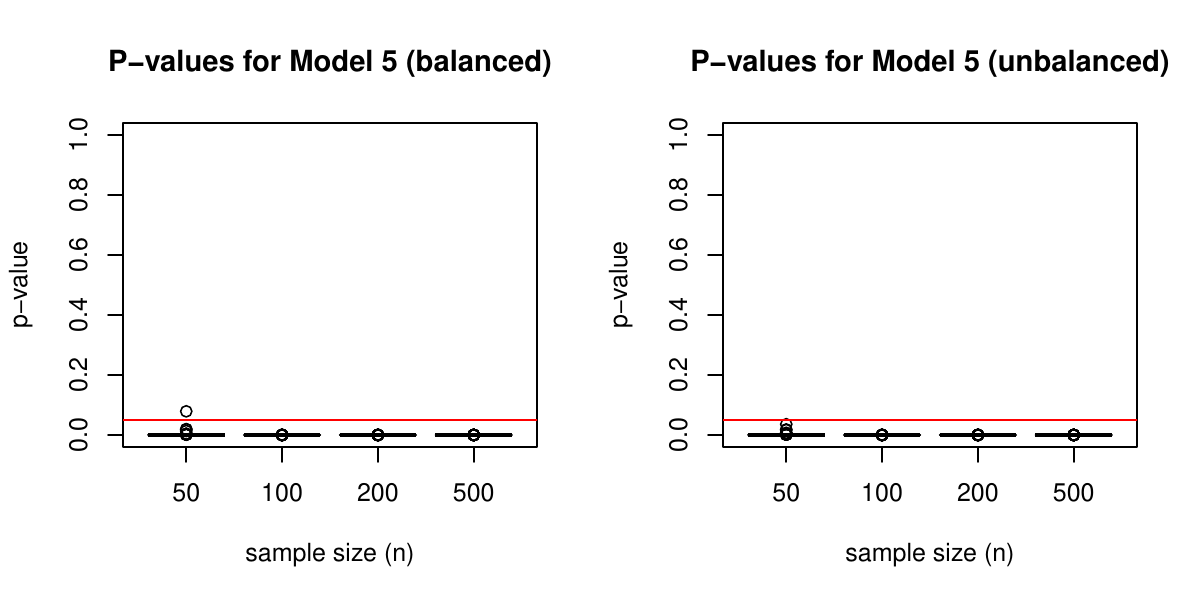}
\caption{Results with Welch-Satterthwaite approximation.}
\end{subfigure}
\caption{Simulation $p$-values for Model 5, where $Y(t)=Z(t)X(t)+\epsilon_3(t)$. The panel arrangement, sample sizes, observation schedules, methods for computing the $p$-values, and significance reference line are the same as in Fig. \ref{fig-res-1}. Both methods successfully detect this conditional dependence.}
\label{fig-res-5} 
\end{figure}

Figs.~\ref{fig-res-1}--\ref{fig-res-5} show that the performances
under the balanced and unbalanced settings are similar for all models. 
{Also, the results of tests using Welch-Satterthwaite approximation is similar to those via weighted sum of $\chi^2_{(1)}$ distributions, indicating that the Welch-Satterthwaite method approximates the asymptotic distribution \eqref{eq-asymp-ccco-chisq} well.}
In particular, Figure~\ref{fig-res-1} indicates that most of the
$p$-values are much larger than $0.05$, suggesting that $X$ and $Y$ are independent
given $Z$, which agrees with the conditional-independence null hypothesis.
For the remaining models, nearly all $p$-values fall below $0.05$, which demonstrate the strong power of
our CCCO-based test. 
The only exception occurs for the small sample size  ($n=50$) in Model~4, as shown in
Figure~\ref{fig-res-4}. This behavior is reasonable because the model
involves the transformation $\log(|X(t)|+1)$, which weakens the
influence of $X(t)$ on $Y(t)$ and makes   conditional dependence
 more difficult to detect when the sample size is small.

{We also examined the simulation running times, as reported in Table \ref{tab-time}. In the table, Stage 1 refers to the computation of the test statistic $T_n$, which is required regardless of whether the Welch-Satterthwaite approximation is used. Stage 2.1 refers to performing the test using the Welch-Satterthwaite approximation, whereas Stage 2.2 refers to performing the test using the original weighted sum of $\chi^2_{(1)}$ distributions.
From Table \ref{tab-time}  we see that Stage 2.1 is significantly faster than Stage 2.2, especially when the sample size is large. In fact, Stage 2.2 requires calculating the eigenvalues of an $n \times n$ matrix $L\trans L $  where $L$ is an $n^2 \times n$ matrix, which is very time consuming for large $n$. In contrast, the Welch-Satterthwaite approximation only requires the trace and Frobenius norm of the $n \times n$ matrix $\tilde{K}_{\Ddot{X}|Z} \odot \tilde{K}_{Y|Z}$, which had already been calculated in Stage 1. Consequently,  very little   additional calculation is needed.   We also note that the   time for computing the test statistic is substantially shorter under the balanced observation schedule, since the inner products of the random functions can be approximated directly by \eqref{eq-balanced}, without the smoothing step required to estimate the entire functions.}

\begin{table}[htbp]
\caption{Computation times, in seconds, for different sample sizes under balanced and unbalanced observation schedules. Stage 1 calculates the test statistic $T_n$. Stages 2.1 and 2.2 compute the $p$-value using the Welch–Satterthwaite approximation and the weighted sum of $\chi^2_{(1)}$ distributions, respectively. Stage 1 takes longer under the unbalanced schedule, while Stage 2.1 is substantially faster than Stage 2.2 for large sample sizes.} 
\label{tab-time}
\centering

\smallskip
\begin{subtable}[htbp]{0.48\textwidth}
\centering
\caption{Elapsed time (in seconds) for balanced cases.}
\begin{tabular}{rrrr}
\hline
$n$ & Stage 1 & Stage 2.1 & Stage 2.2\\
\hline
50 & 0.3763 & 0.0373 & 0.1121\\
100 & 0.6971 & 0.0393 & 0.3808\\
200 & 1.9652 & 0.0273 & 2.2478\\
500 & 11.8508 & 0.0101 & 27.5814\\
\hline
\end{tabular}
\end{subtable}
\hfill
\begin{subtable}[htbp]{0.48\textwidth}
\centering
\caption{Elapsed time (in seconds) for unbalanced cases.}
\begin{tabular}{rrrr}
\hline
$n$ & Stage 1 & Stage 2.1 & Stage 2.2\\
\hline
50 & 5.1011 & 0.0074 & 0.1049\\
100 & 12.9999 & 0.0051 & 0.3918\\
200 & 38.8418 & 0.0050 & 2.3504\\
500 & 201.1896 & 0.0062 & 26.7144\\
\hline
\end{tabular}
\end{subtable}

\smallskip
\end{table}

\section{Application on WISDM Dataset}\label{sec-application}
We apply our method to the Wireless Sensor Data Mining (WISDM) dataset
\citep{kwapisz2010activity}.\footnote{The dataset is available at 
\url{https://www.cis.fordham.edu/wisdm/dataset.php}. 
We use the file \texttt{WISDM\_ar\_v1.1\_raw.txt} from the Activity Prediction Dataset (v1.1).}
This dataset contains labeled accelerometer measurements collected from multiple users
under controlled laboratory conditions by the WISDM Lab.
For each subject, triaxial accelerometer signals corresponding to three spatial directions are recorded during various physical activities.

Thirty-six users participate in the experiment, during which
accelerometer measurements are collected while they perform various
activities, including Walking, Jogging, Upstairs, Downstairs, Sitting,
and Standing. The predictor consists of multivariate
functional data, with time series representing accelerations along the
$x$-, $y$-, and $z$-directions. For each user-activity pair, the
accelerations are recorded at a sampling rate of 20 Hz (i.e., one sample
every 50 ms). {For each user performing a given
activity, we use the observations from the first 10, 15, and 20 seconds,
which correspond to $m=200$, $300$, and $400$ time points, respectively. Taking $m=200$,   for example,  the curves of each activity are given by Figures \ref{fig-curve-1}--\ref{fig-curve-6} in Appendix \ref{sec-app-plot}, where $X(t)$, $Y(t)$, and  $Z(t)$ represent the acceleration curves along the $x$-, $y$-, and $z$-directions, respectively.}

{We plot the first 15 eigenvalues of $\hat{\Sigma}_{ZZ}$ for the six activities  in Figure \ref{fig-ev-zz}. As shown in the figure, the decaying patterns of the eigenvalues for all six activities suggest  that  Assumption \ref{ass-alpha} is justifiable}.

\begin{figure}[!htb] 
\centering 
\includegraphics[width=\linewidth]{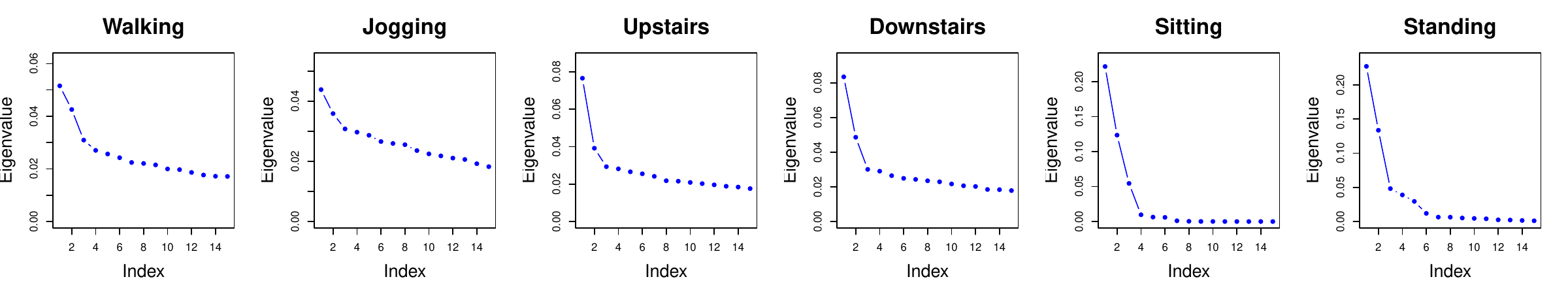}
\caption{The first 15 eigenvalues of the estimated covariance operator $\hat{\Sigma}_{ZZ}$ for the six WISDM activities. The eigenvalues are ordered from largest to smallest within each panel. The vertical scales differ across activities to show the decay patterns clearly. The observed decay supports Assumption \ref{ass-alpha}.}
\label{fig-ev-zz} 
\end{figure}

{To justify Assumption \ref{ass-beta}, by Proposition \ref{proposition:equivalence}, it suffices to justify the boundedness of the three quantities in \eqref{eq-yz-beta} and \eqref{eq-xddotz-beta}. Specifically, we consider the quantities 
\begin{align*}
    &s_{jk}^{(1)} = \hat\mu_j^{1/2} \hat\gamma_k^{-1/2-\beta} |\hat E(\hat\theta_j \hat\xi_k)|, \qquad 
    s_{jk}^{(2)} = \hat\eta_j^{1/2} \hat\gamma_k^{-1/2-\beta} |\hat E(\hat\omega_j \hat\xi_k)|, \qquad 
    j,k \in\{ 1,2,\dots,15\},\\
    &s_{jk\ell}^{(3)} = \hat\eta_j^{1/2} \hat\gamma_k^{1/2} \hat\gamma_\ell ^{-1/2-\beta} |\hat E(\hat\omega_j \hat\xi_k \hat\xi_\ell )|, \qquad j,k,\ell \in\{ 1,2,\dots,15\},
\end{align*}
where all the hats refer to the empirical estimators of the corresponding quantities. We pick $\beta=0.25$.
The heatmaps for $s_{jk}^{(1)}$ and $s_{jk}^{(2)}$ are presented in Figs. \ref{fig-corr-y} and \ref{fig-corr-x}, respectively. For $s_{jk\ell}^{(3)}$, we fix $k$ in each heatmap and present the plots {for $k\in\{1,\dots,5\}$} in Fig. \ref{fig-corr-xz-1} and {Fig. \ref{fig-corr-xz-2}} in the Supplementary Material. {We also report the maximum values of $s_{jk\ell}^{(3)}$, for $k\in\{1,\dots,15\}$ and for all six activities in Table \ref{tab:maximum-corr-xz} in the Supplementary Material. The heatmaps for $k\in\{6,\dots,15\}$ are omitted because all values are less than 0.01 according to Table \ref{tab:maximum-corr-xz}.} Note that the ranges of Figs. \ref{fig-corr-y} and \ref{fig-corr-x} are set to be $[0,0.8]$ while the ranges of Fig. \ref{fig-corr-xz-1} and Fig. \ref{fig-corr-xz-2} are set to be $[0,0.3]$.
We also give the maximum values of $s_{jk}^{(1)}$, $s_{jk}^{(2)}$ and $s_{jk\ell}^{(3)}$ over their corresponding index ranges as Table \ref{tab-max}. As we can see from the figures, most values of the quantities are very small. Relatively large values only appear in Sitting and Standing with small indices $(j,k)$ or $(j,k,\ell)$, while with the increase of the indices, all the quantities have a decaying trend. This justifies the boundedness of the three quantities in \eqref{eq-yz-beta} and \eqref{eq-xddotz-beta}, which further justifies that Assumption \ref{ass-beta} is plausible for this dataset.
}

\begin{figure}[!htb] 
\centering 
\includegraphics[width=\linewidth]{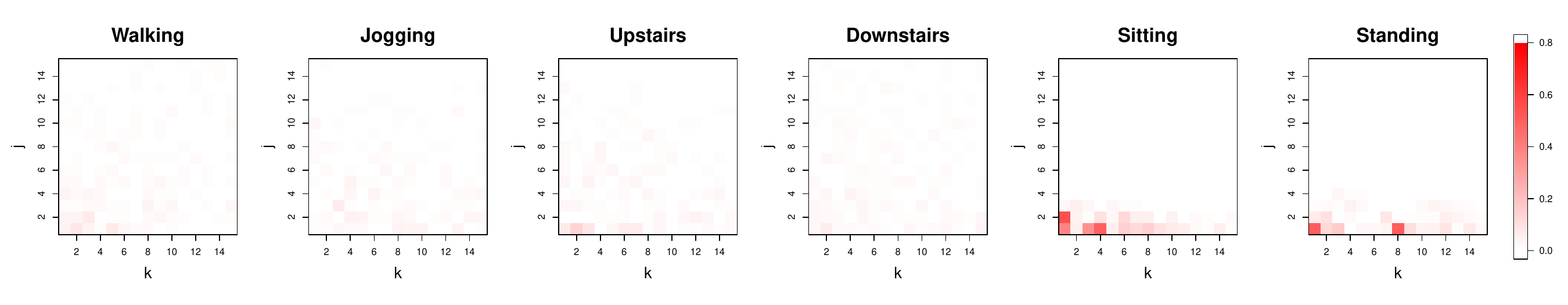}
\caption{Heatmaps of $s_{jk}^{(1)}$ for $j,k\in\{1,\ldots,15\}$, which measures a scaled association between the $j$-th component on the $Y$ side and the $k$-th component of $Z$. All panels use the same color scale from $0$ to $0.8$, with darker colors indicating larger values. Larger values are concentrated in the low-order components for Sitting and Standing.}
\label{fig-corr-y} 
\end{figure}

\begin{figure}[!htb] 
\centering 
\includegraphics[width=\linewidth]{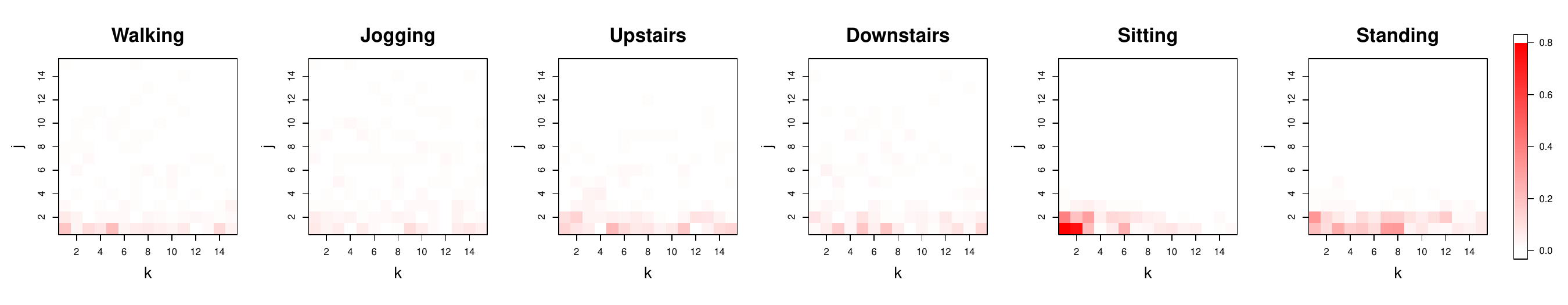}
\caption{Heatmaps of $s_{jk}^{(2)}$ for $j,k\in\{1,\ldots,15\}$, which measures a scaled association between the $j$-th component on the $\ddot{X}$ side and the $k$-th component of $Z$. All panels use the same color scale from $0$ to $0.8$, with darker colors indicating larger values. Larger values are concentrated in the low-order components for Sitting.}
\label{fig-corr-x} 
\end{figure}

\begin{figure}[!htb] 
\centering 
\includegraphics[width=\linewidth]{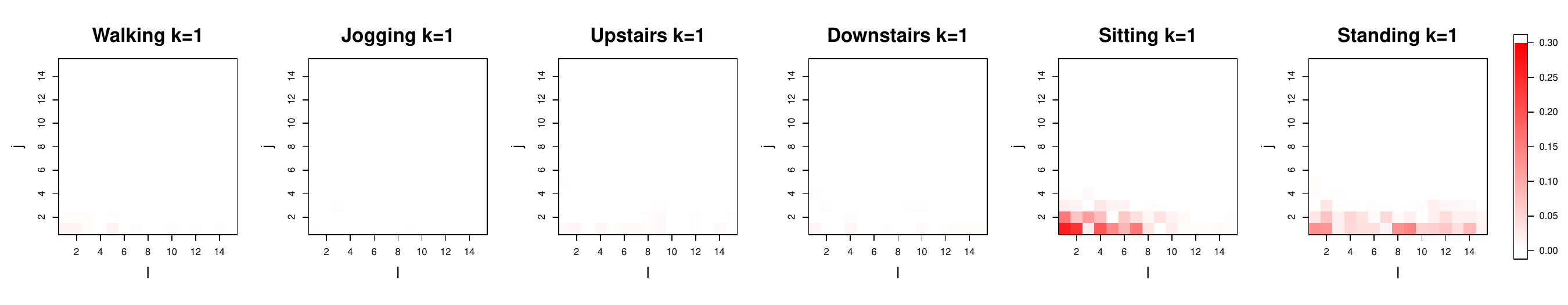}
\includegraphics[width=\linewidth]{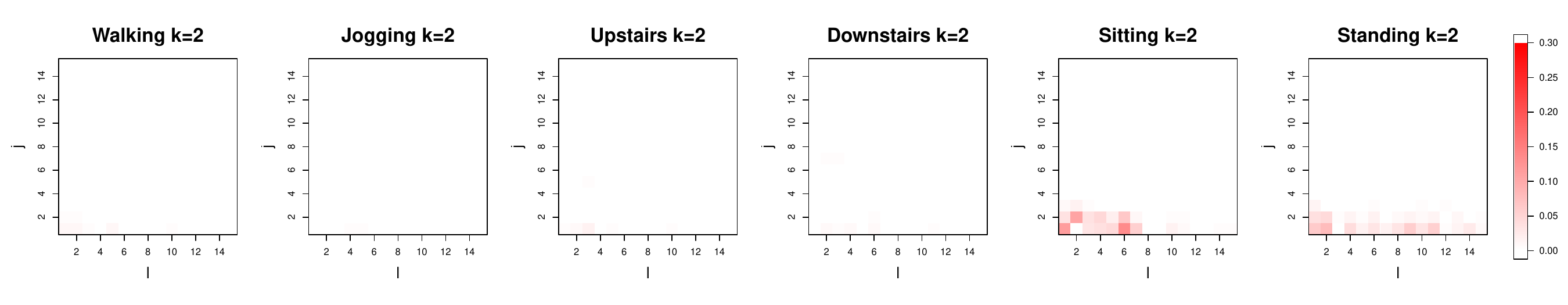}
\caption{Heatmaps of $s_{jk\ell}^{(3)}$ for $j,\ell\in\{1,\ldots,15\}$, which measures a scaled third-order association involving the $j$-th component on the $\ddot{X}$ side and the $k$-th and $\ell$-th components of $Z$. The upper and lower rows correspond to $k=1$ and $k=2$, respectively. All panels use the same color scale from $0$ to $0.3$, with darker colors indicating larger values. Larger values are concentrated in the low-order components for Sitting and Standing.}
\label{fig-corr-xz-1} 
\end{figure}

\begin{table}[h]
\caption{Maximum values of the three empirical quantities over the first 15 components. The rows report $\max_{j,k}s_{jk}^{(1)}$, $\max_{j,k}s_{jk}^{(2)}$, and $\max_{j,k,\ell}s_{jk\ell}^{(3)}$. Sitting has the largest maximum for all three quantities.}
\label{tab-max}
\centering

\smallskip
\begin{tabular}{lcccccc}
\hline
  & Walking & Jogging & Upstairs & Downstairs & Sitting & Standing\\
\hline
$\max s_{jk}^{(1)}$ & 0.0759 & 0.0589 & 0.1370 & 0.0565 & 0.5495 & 0.5054\\
$\max s_{jk}^{(2)}$ & 0.2024 & 0.1046 & 0.2231 & 0.1788 & 0.7947 & 0.3230\\
$\max s_{jk\ell}^{(3)}$ & 0.0154 & 0.0042 & 0.0172 & 0.0096 & 0.2636 & 0.1413\\
\hline
\end{tabular}

\smallskip
\end{table}

Our goal is  to test, for each activity, whether the accelerations in the
$x$- and $y$-directions are conditionally independent given the
acceleration in the $z$-direction. 
Note that the sample size $n$ varies across activities, because some
users do not perform all activities in the dataset.
For each activity, we perform the CCCO-based test under each sample size
and time window. The resulting $p$-values are summarized in
Table~\ref{tab:tab-pval-app}, {where panel (a) is obtained from the weighted sum of $\chi \lo {(1)} \hi 2$ distributions, and panel (b) is obtained from the Welch-Satterthwaite approximation. }

\begin{table}[h]
\caption{CCCO $p$-values for the six WISDM activities. For each activity, the test examines whether the acceleration functions in the $x$- and $y$-directions are conditionally independent given the acceleration function in the $z$-direction. The column $n$ gives the number of users included in the analysis. The values $m=200$, $300$, and $400$ represent the first 10, 15, and 20 seconds of data recorded at 20 Hz. Panels (a) and (b) use the weighted sum of $\chi^2_{(1)}$ distributions and the Welch-Satterthwaite approximation, respectively. At the $0.05$ significance level, Sitting is significant in both panels, while Standing is significant only in panel (b).}
    \label{tab:tab-pval-app}
\centering
\footnotesize
\smallskip
\begin{subtable}{0.48\textwidth}
\centering
\caption{Results via weighted sum of $\chi^2_{(1)}$ distributions.}
\label{tab:tab-pval-app-raw}
\begin{tabular}{lcccc}
\hline
Activity & $n$ & $m=200$ & $m=300$ & $m=400$ \\
\hline
Walking & 36 & 0.1610 & 0.2016 & 0.2221\\
Jogging & 32 & 0.2121 & 0.2435 & 0.2519\\
Upstairs & 32 & 0.1393 & 0.1731 & 0.1974\\
Downstairs & 32 & 0.1951 & 0.2072 & 0.2053\\
Sitting & 23 & 0.0094 & 0.0097 & 0.0102\\
Standing & 24 & 0.0789 & 0.0836 & 0.0796\\
\hline
\end{tabular}
\end{subtable}
\hfill
\begin{subtable}{0.48\textwidth}
\centering
\caption{Results with Welch-Satterthwaite approximation.}
\label{tab:tab-pval-app-sw}
\begin{tabular}{lcccc}
\hline
Activity & $n$ & $m=200$ & $m=300$ & $m=400$ \\
\hline
Walking & 36 & 0.1407 & 0.1977 & 0.2376\\
Jogging & 32 & 0.1702 & 0.2090 & 0.2197\\
Upstairs & 32 & 0.0870 & 0.1285 & 0.1576\\
Downstairs & 32 & 0.1650 & 0.1816 & 0.1740\\
Sitting & 23 & 0.0003 & 0.0003 & 0.0004\\
Standing & 24 & 0.0340 & 0.0395 & 0.0359\\
\hline
\end{tabular}
\end{subtable}

\smallskip
\end{table}

{From  Table~\ref{tab:tab-pval-app-raw} we see that, at the significance level $\alpha = 0.05$,}   for all activities except sitting, the accelerations in the $x$- and $y$-directions are conditionally independent given the acceleration in the $z$-direction. For the sitting activity, however, the accelerations in the $x$- and $y$-directions are conditionally dependent given the acceleration in the $z$-direction, regardless of the time window used. {The results from the Welch-Satterthwaite approximation in Table~\ref{tab:tab-pval-app-sw} are slightly different as  conditional dependence is significant for the standing activity as well. This is reasonable since,  with only $n=24$ samples, the Welch-Satterthwaite approximation might not be very accurate, and the p-values are close to $\alpha=0.05$.}

Our statistical results are consistent with the physical analysis (Figure 2) of \cite{kwapisz2010activity}. For the first four activities (walking, jogging, upstairs, and downstairs), there are clear periodic patterns in all three directions, and the oscillations across directions exhibit noticeable synchronization. As a result, the accelerations in the $z$-direction largely capture the relationship between those in the $x$- and $y$-directions, making the conditional dependence between accelerations in the $x$- and $y$-directions given the $z$-direction very weak.

In contrast, for sitting and standing, the accelerations in the $z$-direction are nearly constant, as shown in Figure 2(e)-(f) of \cite{kwapisz2010activity}, and therefore provide little information about the $x$- and $y$-directions. However, Figure 2(e) indicates that during sitting the oscillations in the $x$- and $y$-directions remain highly synchronized, suggesting a strong dependency between them. By comparison, the pattern for standing in Figure 2(f) appears much weaker, indicating a substantially weaker dependence between the two directions.

\clearpage

\appendix

\renewcommand{\thesection}{\Alph{section}}
\renewcommand{\thesubsection}{\thesection.\arabic{subsection}}
\renewcommand{\thesubsubsection}{\thesection.\arabic{subsection}.\arabic{subsubsection}}

\renewcommand{\theequation}{S.\arabic{equation}}
\renewcommand{\thetable}{S.\arabic{table}}
\renewcommand{\thefigure}{S.\arabic{figure}}

\setcounter{equation}{0}
\setcounter{subsection}{0}
\setcounter{subsubsection}{0}
\setcounter{table}{0}
\setcounter{figure}{0}

\section{Illustrative Examples for Assumption \ref{ass-beta}}\label{app-example}

\begin{example}\em 
In this example we show that if $X \indep Z$ and $Y \indep Z$, then Assumption \ref{ass-beta} is satisfied for any $\beta > 0$.
 Since $Y \indep Z$, we have  $\theta_j \indep \xi_k$ for all $j,k$. Hence $E(\theta_j \xi_k) = E(\theta_j) E(\xi_k) = 0$ for all $j,k$.  So (\ref{eq-yz-beta}) is satisfied.  
 Since $X \indep Z$, we have 
  $\omega_j \indep \xi_k$ for all $j,k$, which implies that $E(\omega_j \xi_k) = E(\omega_j) E(\xi_k) = 0$ and $E(\omega_j \xi_k \xi_ \ell ) = E(\omega_j) E(\xi_k \xi_ \ell ) = 0$ for all $j,k, \ell $. So   \eqref{eq-xddotz-beta} is satisfied. 
\end{example}

\begin{example}\em 
In this example we show that, if 
 there exists an $n_0 \in \mathbb{N}$ such that 
\begin{align}\label{eq:tail independence}
\{\omega_j : j>n_0\} \indep Z, \quad 
X \indep \{\xi_k : k>n_0\}, \quad \mbox{and} \quad 
Y \indep \{\xi_k : k>n_0\},
\end{align}
and $\sup_{j \in \mathbb{N}}E|\omega_j|^3< \infty$, $\sup_{j \in \mathbb{N}}E|\xi_j|^3 < \infty$, then Assumption \ref{ass-beta} is satisfied for any $\beta > 0$. The first two independence conditions in \eqref{eq:tail independence} mean that the dependence between $X$ and $Z$ only happens in both low-frequency components, and the last condition in \eqref{eq:tail independence} means that $Y$ only depends on the low-frequency components in $Z$.

Since $Y \indep \{\xi_k : k>n_0\}$, we have $E(\theta_j \xi_k) = E(\theta_j) E(\xi_k) = 0$ for all $k>n_0$ and all $j$. Also note that for all $k \le n_0$ and all $j$, we have
$|E(\theta_j \xi_k)| \le \sqrt{E(\theta_j^2)E(\xi_k^2)} = 1$, so the supremum in \eqref{eq-yz-beta} is  
\begin{align*}
\sup_{j,k \in \mathbb{N}} \mu_j^{1/2} \gamma_k^{-1/2-\beta} |E(\theta_j \xi_k)|
=\sup_{j,k \in \mathbb{N},k \le n_0} \mu_j^{1/2} \gamma_k^{-1/2-\beta} |E(\theta_j \xi_k)|
\le \mu_1^{1/2} \gamma_{n_0}^{-1/2-\beta} < \infty.
\end{align*}

Since  $X \indep \{\xi_k : k>n_0\}$, we have $E(\omega_j \xi_k) = E(\omega_j) E(\xi_k) = 0$ for all $k>n_0$ and all $j$. Also note that for all $k \le n_0$ and all $j$, we have
$|E(\omega_j \xi_k)| \le \sqrt{E(\omega_j^2)E(\xi_k^2)} = 1$, so the second supremum in \eqref{eq-xddotz-beta} is
\begin{align*}
\sup_{j,k \in \mathbb{N}} \eta_j^{1/2} \gamma_k^{-1/2-\beta} |E(\omega_j \xi_k)| 
= \sup_{j,k \in \mathbb{N}, k \le n_0} \eta_j^{1/2} \gamma_k^{-1/2-\beta} |E(\omega_j \xi_k)| 
\le \eta_1^{1/2} \gamma_{n_0}^{-1/2-\beta} < \infty.
\end{align*}
Furthermore, since  $\{\omega_j : j>n_0\} \indep Z$,  we have $E(\omega_j \xi_k \xi_ \ell ) = E(\omega_j) E(\xi_k \xi_ \ell ) =0$ for all $j>n_0$ and all $k,\ell$. Also note that for all $j \le n_0$ and all $k,\ell$, we have
$|E(\omega_j\xi_k\xi_l)| \le \sqrt[3]{E|\omega_j|^3E|\xi_k|^3E|\xi_l|^3} < \infty$, so the first supremum in \eqref{eq-xddotz-beta} is
\begin{align*}
\sup_{j,k,\ell \in \mathbb{N}} \eta_j^{1/2} \gamma_k^{1/2} \gamma_ \ell ^{-1/2-\beta} |E(\omega_j \xi_k \xi_ \ell )|
&= \sup_{j,k,\ell \in \mathbb{N}, j \le n_0} \eta_j^{1/2} \gamma_k^{1/2} \gamma_ \ell ^{-1/2-\beta} |E(\omega_j \xi_k \xi_ \ell )|\\
&\preceq \eta_1^{1/2} \gamma_1^{1/2} \gamma_{n_0}^{-1/2-\beta}  
<\infty, 
\end{align*} 
as desired. 
\end{example}

\section{Technical Proofs}\label{sec-proof}

\subsection{Proof of Corollary \ref{cor-reg-form}}\label{cor-reg-form-proof}

\begin{proof}
    Again, we  will only show the first part of \eqref{eq-regression-form}. Take $f=\tilde{\kappa}_Z(\cdot,z)$ and $g=\kappa_X(x,\cdot)$ in the first equation in \eqref{eq-regression-inner}, and we have
\begin{align}\label{eq-regression-inter}
    \langle B_{X|Z} \tilde{\kappa}_Z(\cdot,z), \kappa_X(x,\cdot) \rangle_{\mathcal{H}_X} = \langle \tilde{\kappa}_Z(\cdot,z), E[\kappa_X(x,X) | Z=\cdot] \rangle_{\mathcal{H}_Z}.
\end{align}
The right-hand side of \eqref{eq-regression-inter} is 
\begin{align*}
\begin{split}
    \langle \tilde{\kappa}_Z(\cdot,z), E[\kappa_X(x,X) | Z=\cdot] \rangle_{\mathcal{H}_Z} 
    & = E[\kappa_X(x,X) | Z=z] - E \{ E[\kappa_X(x,X) | Z] \} \\
    & = E[\kappa_X(x,X) | Z=z] - E[\kappa_X(x,X)] \\
    & = E[\kappa_X(x,X) - \mu_X(x) | Z=z] \\
    & = E[\tilde{\kappa}_X(x,X) | Z=z].
\end{split}
\end{align*}
The left-hand side of \eqref{eq-regression-inter} is
\begin{align*}
    \langle B_{X|Z} \tilde{\kappa}_Z(\cdot,z), \kappa_X(x,\cdot) \rangle_{\mathcal{H}_X}
    = \langle B_{X|Z} \tilde{\kappa}_Z(\cdot,z), \kappa_X(\cdot,x) \rangle_{\mathcal{H}_X} 
    = [B_{X|Z} \tilde{\kappa}_Z(\cdot,z)](x).
\end{align*}
Equating the two sides of \eqref{eq-regression-inter} gives the desired result \eqref{eq-regression-form}. It remains to check that $f=\tilde{\kappa}_Z(\cdot,z) \in \ker(\Sigma_{ZZ})^\perp$. For any $h \in \ker(\Sigma_{ZZ})$, we have $\Sigma_{ZZ} h = 0$, so
\begin{align*}
    0 = \langle h, \Sigma_{ZZ} h \rangle_{\mathcal{H}_Z} = \mathrm{var}[h(Z)],
\end{align*}
which implies that $h$ is a constant function almost surely. Therefore, 
\begin{align*}
    \langle \tilde{\kappa}_Z(\cdot,z), h \rangle_{\mathcal{H}_Z} = h(z) - E[h(Z)] = 0
\end{align*}
almost surely. This implies that $\tilde{\kappa}_Z(\cdot,z) \in \ker(\Sigma_{ZZ})^\perp$.
\end{proof}

\subsection{Proof of Proposition \ref{prop-rep-cond-cov}}\label{prop-rep-cond-cov-proof}

\begin{proof}
We expand the right-hand side of \eqref{eq-rep-cond-cov-op} as follows.
\begin{align}\label{eq:4 terms}
\begin{split}
    &E\left[\left(\tilde{\kappa}_X(\cdot,X)-E\left[\tilde{\kappa}_X(\cdot,X)|Z\right]\right)\otimes\left(\tilde{\kappa}_Y(\cdot,Y)-E\left[\tilde{\kappa}_Y(\cdot,Y)|Z\right]\right)\right]\\
    =&E\left[\tilde{\kappa}_X(\cdot,X)\otimes\tilde{\kappa}_Y(\cdot,Y)\right]
    -E\left[E\left[\tilde{\kappa}_X(\cdot,X)|Z\right]\otimes\tilde{\kappa}_Y(\cdot,Y)\right]\\
    &-E\left[\tilde{\kappa}_X(\cdot,X)\otimes E\left[\tilde{\kappa}_Y(\cdot,Y)|Z\right]\right]
    +E\left[E\left[\tilde{\kappa}_X(\cdot,X)|Z\right]\otimes E\left[\tilde{\kappa}_Y(\cdot,Y)|Z\right]\right]
\end{split}
\end{align}
By \eqref{eq-def-cross-cov-op}, the first term on the right-hand side is $\Sigma \lo {XY}$. We now consider the other three terms separately. By Corollary \ref{cor-reg-form}, we have
\begin{align*}
E\left[\tilde{\kappa}_X(\cdot,X)|Z\right]=B_{X|Z}\tilde{\kappa}_Z(\cdot,Z)=\Sigma_{XZ}\Sigma_{ZZ}^{\dagger},\quad E\left[\tilde{\kappa}_Y(\cdot,Y)|Z\right]=B_{Y|Z}\tilde{\kappa}_Z(\cdot,Z)=\Sigma_{YZ}\Sigma_{ZZ}^{\dagger}.
\end{align*}
Hence the second term on the right-hand side of (\ref{eq:4 terms}) is 
\begin{align*}
\begin{split}
E\left[E\left[\tilde{\kappa}_X(\cdot,X)|Z\right]\otimes\tilde{\kappa}_Y(\cdot,Y)\right]=&E\left[B_{X|Z}\tilde{\kappa}_Z(\cdot,Z)\otimes\tilde{\kappa}_Y(\cdot,Y)\right]
=B_{X|Z}E\left[\tilde{\kappa}_Z(\cdot,Z)\otimes\tilde{\kappa}_Y(\cdot,Y)\right]\\
    =&\Sigma_{XZ}\Sigma_{ZZ}^{\dagger}\Sigma_{ZY}.
\end{split}
\end{align*}
By the law of total expectation, we have 
\begin{align*}
E\left[E\left[\tilde{\kappa}_X(\cdot,X)|Z\right]\otimes\tilde{\kappa}_Y(\cdot,Y)\right]
= & E\left[E\left[\tilde{\kappa}_X(\cdot,X)|Z\right]\otimes E\left[\tilde{\kappa}_Y(\cdot,Y)|Z\right]\right], \\
E\left[ \tilde{\kappa}_X(\cdot,X) \otimes E\left[\tilde{\kappa}_Y(\cdot,Y)|Z\right]\right] 
= & 
E\left[E\left[\tilde{\kappa}_X(\cdot,X)|Z\right]\otimes E\left[\tilde{\kappa}_Y(\cdot,Y)|Z\right]\right]. 
\end{align*}
Consequently, the four terms on the right-hand of (\ref{eq:4 terms}) are simply 
\begin{align*}
\Sigma_{XY}-\Sigma_{XZ}\Sigma_{ZZ}^{\dagger}\Sigma_{ZY}-\Sigma_{XZ}\Sigma_{ZZ}^{\dagger}\Sigma_{ZY}+\Sigma_{XZ}\Sigma_{ZZ}^{\dagger}\Sigma_{ZY} 
    =\Sigma_{XY|Z},
\end{align*}
as desired. 
\end{proof}

\subsection{Proof of Proposition \ref{proposition:equivalence}}\label{proposition:equivalence-proof}

\begin{proof}
{{\em 1.} {If $\mathrm{rank}(\Sigma_{ZZ}) = \infty$,} by \eqref{eq-eigen}, we have
\begin{align*}
\Sigma_{ZZ}^{-1-\beta} = \sum_{k=1}^\infty \gamma_k^{-1-\beta} (\phi_k \otimes \phi_k).
\end{align*}
By the Karhunen-Lo\'eve expansions of $\kappa_Z(\cdot,Z)$ and $\kappa_Y(\cdot,Y)$, 
\begin{align*}
\Sigma_{YZ}
= E[\tilde\kappa_Y(\cdot,Y) \otimes \tilde\kappa_Z(\cdot,Z)]
= \sum_{j=1}^\infty \sum_{k=1}^\infty \sqrt{\mu_j \gamma_k} E(\theta_j \xi_k) (\psi_j \otimes \phi_k).
\end{align*}
Hence, $\Sigma_{YZ} = S_{YZ} \Sigma_{ZZ}^{1+\beta}$ is equivalent to
\begin{align*}
S_{YZ} 
&= \Sigma_{YZ} \Sigma_{ZZ}^{-1-\beta}
= \left[\sum_{j=1}^\infty \sum_{k=1}^\infty \mu_j^{1/2} \gamma_k^{1/2} E(\theta_j \xi_k) (\psi_j \otimes \phi_k) \right]
\left[ \sum_{ \ell =1}^\infty \gamma_ \ell ^{-1-\beta} (\phi_ \ell  \otimes \phi_ \ell) \right]\\
&= \sum_{j=1}^\infty \sum_{k=1}^\infty \sum_{ \ell =1}^\infty \mu_j^{1/2} \gamma_k^{1/2} \gamma_ \ell ^{-1-\beta} E(\theta_j \xi_k) (\psi_j \otimes \phi_k)(\phi_ \ell  \otimes \phi_ \ell )\\
&= \sum_{j=1}^\infty \sum_{k=1}^\infty \mu_j^{1/2} \gamma_k^{-1/2-\beta} E(\theta_j \xi_k) (\psi_j \otimes \phi_k).
\end{align*}
Therefore, $S_{YZ}$ is bounded if and only if \eqref{eq-yz-beta} holds.}

{If $\mathrm{rank}(\Sigma_{ZZ}) = K < \infty$, then $\gamma_k=0$ for $k>K$. Let $\{\phi_{K+1},\phi_{K+2},\dots\}$ be a completion of $\{\phi_0,\phi_1,\dots,\phi_K\}$ as an orthonormal basis of $\mathcal{G}_Z$, where $\phi_0 = \mu_Z/\|\mu_Z\|_{\mathcal{G}_Z}$. Suppose that 
\begin{align*}
    S_{YZ} = \sum_{j=1}^\infty \sum_{k=0}^\infty c_{jk} (\psi_j \otimes \phi_k)
\end{align*}
is a bounded operator such that $\Sigma_{YZ} = S_{YZ} \Sigma_{ZZ}^{1+\beta}$. Then, we have
\begin{align*}
\sum_{j=1}^\infty \sum_{k=1}^K \mu_j^{1/2} \gamma_k^{1/2} E(\theta_j \xi_k) (\psi_j \otimes \phi_k)
&= \left[ \sum_{j=1}^\infty \sum_{k=0}^\infty c_{jk} (\psi_j \otimes \phi_k) \right] \left[ \sum_{k=1}^K \gamma_k^{1+\beta} (\phi_k \otimes \phi_k)\right]\\
&= \sum_{j=1}^\infty \sum_{k=0}^\infty \sum_{\ell=1}^K c_{jk} \gamma_\ell^{1+\beta} (\psi_j \otimes \phi_k)  (\phi_\ell \otimes \phi_\ell)\\
&= \sum_{j=1}^\infty \sum_{k=1}^K c_{jk} \gamma_k^{1+\beta} (\psi_j \otimes \phi_k),
\end{align*}
which indicates that 
\begin{align*}
c_{jk} = \mu_j^{1/2} \gamma_k^{-1/2-\beta} E(\theta_j \xi_k), \qquad \text{for all}\ j=1,2,\dots, \quad k=1,2,\dots,K.
\end{align*}
Thus, the boundedness of $S_{YZ}$ implies $\sup_{j\in \mathbb{N}, k \in \mathbb{N}_0} |c_{jk}| < \infty$ ($\mathbb N \lo 0$ being the set of nonnegative integers), which further implies \eqref{eq-yz-beta}. Conversely, if \eqref{eq-yz-beta} holds, we can set 
\begin{align*}
S_{YZ} = \sum_{j=1}^\infty \sum_{k=1}^K \mu_j^{1/2} \gamma_k^{-1/2-\beta} E(\theta_j \xi_k) (\psi_j \otimes \phi_k),
\end{align*}
which is clearly bounded due to \eqref{eq-yz-beta}, so that
\begin{align*}
S_{YZ} \Sigma_{ZZ}^{1+\beta} 
& = \left[ \sum_{j=1}^\infty \sum_{k=1}^K \mu_j^{1/2} \gamma_k^{-1/2-\beta} E(\theta_j \xi_k) (\psi_j \otimes \phi_k) \right] \left[ \sum_{k=1}^K \gamma_k^{1+\beta} (\phi_k \otimes \phi_k) \right] \\
& = \sum_{j=1}^\infty \sum_{k=1}^K \sum_{\ell=1}^K \mu_j^{1/2} \gamma_k^{-1/2-\beta} \gamma_\ell^{1+\beta} E(\theta_j \xi_k) (\psi_j \otimes \phi_k) (\phi_\ell \otimes \phi_\ell) \\
& = \sum_{j=1}^\infty \sum_{k=1}^K \mu_j^{1/2} \gamma_k^{1/2} E(\theta_j \xi_k) (\psi_j \otimes \phi_k)  
= \Sigma_{YZ}.
\end{align*}
}

{{\em 2.} {If $\mathrm{rank}(\Sigma_{ZZ}) = \infty$,} by the Karhunen-Lo\'eve expansions of $\kappa_Z(\cdot,Z)$ and $\kappa_X(\cdot,X)$, 
\begin{align*}
\Lambda_{\Ddot{X}Z}
&= E[\tilde\kappa_X(\cdot,X) \kappa_Z(\cdot,Z) \otimes \tilde\kappa_Z(\cdot,Z)]\\
&= E[\tilde\kappa_X(\cdot,X) \tilde\kappa_Z(\cdot,Z) \otimes \tilde\kappa_Z(\cdot,Z)]+ E[\tilde\kappa_X(\cdot,X) \mu_Z \otimes \tilde\kappa_Z(\cdot,Z)]\\
&= \sum_{j=1}^\infty \sum_{k=1}^\infty \sum_{ \ell =1}^\infty \sqrt{\eta_j \gamma_k \gamma_ \ell } E(\omega_j \xi_k \xi_ \ell ) [(\chi_j \phi_k) \otimes \phi_ l ]
+\sum_{j=1}^\infty \sum_{k=1}^\infty \sqrt{\eta_j \gamma_k} E(\omega_j \xi_k) [(\chi_j \mu_Z) \otimes \phi_k].
\end{align*}
Hence, $\Lambda_{\Ddot{X}Z} = S_{\Ddot{X}Z} \Sigma_{ZZ}^{1+\beta}$ is equivalent to
\begin{align*}
S_{\Ddot{X}Z} 
&= \Lambda_{\Ddot{X}Z} \Sigma_{ZZ}^{-1-\beta}\\
&= \left[\sum_{j=1}^\infty \sum_{k=1}^\infty \sum_{ \ell =1}^\infty \eta_j^{1/2} \gamma_k^{1/2} \gamma_ \ell ^{1/2} E(\omega_j \xi_k \xi_ \ell ) [( \chi_j \phi_k) \otimes \phi_ \ell]  \right]
\left[ \sum_{r=1}^\infty \gamma_r^{-1-\beta} (\phi_r \otimes \phi_r)\right]\\
&\quad + \left[\sum_{j=1}^\infty \sum_{k=1}^\infty \eta_j^{1/2} \gamma_k^{1/2} E(\omega_j \xi_k) [(\chi_j \mu_Z) \otimes \phi_k] \right]
\left[ \sum_{ \ell =1}^\infty \gamma_ \ell ^{-1-\beta} (\phi_ \ell  \otimes \phi_ \ell) \right]\\
&= \sum_{j=1}^\infty \sum_{k=1}^\infty \sum_{ \ell =1}^\infty \sum_{r=1}^\infty \eta_j^{1/2} \gamma_k^{1/2} \gamma_ \ell ^{1/2} \gamma_r^{-1-\beta} E(\omega_j \xi_k \xi_ \ell ) [(\chi_j \phi_k) \otimes \phi_ \ell ] (\phi_r \otimes \phi_r)\\
&\quad + \sum_{j=1}^\infty \sum_{k=1}^\infty \sum_{ \ell =1}^\infty \eta_j^{1/2} \gamma_k^{1/2} \gamma_ \ell ^{-1-\beta} E(\omega_j \xi_k) [(\chi_j \mu_Z) \otimes \phi_k] (\phi_ \ell  \otimes \phi_ \ell )\\
&= \sum_{j=1}^\infty \sum_{k=1}^\infty \sum_{ \ell =1}^\infty  \eta_j^{1/2} \gamma_k^{1/2} \gamma_ \ell ^{-1/2-\beta} E(\omega_j \xi_k \xi_ \ell ) [(\chi_j \phi_k) \otimes \phi_ \ell ]\\
&\quad + \sum_{j=1}^\infty \sum_{k=1}^\infty \eta_j^{1/2} \gamma_k^{-1/2-\beta} E(\omega_j \xi_k) [(\chi_j \mu_Z) \otimes \phi_k].
\end{align*}
Therefore, $S_{\Ddot{X}Z}$ is bounded if and only if \eqref{eq-xddotz-beta} holds.}

{If $\mathrm{rank}(\Sigma_{ZZ}) = K < \infty$, then $\gamma_k=0$ for $k>K$. Similarly, let $\{\phi_{K+1},\phi_{K+2},\dots\}$ be a completion of $\{\phi_0,\phi_1,\dots,\phi_K\}$ as an orthonormal basis of $\mathcal{G}_Z$, where $\phi_0 = \mu_Z/\|\mu_Z\|_{\mathcal{G}_Z}$. Suppose that 
\begin{align*}
    S_{\Ddot{X}Z} = \sum_{j=1}^\infty \sum_{k=0}^\infty \sum_{\ell=0}^\infty c_{jk\ell} [(\chi_j \phi_k) \otimes \phi_\ell]
\end{align*}
is a bounded operator such that $\Lambda_{\Ddot{X}Z} = S_{\Ddot{X}Z} \Sigma_{ZZ}^{1+\beta}$. Then, we have
\begin{align*}
&\sum_{j=1}^\infty \sum_{k=1}^K \sum_{ \ell =1}^K \eta_j^{1/2} \gamma_k^{1/2} \gamma_\ell^{1/2}  E(\omega_j \xi_k \xi_ \ell ) [(\chi_j \phi_k) \otimes \phi_ \ell ]
+\sum_{j=1}^\infty \sum_{k=1}^K \eta_j^{1/2} \gamma_k^{1/2} E(\omega_j \xi_k) [(\chi_j \mu_Z) \otimes \phi_k]\\
&= \left[ \sum_{j=1}^\infty \sum_{k=0}^\infty \sum_{\ell=0}^\infty c_{jk\ell} [(\chi_j \phi_k) \otimes \phi_\ell]\right] \left[ \sum_{k=1}^K \gamma_k^{1+\beta} (\phi_k \otimes \phi_k) \right] \\
&= \sum_{j=1}^\infty \sum_{k=0}^\infty \sum_{\ell=0}^\infty \sum_{r=1}^K c_{jk\ell} \gamma_r^{1+\beta} [(\chi_j \phi_k) \otimes \phi_\ell]  (\phi_r \otimes \phi_r) \\
&= \sum_{j=1}^\infty \sum_{k=0}^\infty \sum_{\ell=1}^K c_{jk\ell} \gamma_l^{1+\beta} [(\chi_j \phi_k) \otimes \phi_\ell] \\
&= \sum_{j=1}^\infty \sum_{k=1}^\infty \sum_{\ell=1}^K c_{jk\ell} \gamma_l^{1+\beta} [(\chi_j \phi_k) \otimes \phi_\ell]
+ \sum_{j=1}^\infty \sum_{k=1}^K c_{j0k} \|\mu_Z\|_{\mathcal G_Z}^{-1} \gamma_l^{1+\beta} [(\chi_j \mu_Z) \otimes \phi_k]
\end{align*}
which indicates that
\begin{align*}
& c_{jk\ell} = \eta_j^{1/2} \gamma_k^{1/2} \gamma_\ell^{-1/2-\beta}  E(\omega_j \xi_k \xi_ \ell ), \qquad \text{for all} \ j=1,2,\dots, \quad k,\ell=1,2,\dots,K, \quad \text{and}\\
& c_{j0k} = \eta_j^{1/2} \gamma_k^{-1/2-\beta} \|\mu_Z\|_{\mathcal G_Z} E(\omega_j \xi_k), \qquad \text{for all} \ j=1,2,\dots, \quad k=1,2,\dots,K.
\end{align*}
Thus, the boundedness of $S_{\Ddot{X}Z}$ implies $\sup_{j\in \mathbb{N}, k,l \in \mathbb{N}_0} |c_{jk\ell}| < \infty$, which further implies \eqref{eq-xddotz-beta}. Conversely, if \eqref{eq-xddotz-beta} holds, we can set 
\begin{align*}
S_{\Ddot{X}Z} 
&= \sum_{j=1}^\infty \sum_{k=1}^K \sum_{\ell=1}^K  \eta_j^{1/2} \gamma_k^{1/2} \gamma_ \ell ^{-1/2-\beta} E(\omega_j \xi_k \xi_ \ell ) [(\chi_j \phi_k) \otimes \phi_ \ell ] \\
& \quad + \sum_{j=1}^\infty \sum_{k=1}^K \eta_j^{1/2} \gamma_k^{-1/2-\beta} E(\omega_j \xi_k) [(\chi_j \mu_Z) \otimes \phi_k],
\end{align*}
which is clearly bounded due to \eqref{eq-xddotz-beta}, so that
\begin{align*}
S_{\Ddot{X}Z} \Sigma_{ZZ}^{1+\beta} 
& = \left[ \sum_{j=1}^\infty \sum_{k=1}^K \sum_{\ell=1}^K  \eta_j^{1/2} \gamma_k^{1/2} \gamma_ \ell ^{-1/2-\beta} E(\omega_j \xi_k \xi_ \ell ) [(\chi_j \phi_k) \otimes \phi_ \ell ]  \right] \left[ \sum_{k=1}^K \gamma_k^{1+\beta} (\phi_k \otimes \phi_k) \right] \\
& \quad + \left[ \sum_{j=1}^\infty \sum_{k=1}^K \eta_j^{1/2} \gamma_k^{-1/2-\beta} E(\omega_j \xi_k) [(\chi_j \mu_Z) \otimes \phi_k]  \right] \left[ \sum_{k=1}^K \gamma_k^{1+\beta} (\phi_k \otimes \phi_k) \right] \\
& = \sum_{j=1}^\infty \sum_{k=1}^K \sum_{\ell=1}^K  \sum_{r=1}^K \eta_j^{1/2} \gamma_k^{1/2} \gamma_ \ell ^{-1/2-\beta} \gamma_r^{1+\beta} E(\omega_j \xi_k \xi_ \ell ) [(\chi_j \phi_k) \otimes \phi_ \ell ]    (\phi_r \otimes \phi_r)  \\
& \quad + \sum_{j=1}^\infty \sum_{k=1}^K \sum_{\ell=1}^K \eta_j^{1/2} \gamma_k^{-1/2-\beta} \gamma_\ell^{1+\beta} E(\omega_j \xi_k) [(\chi_j \mu_Z) \otimes \phi_k]   (\phi_\ell \otimes \phi_\ell) \\
& = \sum_{j=1}^\infty \sum_{k=1}^K \sum_{\ell=1}^K  \eta_j^{1/2} \gamma_k^{1/2} \gamma_ \ell ^{1/2} E(\omega_j \xi_k \xi_ \ell ) [(\chi_j \phi_k) \otimes \phi_ \ell ]  \\
& \quad + \sum_{j=1}^\infty \sum_{k=1}^K\eta_j^{1/2} \gamma_k^{1/2}  E(\omega_j \xi_k) [(\chi_j \mu_Z) \otimes \phi_k]   \\
& = \Lambda_{\Ddot{X}Z}.
\end{align*}}
\end{proof}

\subsection{Proof of Lemma \ref{lem-alpha}}\label{lem-alpha-proof}

\begin{proof}
{Under Assumption \ref{ass-alpha}, there exists a constant $C_\alpha$ such that $\gamma_j \le C_\alpha j^{-\alpha}$. Let
$J_\epsilon=\left\lceil (C_\alpha/\epsilon_n)^{1/\alpha}\right\rceil$.
Then
\begin{align}\label{eq:two terms}
\sum_{j=1}^{\infty}
\frac{\gamma_j}{(\gamma_j+\epsilon_n)^2}
=
\sum_{j\le J_\epsilon}
\frac{\gamma_j}{(\gamma_j+\epsilon_n)^2}
+
\sum_{j>J_\epsilon}
\frac{\gamma_j}{(\gamma_j+\epsilon_n)^2}.
\end{align}
For the first term on the right-hand side of (\ref{eq:two terms}), we have
\begin{align*}
\sum_{j\le J_\epsilon}
\frac{\gamma_j}{(\gamma_j+\epsilon_n)^2}
\le
\sum_{j\le J_\epsilon}\frac{1}{\epsilon_n}
=
\frac{J_\epsilon}{\epsilon_n}
=
O\left(\epsilon_n^{-1-1/\alpha}\right)
=
O\left(\epsilon_n^{-(\alpha+1)/\alpha}\right).
\end{align*}
For the second term on the right-hand side of (\ref{eq:two terms}), since
\begin{align*}
\sum_{j>J_\epsilon}j^{-\alpha}
\le
\int_{J_\epsilon}^\infty x^{-\alpha} dx 
=
\frac{J_\epsilon^{1-\alpha}}{\alpha-1}
=
O\left( J_\epsilon^{1-\alpha}\right),
\end{align*}
we have
\begin{align*}
\sum_{j>J_\epsilon}
\frac{\gamma_j}{(\gamma_j+\epsilon_n)^2}
\le
\epsilon_n^{-2}
\sum_{j>J_\epsilon}\gamma_j
\le
C_\alpha\epsilon_n^{-2}
\sum_{j>J_\epsilon}j^{-\alpha}
=
O\left(\epsilon_n^{-2}J_\epsilon^{1-\alpha}\right)
=
O\left(\epsilon_n^{-1-1/\alpha}\right)
=
O\left(\epsilon_n^{-(\alpha+1)/\alpha}\right).
\end{align*}
Combining the two terms above gives the {first} rate.}

{For the second summation, we apply the same decomposition
\begin{align}\label{eq:two terms-2}
\sum_{j=1}^{\infty}
\frac{\gamma_j^2}{(\gamma_j+\epsilon_n)^2}
=
\sum_{j\le J_\epsilon}
\frac{\gamma_j^2}{(\gamma_j+\epsilon_n)^2}
+
\sum_{j>J_\epsilon}
\frac{\gamma_j^2}{(\gamma_j+\epsilon_n)^2}.
\end{align}
For the first term on the right-hand side of (\ref{eq:two terms-2}), we have
\begin{align*}
\sum_{j\le J_\epsilon}
\frac{\gamma_j^2}{(\gamma_j+\epsilon_n)^2}
\le
\sum_{j\le J_\epsilon}1
=
J_\epsilon
=
O\left(\epsilon_n^{-1/\alpha}\right).
\end{align*}
For the second term on the right-hand side of (\ref{eq:two terms-2}), since
\begin{align*}
\sum_{j>J_\epsilon}j^{-2\alpha}
\le
\int_{J_\epsilon}^\infty x^{-2\alpha} dx 
=
\frac{J_\epsilon^{1-2\alpha}}{2\alpha-1}
=
O\left( J_\epsilon^{1-2\alpha}\right),
\end{align*}
we have
\begin{align*}
\sum_{j>J_\epsilon}
\frac{\gamma_j^2}{(\gamma_j+\epsilon_n)^2}
\le
\epsilon_n^{-2}
\sum_{j>J_\epsilon}\gamma_j^2
\le
C_\alpha^2\epsilon_n^{-2}
\sum_{j>J_\epsilon}j^{-2\alpha}
=
O\left(\epsilon_n^{-2}J_\epsilon^{1-2\alpha}\right)
=
O\left(\epsilon_n^{-1/\alpha}\right).
\end{align*}
Combining the two terms above gives the second rate.}
\end{proof}

\subsection{Proof of Proposition \ref{prop-brownian}}\label{prop-brownian-proof}

\begin{proof}
{
For a Brownian bridge $Z$ on $[0,1]$, the covariance function is
$c(s,t) = \min(s,t) - st$ for all $s,t \in [0,1]$, and
its Karhunen-Lo\`eve expansion can be written as
\begin{align}\label{eq:kl-expansion-bb}
Z(t)
=
\sum_{m=1}^{\infty}
\sqrt{\rho_m}Z_m e_m(t),
\end{align}
where
\begin{align*}
\rho_m=\frac{1}{\pi^2m^2},
\quad
e_m(t)=\sqrt{2}\sin(m\pi t),
\quad m=1,2,\dots,
\end{align*}
and $Z_1,Z_2,\dots$ are i.i.d. standard normal random variables. See, for example, \cite{deheuvels2006karhunen,chigansky2020eigen} for detailed discussions of the Karhunen-Lo\`eve expansions of the Brownian bridge. 
The rest of the proof will be based on the map of $Z$ to its Karhunen-Lo\`eve coefficient sequence $\tilde Z = (Z_1,Z_2,\dots)$.} 

{\paragraph{Step 1: Construction of space of $\tilde Z$ and its relationship to $\mathcal H_Z$.} 
We first define the space where the Karhunen-Lo\`eve coefficient sequence $\tilde{Z}$ resides.
Note that $\{e_m: m=1,2,\dots\}$ forms a complete orthonormal basis of $L^2[0,1]$ (see, for example, Section 5.4 of \cite{vretblad2003fourier}). 
Let $\mathcal H_Z = L^2[0,1]$ and define another Hilbert space
\begin{align*}
\tilde{\mathcal H}_Z = \left\{ \tilde{z} = (z_1,z_2,\dots) \in \mathbb{R}^\infty: \sum_{m=1}^\infty \rho_m z_m^2 < \infty \right\},
\end{align*}
where the inner product in $\tilde{\mathcal H}_Z$ is defined as
\begin{align*}
\langle \tilde{z}, \tilde{z}' \rangle_{\tilde{\mathcal H}_Z} = \sum_{m=1}^\infty \rho_m z_m z_m', \quad \text{for all}\ \tilde{z}=(z_1,z_2,\dots), \ \tilde{z}'=(z_1',z_2',\dots) \in \tilde{\mathcal H}_Z.
\end{align*}
Define $T: \mathcal H_Z \to \tilde{\mathcal H}_Z$ by
$T z = (T \lo 1 z, T \lo 2 z, \ldots )$, where $T \lo m : \ca H \lo Z \to \real$ is the mapping 
\begin{align*}
  T_m z = \frac{\langle z, e_m \rangle_{\mathcal H_Z}}{\sqrt{\rho_m}}.
\end{align*}
Clearly, $T$ is linear and bijective, and 
\begin{align*}
\langle z, z' \rangle _{\mathcal H_Z} 
=
\sum_{m=1}^{\infty}
\rho_m (T_mz)(T_mz')
= \langle Tz, T z' \rangle _{\tilde{\mathcal H}_Z}.
\end{align*}
Therefore, $\mathcal H_Z$ and $\tilde{\mathcal H}_Z$ are isomorphic, and $T$ is an isomorphism. }

{\paragraph{Step 2:  Constructing RKHS's  on $\ca H \lo Z$ and  $\tilde H_Z$ and developing their relation.}
Define another positive-definite kernel function $\kappa: \tilde H_Z \times \tilde H_Z \to \mathbb{R}$ by
\begin{align*}
\kappa(\tilde z, \tilde z') = \exp \{ -\eta \|\tilde z-\tilde z'\|_{\tilde{\mathcal H}_Z}^2\}, \quad \text{for all} \ \tilde z, \tilde z' \in \tilde{\mathcal H}_Z.
\end{align*}
Then, we have
\begin{align*}
\kappa_Z(z,z') 
= \exp \{ -\eta \|z-z'\|_{\mathcal H_Z}^2\}
= \exp \{ -\eta \|Tz-Tz'\|_{\tilde{\mathcal H}_Z}^2\}
= \kappa(Tz,Tz').
\end{align*}
Let $\ca G \lo Z$ be the RKHS of real-valued functions on $\ca H \lo Z$ generated by $\ka_Z (z, z')$ and, similarly, let $\tilde{\mathcal G}_Z$ be the RKHS   on $\tilde{H}_Z$ generated  $\kappa (\tilde z, \tilde z')$. Let $U:\tilde{\mathcal G}_Z \to \mathcal G_Z$ be the map defined by $U (\tilde f)= \tilde f \circ T$.
Clearly, $U$ is linear and bijective. Furthermore, for all $z \in \mathcal H_Z$, we have
\begin{align*}
    U \kappa(\cdot,Tz) = \kappa (T\cdot, Tz) = \kappa_Z(\cdot, z).
\end{align*}
Therefore, for all $z,z' \in \mathcal H_Z$, using the reproducing properties of $\mathcal G_Z$ and $\tilde{\mathcal G}_Z$, we have
\begin{align*}
\langle U \kappa(\cdot,Tz), U \kappa(\cdot,Tz') \rangle_{\mathcal G_Z}
&= \langle \kappa_Z(\cdot,z), \kappa_Z(\cdot,z') \rangle_{\mathcal G_Z}
= \kappa_Z(z,z')\\
&= \kappa(Tz,Tz')
= \langle  \kappa(\cdot,Tz),  \kappa(\cdot,Tz') \rangle_{\tilde{\mathcal G}_Z}.
\end{align*}
Based on the definitions of $\mathcal G_Z$ and $\tilde{\mathcal G}_Z$, we know that for all $\tilde f,\tilde f' \in \tilde{\mathcal G}_Z$, we have
\begin{align*}
\langle U \tilde f, U \tilde f' \rangle_{\mathcal G_Z}
= \langle  \tilde f,  \tilde f' \rangle_{\tilde{\mathcal G}_Z}.
\end{align*}
Therefore, $\mathcal G_Z$ and $\tilde{\mathcal G}_Z$ are isomorphic, and $U$ is an isomorphism. }

{\paragraph{Step 3: Second moment operators in two RKHS's.}
Let $\tilde{M}_{ZZ}:\tilde{\mathcal G}_Z \to \tilde{\mathcal G}_Z$ and $M_{ZZ}:\mathcal G_Z \to \mathcal G_Z$ be the second moment operators 
\begin{align*}
\tilde{M}_{ZZ} = E_{\tilde{Z} \sim \mu}[\kappa(\cdot,\tilde{Z}) \otimes \kappa(\cdot, \tilde{Z})],\qquad
M_{ZZ} = E_{Z \sim P_Z}[\kappa_Z(\cdot,Z) \otimes \kappa_Z(\cdot, Z)],
\end{align*}
where $\mu = P \circ \tilde Z\inv$ is the distribution of a sequence of i.i.d. standard normal variables. Note that $Z \sim P_Z$ (i.e., $Z$ is a Brownian bridge) if and only if $TZ \sim \mu$ \citep{deheuvels2006karhunen,chigansky2020eigen}. }

{By construction, we know that, for all $f \in \mathcal{G}_Z$ and $\tilde f \in \tilde{\mathcal{G}}_Z$, we have
\begin{align*}
(M_{ZZ} f)(\cdot) 
= E_{Z\sim P_Z}[\kappa_Z(\cdot,Z) f(Z)],\qquad
(\tilde M_{ZZ} \tilde f) (\cdot) 
= E_{\tilde Z\sim \mu}[\kappa(\cdot,\tilde Z) \tilde f(\tilde Z)].
\end{align*}
Thus, for all $\tilde f \in \tilde{\mathcal G}_Z$, for all $z \in \mathcal H_Z$, we have
\begin{align*}
(M_{ZZ}U\tilde f)(z)
&= E_{Z\sim P_Z}[\kappa_Z(z,Z) (U\tilde f)(Z)]
= E_{Z\sim P_Z}[\kappa(Tz,TZ) \tilde f(TZ)]\\
&= E_{\tilde Z\sim \mu}[\kappa(Tz,\tilde Z) \tilde f(\tilde Z)]
= (\tilde M_{ZZ} \tilde f)(Tz) 
= (U \tilde M_{ZZ} \tilde f)(z).
\end{align*}
Therefore, we have $\tilde M_{ZZ} = U^{-1} M_{ZZ} U$, and since $U$ is a Hilbert space homomorphism, by Theorem 8.3 of \cite{conway2019course}, $\tilde M_{ZZ}$ and $M_{ZZ}$ have the same eigenvalues. 
Hence, to find the properties of eigenvalues of $M_{ZZ}$, it suffices to study the eigenvalues of $\tilde M_{ZZ}$. }

{\paragraph{Step 4: Mercer decomposition of $\kappa(\cdot,\cdot)$ under $L^2(\mu \times \mu)$.}
Since, for all $\tilde f \in \tilde{\mathcal G}_Z$, we have
\begin{align*}
\tilde M_{ZZ} \tilde f (\cdot) = \int \kappa(\cdot, \tilde{z}') \tilde f(\tilde{z}') d \mu(\tilde{z}'),
\end{align*}
we can view $\tilde M_{ZZ}$ as the integral operator on $\tilde{\mathcal G}_Z$ with $\kappa(\cdot, \cdot)$ as its kernel function. By Mercer's Theorem \citep{steinwart2012mercers}, to find the eigenvalues and eigenfunctions of $\tilde M_{ZZ}$, it suffices to find the Mercer decomposition of $\kappa(\cdot, \cdot)$.}

{Clearly, for all $\tilde{z} = (z_1,z_2,\dots) \in \tilde{\mathcal H}_Z$ and $\tilde{z}' = (z_1',z_2',\dots) \in \tilde{\mathcal H}_Z$, we have
\begin{align*}
\kappa(\tilde z, \tilde z') 
= \exp \{ -\eta \|\tilde z-\tilde z'\|_{\tilde{\mathcal H}_Z}^2\}
= \exp \left\{ -\eta \sum_{m=1}^\infty \rho_m (z_m - z_m')^2\right\}
= \prod_{m=1}^{\infty}
\kappa_m (z_m, z_m')
\end{align*}
where
\begin{align*}
\kappa_m(z_m,z_m') = \exp\{-\eta\rho_m(z_m-z_m')^2\}, \quad m=1,2,\dots.
\end{align*}
Since $\sum_{m=1}^{\infty}\rho_m<\infty$, the product above is well-defined
almost surely $\mu \times \mu$.
Let $\mathcal H_m$ be the RKHS with reproducing kernel $\kappa_m$.
Then, equation (2.4) of \cite{karvonen2019gaussian}
gives the Mercer decomposition of the Gaussian kernel as
\begin{align}\label{eq:kappam-decomp}
\kappa_m(z_m,z_m')=\sum_{j=1}^\infty \lambda_{m,j} \psi_{m,j}(z_m) \psi_{m,j}(z_m') 
\quad \text{in}\quad L^2(\Phi \times \Phi).
\end{align}
where $\Phi$ is the standard normal distributions and the eigenvalues have the form
\begin{align*}
\lambda_{m,j}=a_m b_m^j,
\qquad j=0,1,\ldots,
\end{align*}
with,
\begin{align}\label{eq:am-bm}
a_m=
\left(
\frac{1}{2(1+8\eta\rho_m)^{1/2}-1+2\eta\rho_m}
\right)^{1/2},
\qquad
b_m=
\frac{2\eta\rho_m}{2(1+8\eta\rho_m)^{1/2}-1+2\eta\rho_m}, 
\end{align}
The forms of the eigenfunctions $\psi_{m,j}$'s are omitted, and can be written in terms of the Hermite polynomials according to equation (2.5) of \cite{karvonen2019gaussian}.}

{Next, for any $M \ge 1$, let  
\begin{align}\label{eq:kappam-prod}
\kappa^{(M)}(\tilde{z},\tilde{z}')=\prod_{m=1}^{M}\kappa_m(z_m,z_m'). 
\end{align}
By \eqref{eq:kappam-decomp} and \eqref{eq:kappam-prod},
 for each $M$, we have
\begin{align*}
\kappa^{(M)}(\tilde{z},\tilde{z}')
=
\sum_{\nu_1,\ldots,\nu_M\ge0}
\left(\prod_{m=1}^{M}\lambda_{m,\nu_m}\right)
\prod_{m=1}^{M}
\psi_{m,\nu_m}(z_m)\psi_{m,\nu_m}(z_m'),
\end{align*}
with convergence in $L^2(\mu_M\times\mu_M)$ where $\mu_M= \Phi \times \dots \times \Phi$ is the $m$-fold product measure. 
Moreover, since $0\le \kappa^{(M)}\le1$ and
$\kappa^{(M)}(\tilde{z},\tilde{z}')\to \kappa(\tilde{z},\tilde{z}')$ almost surely $\mu \times \mu$, using dominated convergence theorem, we have
\begin{align*}
\|\kappa^{(M)}-\kappa\|_{L^2(\mu\otimes\mu)}\to0.
\end{align*}
Consequently, the Mercer's decomposition of $\kappa(\tilde{z},\tilde{z}')$ is
\begin{align}\label{eq:kappa-decomp}
\kappa(\tilde{z},\tilde{z}')
=
\sum_{\nu\in\mathcal I}
\lambda_\nu \hi *
 \Psi_\nu \hi * (\tilde{z}) \Psi_\nu \hi *(\tilde{z}')
\quad\text{in }L^2(\mu\times\mu),
\end{align}
where $\ca I$ is the collection of all sequences,  whose entries are nonnegative integers, and with at most  finitely many nonzero entries, 
and
\begin{align*}
\lambda_\nu \hi * =\prod_{m=1}^{\infty}\lambda_{m,\nu_m}, \qquad
\Psi_\nu \hi * (\tilde{z})=\prod_{m=1}^{\infty}\psi_{m,\nu_m}(z_m).
\end{align*}
Since $\{\psi_{m,j}: j=1,2,\dots\}$ is an orthonormal basis of $L^2(\Phi)$, we know that $\{\Psi_\nu \hi *: \nu \in \mathcal I\}$ is an orthonormal basis of $L^2(\mu)$. Therefore, $\lambda_\nu$ for $\nu \in \mathcal I$ are all the eigenvalues of $\tilde{M}_{ZZ}$. }

{\paragraph{Step 5: Order of magnitude of the eigenvalues of $\tilde M_{ZZ}$.}
We first analyze the order of magnitude of the quantities $a \lo m$ and $b \lo m$ in \eqref{eq:am-bm}.
By
\begin{align*}
(1+8\eta\rho_m)^{1/2} \le 1 + 4\eta\rho_m,
\end{align*}
we have
\begin{align*}
a_m\ge
\left(
\frac{1}{1+10\eta\rho_m}
\right)^{1/2}.
\end{align*}
By
\begin{align*}
    \sum_{m=1}^{\infty}\rho_m = \frac{1}{\pi^2} \sum_{m=1}^{\infty} \frac{1}{m^2}<\infty,
\end{align*} 
we have
\begin{align*}
\sum_{m=1}^{\infty}|\log a_m|
\le \frac{1}{2} \sum_{m=1}^{\infty} \log (1+10\eta\rho_m)
\le 5 \eta \sum_{m=1}^{\infty} \rho_m <\infty.
\end{align*}
Hence,
\begin{align*}
0<\prod_{m=1}^{\infty}a_m<\infty.
\end{align*}
Also, it is obvious that $2(1+8\eta\rho_m)^{1/2}-1+2\eta\rho_m \ge 1$, so
\begin{align*}
b_m \le 2 \eta \rho_m.
\end{align*}
Applying this, and invoking $\eta \le \pi^2/8$, we have
\begin{align}\label{eq-lambda-m-j}
\lambda_{m,j} \le a_m (2\eta \rho_m)^j 
= a_m (2\eta/\pi^2)^j m^{-2j}
\le a_m (m+1)^{-2j},
\end{align}
and consequently, for $\nu \in \mathcal I$,  
\begin{align*}
\lambda_\nu \hi * \le C \prod_{m=1}^{\infty}(m+1)^{-2\nu_m}
\end{align*}
where $C = \prod \lo {m=1} \hi \infty a \lo m$. }

{Fix an $s>1/2$. Then
\begin{align*}
\sum_{\nu \in \mathcal I}(\lambda_\nu \hi *)^s
\le
C^s
\sum_{\nu \in \mathcal I}
\prod_{m=1}^{\infty}(m+1)^{-2s\nu_m}.
\end{align*}
Note that, if $\nu \in \mathbb{N}_0^\infty \setminus \mathcal I$, then for any $m_0>0$, there exists $r>m_0$ such that $\nu_r \ge 1$,  which implies that the product $\prod_{m=1}^\infty (m+1)^{-2s\nu_m} \le (r+1)^{-2s \nu_r} \le m_0^{-2s}$. By taking   $\lim \lo {m \lo 0 \to \infty}$ on the right, we see that $\prod_{m=1}^\infty (m+1)^{-2s\nu_m}=0$. Therefore,
\begin{align*}
\sum_{\nu \in \mathcal I}
\prod_{m=1}^{\infty}(m+1)^{-2s\nu_m}
&=
\sum_{\nu \in \mathbb{N}_0^\infty}
\prod_{m=1}^{\infty}(m+1)^{-2s\nu_m}\\
&=
\prod_{m=1}^{\infty}
\left\{
\sum_{r=0}^{\infty}(m+1)^{-2sr}
\right\}
=
\prod_{m=1}^{\infty}
\frac{1}{1-(m+1)^{-2s}}.
\end{align*}
The product on the right-hand side is finite because, when $s>1/2$, we have $\sum_{m=1}^{\infty}(m+1)^{-2s}<\infty$ and $0<(m+1)^{-2sr}<1$ for all $m$ and $r$, so by Corollary 42.5 of \cite{agarwal2011introduction}, $\prod_{m=1}^{\infty}[1-(m+1)^{-2s}]>0$. Thus
\begin{align}\label{eq-lambda-nu-star-summation}
\sum_{\nu \in \mathcal I}(\lambda_\nu \hi *)^s<\infty.
\end{align}
Let $\lambda_1\ge\lambda_2\ge\cdots$ be the decreasing rearrangement of
$\{\lambda_\nu \hi * : \nu  \in \mathcal I\}$.  Then
\begin{align*}
j\lambda_j^s
\le
\sum_{ \ell =1}^{j}\lambda_ \ell ^s
\le
\sum_{ \ell =1}^{\infty}\lambda_ \ell ^s
<\infty,
\end{align*}
so
\begin{align*}
\lambda_j\le C_{1/s} j^{-1/s},
\qquad \text{where} \
C_{1/s} = \left( \sum_{ \ell =1}^{\infty}\lambda_ \ell ^s \right)^{1/s}.
\end{align*}
Since $s>1/2$ is arbitrary, this gives
\begin{align*}
\lambda_j \le C_\alpha j^{-\alpha}
\qquad
\text{for every }1<\alpha<2,
\qquad j=1,2,\dots.
\end{align*}}

{\paragraph{Step 6: Relationship between $\gamma_j$ and $\lambda_j$.}
By Corollary 12.3 of \cite{schmudgen2012unbounded}, since $\Sigma_{ZZ} \preceq M_{ZZ}$ in the operator sense (i.e., $\langle f, \Sigma_{ZZ}f \rangle_{\mathcal H_Z} \le \langle f, M_{ZZ}f \rangle_{\mathcal H_Z}$ for all $f \in \mathcal H_Z$), for all $j=1,2,\dots$, we have $\gamma_j \le \lambda_j$, where $\gamma_j$, $j=1,2,\dots$, are eigenvalues of $\Sigma_{ZZ}$. Thus, we have
\begin{align*}
\gamma_j \le \lambda_j \le C_\alpha j^{-\alpha}
\qquad
\text{for every }1<\alpha<2,
\qquad j=1,2,\dots, 
\end{align*}
ad desired.  }

{
\paragraph{General case.}
Regarding the general case where $\rho_m \preceq m^{-r}$ for some $r>1$, we only need to make some changes in Step 5. Specifically, the exponent $2$ should be replaced by $r$ starting from \eqref{eq-lambda-m-j}. 
For example, \eqref{eq-lambda-m-j} becomes
\begin{align*}
\lambda_{m,j} \le a_m (2\eta \rho_m)^j 
= a_m (2\eta/\pi^2)^j m^{-rj}
\le a_m (m+1)^{-rj}.
\end{align*}
In the subsequent analysis, the condition $s>1/2$ can be replaced by $s>1/r$, and the essential convergence result \eqref{eq-lambda-nu-star-summation} still holds. The rest of the argument is unchanged, which gives $\lambda_j \preceq j^{-1/s}$, and for any $1<\alpha<r$, choosing $1/r<s<1/\alpha$ gives $\lambda_j \preceq j^{-\alpha}$.
}
\end{proof}

\subsection{Proof of Lemma \ref{lem-R-rate}}\label{lem-R-rate-proof}

\begin{proof}
{
By definition \eqref{eq-r-def}, we have
\begin{align*}
R 
= \mu_X\hat{\Sigma}_{ZZ}(\hat{\Sigma}_{ZZ}+\epsilon_nI)^{-1} - \mu_X
= \epsilon_n \mu_X (\hat{\Sigma}_{ZZ}+\epsilon_nI)^{-1}.
\end{align*}
Therefore, 
\begin{align*}
\|R\|_{\mathrm{OP}} 
& = \epsilon_n \| \mu_X (\hat{\Sigma}_{ZZ}+\epsilon_nI)^{-1}\|_{\mathrm{OP}} 
\le \epsilon_n \|\mu_X\|_{\mathcal G_X} \|(\hat{\Sigma}_{ZZ}+\epsilon_nI)^{-1}\|_{\mathrm{OP}} \\
& \le \epsilon_n \|\mu_X\|_{\mathcal G_X} \epsilon_n^{-1} 
= \|\mu_X\|_{\mathcal G_X},
\end{align*}
which proves assertion 1.  }

{
To prove assertion 2, we decompose $R$ into  $R_1+R_2$, where
\begin{align}\label{eq:r-r1-r2-decomp}
R_1 = \epsilon_n \mu_X [(\hat{\Sigma}_{ZZ}+\epsilon_nI)^{-1} - (\Sigma_{ZZ}+\epsilon_nI)^{-1}], \qquad
R_2 = \epsilon_n \mu_X(\Sigma_{ZZ}+\epsilon_nI)^{-1}.
\end{align}
For any $\theta>0$, we have
\begin{align*}
R_1 \Sigma_{ZZ}^\theta 
&= \epsilon_n \mu_X [(\hat{\Sigma}_{ZZ}+\epsilon_nI)^{-1} - (\Sigma_{ZZ}+\epsilon_nI)^{-1}] \Sigma_{ZZ}^\theta \\
&= \epsilon_n \mu_X (\hat{\Sigma}_{ZZ}+\epsilon_nI)^{-1} (\Sigma_{ZZ} - \hat{\Sigma}_{ZZ}) (\Sigma_{ZZ}+\epsilon_nI)^{-1} \Sigma_{ZZ}^\theta.
\end{align*}
Thus,
\begin{align}\label{eq:r1-sigma-theta-op}
\|R_1 \Sigma_{ZZ}^\theta\|_{\mathrm{OP}}
&\le \epsilon_n \|\mu_X (\hat{\Sigma}_{ZZ}+\epsilon_nI)^{-1}\|_{\mathrm{OP}} \|\Sigma_{ZZ} - \hat{\Sigma}_{ZZ}\|_{\mathrm{OP}} \|(\Sigma_{ZZ}+\epsilon_nI)^{-1} \Sigma_{ZZ}^\theta\|_{\mathrm{OP}}.
\end{align}
In the right-hand side above, $\|\hat{\Sigma}_{ZZ}-\Sigma_{ZZ}\|_{\mathrm{OP}} = O_P(n^{-1/2})$ by Hilbert space central limit theorem, 
\begin{align*}
\|\mu_X(\hat{\Sigma}_{ZZ}+\epsilon_nI)^{-1}\|_{\mathrm{OP}} \le \epsilon_n^{-1}\|\mu_X I\|_{\mathrm{OP}} = \epsilon_n^{-1}\|\mu_X\|_{\mathcal H_X}
= O(\epsilon_n^{-1}), 
\end{align*} 
and 
\begin{align}\label{eq-r1-theta-3}
\|(\Sigma_{ZZ}+\epsilon_nI)^{-1}\Sigma_{ZZ}^\theta\|_{\mathrm{OP}}
\le \|(\Sigma_{ZZ}+\epsilon_nI)^{-1+\theta}\|_{\mathrm{OP}}
= O(\epsilon_n^{\theta\wedge 1 -1}).
\end{align}
Therefore, 
\begin{align*}
\|R_1 \Sigma_{ZZ}^\theta\|_{\mathrm{OP}}
= \epsilon_n  O(\epsilon_n^{-1}) O_P(n^{-1/2}) O(\epsilon_n^{\theta\wedge 1 -1})
= O_P(n^{-1/2}\epsilon_n^{\theta\wedge 1 -1}).
\end{align*}
Also, by \eqref{eq-r1-theta-3}, we have
\begin{align*}
\|R_2 \Sigma_{ZZ}^\theta\|_{\mathrm{OP}}
= \epsilon_n  \|\mu_X\|_{\mathcal H_X} \|(\Sigma_{ZZ}+\epsilon_nI)^{-1}\Sigma_{ZZ}^\theta\|_{\mathrm{OP}}
= \epsilon_n O_P(\epsilon_n^{\theta\wedge 1 -1})
= O_P(\epsilon_n^{\theta\wedge 1}).
\end{align*}
Hence, 
\begin{align*}
\|R \Sigma_{ZZ}^\theta\|_{\mathrm{OP}}
\le \|R_1 \Sigma_{ZZ}^\theta\|_{\mathrm{OP}} + \|R_2 \Sigma_{ZZ}^\theta\|_{\mathrm{OP}} 
= O_P(n^{-1/2}\epsilon_n^{\theta\wedge 1 -1} + \epsilon_n^{\theta\wedge 1}),
\end{align*}
which gives the desired result.}

{To prove assertion 3, we use the same decomposition $R=R_1+R_2$ as in \eqref{eq:r-r1-r2-decomp}. Similar to \eqref{eq:r1-sigma-theta-op}, we have
\begin{align*}
\|R_1 \Sigma_{ZZ}\|_{\mathrm{HS}}
&\le \epsilon_n \|\mu_X (\hat{\Sigma}_{ZZ}+\epsilon_nI)^{-1}\|_{\mathrm{OP}} \|\Sigma_{ZZ} - \hat{\Sigma}_{ZZ}\|_{\mathrm{HS}} \|(\Sigma_{ZZ}+\epsilon_nI)^{-1} \Sigma_{ZZ}\|_{\mathrm{OP}}\\
&=\epsilon_n O(\epsilon_n^{-1}) O_P(n^{-1/2}) O(1) = O_P(n^{-1/2})
\end{align*}
since $\|\hat{\Sigma}_{ZZ}-\Sigma_{ZZ}\|_{\mathrm{HS}} = O_P(n^{-1/2})$ by Hilbert space central limit theorem. Different from assertion 2, regarding $R_2\Sigma_{ZZ}$, we have
\begin{align*}
\|R_2\Sigma_{ZZ}\|_{\mathrm{HS}}
= \epsilon_n  \|\mu_X\|_{\mathcal H_X} \|(\Sigma_{ZZ}+\epsilon_nI)^{-1}\Sigma_{ZZ}\|_{\mathrm{HS}}
= \epsilon_n O_P(\epsilon_n^{-1/(2\alpha)})
= O_P(\epsilon_n^{1-1/(2\alpha)}),
\end{align*}
where the second equality holds because, under Assumption \ref{ass-alpha}, by Lemma \ref{lem-alpha}, 
\begin{align*}
\|(\Sigma_{ZZ}+\epsilon_nI)^{-1}\Sigma_{ZZ}\|_{\mathrm{HS}}^2
= \sum_{j=1}^\infty \left(\frac{\gamma_j}{\gamma_j+\epsilon_n}\right)^2
=O(\epsilon_n^{-1/\alpha}).
\end{align*}
Hence, 
\begin{align*}
\|R \Sigma_{ZZ}\|_{\mathrm{HS}}
\le \|R_1 \Sigma_{ZZ}\|_{\mathrm{HS}} + \|R_2 \Sigma_{ZZ}\|_{\mathrm{HS}} 
= O_P(n^{-1/2} + \epsilon_n^{1-1/(2\alpha)}),
\end{align*}
which gives the desired result.}
\end{proof}

\subsection{Proof of Lemma \ref{lemma-regul-conv}}

\begin{proof}
{
We first decompose $\hat{B}_{\Ddot{X}|Z}-B_{\Ddot{X}|Z}-R$ as follows:
\begin{align*}
&\hat{B}_{\Ddot{X}|Z}-B_{\Ddot{X}|Z}-R\\
&= \hat{\Sigma}_{\Ddot{X}Z}(\hat{\Sigma}_{ZZ}+\epsilon_nI)^{-1}-\Sigma_{\Ddot{X}Z}\Sigma_{ZZ}^{\dagger} - \mu_X\hat{\Sigma}_{ZZ}(\hat{\Sigma}_{ZZ}+\epsilon_nI)^{-1} + \mu_X\\
&= \hat{\Sigma}_{\Ddot{X}Z}(\hat{\Sigma}_{ZZ}+\epsilon_nI)^{-1}-\Sigma_{\Ddot{X}Z}\Sigma_{ZZ}^{\dagger}\hat{\Sigma}_{ZZ}(\hat{\Sigma}_{ZZ}+\epsilon_nI)^{-1}\\
&\quad +\Sigma_{\Ddot{X}Z}\Sigma_{ZZ}^{\dagger}\hat{\Sigma}_{ZZ}(\hat{\Sigma}_{ZZ}+\epsilon_nI)^{-1}- \mu_X\hat{\Sigma}_{ZZ}(\hat{\Sigma}_{ZZ}+\epsilon_nI)^{-1} - \Sigma_{\Ddot{X}Z}\Sigma_{ZZ}^{\dagger}
+\mu_X\\
&= \hat{\Sigma}_{\Ddot{X}Z}(\hat{\Sigma}_{ZZ}+\epsilon_nI)^{-1}-\Sigma_{\Ddot{X}Z}\Sigma_{ZZ}^{\dagger}\hat{\Sigma}_{ZZ}(\hat{\Sigma}_{ZZ}+\epsilon_nI)^{-1}\\
& \quad + \Lambda_{\Ddot{X}Z}\Sigma_{ZZ}^{\dagger}\hat{\Sigma}_{ZZ}(\hat{\Sigma}_{ZZ}+\epsilon_nI)^{-1} - \Lambda_{\Ddot{X}Z}\Sigma_{ZZ}^{\dagger}\\
&= \left[ \Lambda_{\Ddot{X}Z} (\Sigma_{ZZ}+\epsilon_nI)^{-1} - \Lambda_{\Ddot{X}Z} \Sigma_{ZZ}^{\dagger}\right]
+ \left[\Lambda_{\Ddot{X}Z}\Sigma_{ZZ}^{\dagger}\hat{\Sigma}_{ZZ}(\hat{\Sigma}_{ZZ}+\epsilon_nI)^{-1} - \Lambda_{\Ddot{X}Z} (\Sigma_{ZZ}+\epsilon_nI)^{-1} \right] \\
&\quad  + \left[ (\hat{\Sigma}_{\Ddot{X}Z} -\Sigma_{\Ddot{X}Z}\Sigma_{ZZ}^{\dagger}\hat{\Sigma}_{ZZ} -\tilde{\Sigma}_{\Ddot{X}Z|Z})(\hat{\Sigma}_{ZZ}+\epsilon_nI)^{-1} \right] \\
&\quad
+ \left[ \tilde{\Sigma}_{\Ddot{X}Z|Z}[(\hat{\Sigma}_{ZZ}+\epsilon_nI)^{-1} -  (\Sigma_{ZZ}+\epsilon_nI)^{-1}] \right]  
+ \left[ \tilde{\Sigma}_{\Ddot{X}Z|Z}  (\Sigma_{ZZ}+\epsilon_nI)^{-1} \right],
\end{align*}
where
$\tilde{\Sigma}_{\Ddot{X}Z|Z} = E_n[\tilde{\kappa}_{\Ddot{X}|Z}(\cdot,\Ddot{X}|Z) \otimes \tilde{\kappa}_Z(\cdot,Z)]$}
and $\tilde{\kappa}_{\Ddot{X}|Z}$ is defined by \eqref{eq-k-tilde-xz}.
{By the triangular inequality, 
\begin{align}\label{eq-decomp-bhat-b}
\begin{split}
&\|\hat{B}_{\Ddot{X}|Z}-B_{\Ddot{X}|Z}-R\|_{\mathrm{OP}}\\
&\le\|\Lambda_{\Ddot{X}Z} (\Sigma_{ZZ}+\epsilon_nI)^{-1} - \Lambda_{\Ddot{X}Z} \Sigma_{ZZ}^{\dagger}\|_{\mathrm{OP}}\\
& \quad+ \|\Lambda_{\Ddot{X}Z}\Sigma_{ZZ}^{\dagger}\hat{\Sigma}_{ZZ}(\hat{\Sigma}_{ZZ}+\epsilon_nI)^{-1} - \Lambda_{\Ddot{X}Z} (\Sigma_{ZZ}+\epsilon_nI)^{-1}\|_{\mathrm{OP}} \\
& \quad + \| (\hat{\Sigma}_{\Ddot{X}Z} -\Sigma_{\Ddot{X}Z}\Sigma_{ZZ}^{\dagger}\hat{\Sigma}_{ZZ} -\tilde{\Sigma}_{\Ddot{X}Z|Z})(\hat{\Sigma}_{ZZ}+\epsilon_nI)^{-1} \|_{\mathrm{OP}} \\
& \quad + \| \tilde{\Sigma}_{\Ddot{X}Z|Z}[(\hat{\Sigma}_{ZZ}+\epsilon_nI)^{-1} -  (\Sigma_{ZZ}+\epsilon_nI)^{-1}] \|_{\mathrm{OP}} 
      + \| \tilde{\Sigma}_{\Ddot{X}Z|Z}  (\Sigma_{ZZ}+\epsilon_nI)^{-1} \|_{\mathrm{OP}}.
\end{split}
\end{align}}

We next analyze the five terms in \eqref{eq-decomp-bhat-b} one by one. For the second term of \eqref{eq-decomp-bhat-b}, since
\begin{align*}
    &{\Lambda_{\Ddot{X}Z}}\Sigma_{ZZ}^{\dagger}\hat{\Sigma}_{ZZ}(\hat{\Sigma}_{ZZ}+\epsilon_nI)^{-1} - {\Lambda_{\Ddot{X}Z}} (\Sigma_{ZZ}+\epsilon_nI)^{-1}\\
    =&{\Lambda_{\Ddot{X}Z}}\Sigma_{ZZ}^{\dagger}\hat{\Sigma}_{ZZ}(\hat{\Sigma}_{ZZ}+\epsilon_nI)^{-1} - {\Lambda_{\Ddot{X}Z}}\Sigma_{ZZ}^{\dagger}\Sigma_{ZZ}(\Sigma_{ZZ}+\epsilon_nI)^{-1}\\
    =&{\Lambda_{\Ddot{X}Z}}\Sigma_{ZZ}^{\dagger}(\Sigma_{ZZ}+\epsilon_nI)^{-1}\left[(\Sigma_{ZZ}+\epsilon_nI)\hat{\Sigma}_{ZZ} - \Sigma_{ZZ}(\hat{\Sigma}_{ZZ}+\epsilon_nI)\right](\hat{\Sigma}_{ZZ}+\epsilon_nI)^{-1}\\
    =&\epsilon_n{\Lambda_{\Ddot{X}Z}}\Sigma_{ZZ}^{\dagger}(\Sigma_{ZZ}+\epsilon_nI)^{-1}(\hat{\Sigma}_{ZZ}-\Sigma_{ZZ})(\hat{\Sigma}_{ZZ}+\epsilon_nI)^{-1}, 
\end{align*}
we have
\begin{align*}
    &\|{\Lambda_{\Ddot{X}Z}}\Sigma_{ZZ}^{\dagger}\hat{\Sigma}_{ZZ}(\hat{\Sigma}_{ZZ}+\epsilon_nI)^{-1} - {\Lambda_{\Ddot{X}Z}} (\Sigma_{ZZ}+\epsilon_nI)^{-1}\|_{\mathrm{OP}}\\
    \le&\epsilon_n\|{\Lambda_{\Ddot{X}Z}}\Sigma_{ZZ}^{\dagger}(\Sigma_{ZZ}+\epsilon_nI)^{-1}\|_{\mathrm{OP}}\|\hat{\Sigma}_{ZZ}-\Sigma_{ZZ}\|_{\mathrm{OP}}\|(\hat{\Sigma}_{ZZ}+\epsilon_nI)^{-1}\|_{\mathrm{OP}}.
\end{align*}
{By Hilbert-space central limit theorem, $\|\hat{\Sigma}_{ZZ}-\Sigma_{ZZ}\|_{\mathrm{OP}} = O_P(n^{-1/2})$. Also,   
\begin{align*}
    \|(\hat{\Sigma}_{ZZ}+\epsilon_nI)^{-1}\|_{\mathrm{OP}} \le \epsilon_n^{-1}\|I\|_{\mathrm{OP}} = \epsilon_n^{-1}.
\end{align*} 
Regarding  $\|{\Lambda_{\Ddot{X}Z}}\Sigma_{ZZ}^{\dagger}(\Sigma_{ZZ}+\epsilon_nI)^{-1}\|_{\mathrm{OP}}$, by ${\Lambda_{\Ddot{X}Z}}=S_{\Ddot{X}Z}\Sigma_{ZZ}^{1+\beta}$ in Assumption \ref{ass-beta},} we have
\begin{align}\label{eq-term2-med}
\begin{split}
    \|{\Lambda_{\Ddot{X}Z}}\Sigma_{ZZ}^{\dagger}(\Sigma_{ZZ}+\epsilon_nI)^{-1}\|_{\mathrm{OP}}
    =& \|S_{\Ddot{X}Z}\Sigma_{ZZ}^{1+\beta}\Sigma_{ZZ}^{\dagger}(\Sigma_{ZZ}+\epsilon_nI)^{-1}\|_{\mathrm{OP}}\\
    \le& \|S_{\Ddot{X}Z}\|_{\mathrm{OP}} \|\Sigma_{ZZ}^{\beta}(\Sigma_{ZZ}+\epsilon_nI)^{-1}\|_{\mathrm{OP}}\\
    \le& \|S_{\Ddot{X}Z}\|_{\mathrm{OP}} \|(\Sigma_{ZZ}+\epsilon_nI)^{-1+\beta}\|_{\mathrm{OP}}
    = O(\epsilon_n^{\beta\wedge 1 -1}).
\end{split}
\end{align}
{Hence  the second term in \eqref{eq-decomp-bhat-b} has the order of magnitude}
\begin{align}\label{eq-term2-rate}
\begin{split}
    \|{\Lambda_{\Ddot{X}Z}}\Sigma_{ZZ}^{\dagger}\hat{\Sigma}_{ZZ}(\hat{\Sigma}_{ZZ}+\epsilon_nI)^{-1} - {\Lambda_{\Ddot{X}Z}} (\Sigma_{ZZ}+\epsilon_nI)^{-1}\|_{\mathrm{OP}}
    =&\epsilon_n O(\epsilon_n^{\beta\wedge 1 -1}) O_P(n^{-1/2}) O_P(\epsilon_n^{-1})\\
    =&O_P(n^{-1/2} \epsilon_n^{\beta\wedge 1 -1}).
\end{split}
\end{align}

For the first term in \eqref{eq-decomp-bhat-b}, we have
\begin{align*}
    {\Lambda_{\Ddot{X}Z}} (\Sigma_{ZZ}+\epsilon_nI)^{-1} - {\Lambda_{\Ddot{X}Z}} \Sigma_{ZZ}^{\dagger}
    =& {\Lambda_{\Ddot{X}Z}}  \Sigma_{ZZ}^{\dagger}[\Sigma_{ZZ} - (\Sigma_{ZZ}+\epsilon_nI)](\Sigma_{ZZ}+\epsilon_nI)^{-1}\\
    =& -\epsilon_n{\Lambda_{\Ddot{X}Z}} \Sigma_{ZZ}^{\dagger} (\Sigma_{ZZ}+\epsilon_nI)^{-1}. 
\end{align*}
{Hence, by \eqref{eq-term2-med} and Assumption \ref{ass-beta},} we have
\begin{align}\label{eq-term1-rate}
    \|{\Lambda_{\Ddot{X}Z}} (\Sigma_{ZZ}+\epsilon_nI)^{-1} - {\Lambda_{\Ddot{X}Z}} \Sigma_{ZZ}^{\dagger}\|_{\mathrm{OP}}
    = O_P(\epsilon_n^{\beta\wedge 1}).
\end{align}

For the third term of \eqref{eq-decomp-bhat-b}, since
\begin{align*}
    & \hat{\Sigma}_{\Ddot{X}Z} -\Sigma_{\Ddot{X}Z}\Sigma_{ZZ}^{\dagger}\hat{\Sigma}_{ZZ} -\tilde{\Sigma}_{\Ddot{X}Z|Z} \\
    =& E_n[\hat{\kappa}_{\Ddot{X}}(\cdot,\Ddot{X}) \otimes \hat{\kappa}_Z(\cdot,Z)] - B_{\Ddot{X}|Z} E_n[\hat{\kappa}_Z(\cdot,Z) \otimes \hat{\kappa}_Z(\cdot,Z)] - E_n[\tilde{\kappa}_{\Ddot{X}|Z}(\cdot,\Ddot{X}|Z) \otimes \tilde{\kappa}_Z(\cdot,Z)] \\
    =& E_n[\tilde{\kappa}_{\Ddot{X}}(\cdot,\Ddot{X}) \otimes \hat{\kappa}_Z(\cdot,Z)] - B_{\Ddot{X}|Z} E_n[\tilde{\kappa}_Z(\cdot,Z) \otimes \hat{\kappa}_Z(\cdot,Z)] - E_n[\tilde{\kappa}_{\Ddot{X}|Z}(\cdot,\Ddot{X}|Z) \otimes \tilde{\kappa}_Z(\cdot,Z)] \\
    =& E_n\{[\tilde{\kappa}_{\Ddot{X}}(\cdot,\Ddot{X}) - B_{\Ddot{X}|Z} \tilde{\kappa}_Z(\cdot,Z)] \otimes \hat{\kappa}_Z(\cdot,Z)\} - E_n[\tilde{\kappa}_{\Ddot{X}|Z}(\cdot,\Ddot{X}|Z) \otimes \tilde{\kappa}_Z(\cdot,Z)] \\
    =& E_n[\tilde{\kappa}_{\Ddot{X}|Z}(\cdot,\Ddot{X}|Z) \otimes \hat{\kappa}_Z(\cdot,Z)] - E_n[\tilde{\kappa}_{\Ddot{X}|Z}(\cdot,\Ddot{X}|Z) \otimes \tilde{\kappa}_Z(\cdot,Z)] \\
    =& E_n\{\tilde{\kappa}_{\Ddot{X}|Z}(\cdot,\Ddot{X}|Z) \otimes [\hat{\kappa}_Z(\cdot,Z) - \tilde{\kappa}_Z(\cdot,Z)] \} \\
    =& - E_n(\tilde{\kappa}_{\Ddot{X}|Z}(\cdot,\Ddot{X}|Z) \otimes \{E_n[\kappa_Z(\cdot,Z)] - E[\kappa_Z(\cdot,Z)] \}) \\
    =& - E_n[\tilde{\kappa}_{\Ddot{X}|Z}(\cdot,\Ddot{X}|Z)] \otimes \{E_n[\kappa_Z(\cdot,Z)] - E[\kappa_Z(\cdot,Z)] \}.
\end{align*}
Note that $E[\tilde{\kappa}_{\Ddot{X}|Z}(\cdot,\Ddot{X}|Z)]=0$. By Hilbert space central limit theorem, we have 
\begin{align*}
    \|E_n[\tilde{\kappa}_{\Ddot{X}|Z}(\cdot,\Ddot{X}|Z)]\|_{\mathcal{G}_{\Ddot{X}}} = O_P(n^{-1/2}), \quad \|E_n[\kappa_Z(\cdot,Z)] - E[\kappa_Z(\cdot,Z)]\|_{\mathcal{G}_Z} = O_P(n^{-1/2}).
\end{align*}
Thus, 
\begin{align*}
    &\|\hat{\Sigma}_{\Ddot{X}Z} -\Sigma_{\Ddot{X}Z}\Sigma_{ZZ}^{\dagger}\hat{\Sigma}_{ZZ} -\tilde{\Sigma}_{\Ddot{X}Z|Z}\|_{\mathrm{OP}} 
    \le \|\hat{\Sigma}_{\Ddot{X}Z} -\Sigma_{\Ddot{X}Z}\Sigma_{ZZ}^{\dagger}\hat{\Sigma}_{ZZ} -\tilde{\Sigma}_{\Ddot{X}Z|Z}\|_{\mathrm{HS}} \\
    &= \|E_n[\tilde{\kappa}_{\Ddot{X}|Z}(\cdot,\Ddot{X}|Z)]\|_{\mathcal{G}_{\Ddot{X}}}  \|E_n[\kappa_Z(\cdot,Z)] - E[\kappa_Z(\cdot,Z)]\|_{\mathcal{G}_Z} 
    = O_P(n^{-1/2}) O_P(n^{-1/2}) = O_P(n^{-1}),
\end{align*}
which implies that
\begin{align}\label{eq-term3-rate}
\begin{split}
    &\|(\hat{\Sigma}_{\Ddot{X}Z} -\Sigma_{\Ddot{X}Z}\Sigma_{ZZ}^{\dagger}\hat{\Sigma}_{ZZ} -\tilde{\Sigma}_{\Ddot{X}Z|Z})(\hat{\Sigma}_{ZZ}+\epsilon_nI)^{-1}\|_{\mathrm{OP}}\\
    &\le \|\hat{\Sigma}_{\Ddot{X}Z} -\Sigma_{\Ddot{X}Z}\Sigma_{ZZ}^{\dagger}\hat{\Sigma}_{ZZ} -\tilde{\Sigma}_{\Ddot{X}Z|Z}\|_{\mathrm{OP}} \|(\hat{\Sigma}_{ZZ}+\epsilon_nI)^{-1}\|_{\mathrm{OP}}
    = O_P(n^{-1}) O_P(\epsilon_n^{-1}) = O_P(n^{-1}\epsilon_n^{-1}).
\end{split}
\end{align}

We now consider the fifth term of \eqref{eq-decomp-bhat-b}, $\| \tilde{\Sigma}_{\Ddot{X}Z|Z}(\Sigma_{ZZ}+\epsilon_nI)^{-1} \|_{\mathrm{OP}}$. Note that 
this term is bounded by $\| \tilde{\Sigma}_{\Ddot{X}Z|Z}(\Sigma_{ZZ}+\epsilon_nI)^{-1} \|_{\mathrm{HS}}$, and we calculate the expectation of its square as
\begin{align}\label{eq-ehssq}
\begin{split}
     E \| \tilde{\Sigma}_{\Ddot{X}Z|Z}(\Sigma_{ZZ}+\epsilon_nI)^{-1} \|_{\mathrm{HS}}^2 
    =& E \| E_n[\tilde{\kappa}_{\Ddot{X}|Z}(\cdot,\Ddot{X}|Z) \otimes \tilde{\kappa}_Z(\cdot,Z)] (\Sigma_{ZZ}+\epsilon_nI)^{-1} \|_{\mathrm{HS}}^2 \\
    =& n^{-1} E \| [\tilde{\kappa}_{\Ddot{X}|Z}(\cdot,\Ddot{X}|Z) \otimes \tilde{\kappa}_Z(\cdot,Z)] (\Sigma_{ZZ}+\epsilon_nI)^{-1} \|_{\mathrm{HS}}^2 \\
    =& n^{-1} E [ \| \tilde{\kappa}_{\Ddot{X}|Z}(\cdot,\Ddot{X}|Z) \|_{\mathcal{G}_{\Ddot{X}}}^2 \|(\Sigma_{ZZ}+\epsilon_nI)^{-1} \tilde{\kappa}_Z(\cdot,Z)  \|_{\mathcal{G}_Z}^2],
\end{split}
\end{align}
where the second equality holds because $(X_i,Z_i)$ are i.i.d. samples for $i=1,\dots,n$, as well as 
\begin{align*}
    E[\tilde{\kappa}_{\Ddot{X}|Z}(\cdot,\Ddot{X}|Z) \otimes \tilde{\kappa}_Z(\cdot,Z)] 
    = E\{ E[\tilde{\kappa}_{\Ddot{X}|Z}(\cdot,\Ddot{X}|Z)] \otimes \tilde{\kappa}_Z(\cdot,Z)\} = 0.
\end{align*}
Note that Assumption \ref{ass-moment2} implies that, for some uniform $C>0$, we have
\begin{align*}
    \|\tilde{\kappa}_{\Ddot{X}|Z}(\cdot,\Ddot{X}|Z) \|_{\mathcal{G}_{\Ddot{X}}}
    &= \|\kappa_{\Ddot{X}}(\cdot,\Ddot{X}) - E[\kappa_{\Ddot{X}}(\cdot,\Ddot{X})|Z] \|_{\mathcal{G}_{\Ddot{X}}}\\
    &\le 2 \|\kappa_{\Ddot{X}}(\cdot,\Ddot{X})\|_{\mathcal{G}_{\Ddot{X}}} + 2 \| E[\kappa_{\Ddot{X}}(\cdot,\Ddot{X})|Z] \|_{\mathcal{G}_{\Ddot{X}}} \le 4C.
\end{align*}
It remains to analyze   $\|(\Sigma_{ZZ}+\epsilon_nI)^{-1} \tilde{\kappa}_Z(\cdot,Z)  \|_{\mathcal{G}_Z}^2$. By the Karhunen-Lo\'eve expansion \eqref{eq-kl-expansion}, we have
\begin{align*}
    (\Sigma_{ZZ}+\epsilon_nI)^{-1} \tilde{\kappa}_Z(\cdot,Z) = \sum_{k=1}^\infty (\Sigma_{ZZ}+\epsilon_nI)^{-1}  {\gamma_k^{1/2}}\xi_k \phi_k = \sum_{k=1}^\infty (\gamma_k+\epsilon_n)^{-1} {\gamma_k^{1/2}}\xi_k \phi_k,
\end{align*}
we have
\begin{align*}
    \|(\Sigma_{ZZ}+\epsilon_nI)^{-1} \tilde{\kappa}_Z(\cdot,Z)  \|_{\mathcal{G}_Z}^2 = \sum_{k=1}^\infty (\gamma_k+\epsilon_n)^{-2} {\gamma_k}\xi_k^2.
\end{align*}
Note that $E(\xi_k^2) = {1}$. Plugging the results back into \eqref{eq-ehssq}, by Fubini's theorem, we have
\begin{align*}
     E \| \tilde{\Sigma}_{\Ddot{X}Z|Z}(\Sigma_{ZZ}+\epsilon_nI)^{-1} \|_{\mathrm{HS}}^2 
    =& n^{-1} 16C^2 E\left[\sum_{k=1}^\infty (\gamma_k+\epsilon_n)^{-2} {\gamma_k}\xi_k^2\right] \\
    =& n^{-1} 16C^2 \sum_{k=1}^\infty (\gamma_k+\epsilon_n)^{-2} \gamma_k 
    = O(n^{-1} \epsilon_n^{-(\alpha+1)/\alpha}),
\end{align*}
where the last equality holds because $\sum_{k=1}^\infty (\gamma_k+\epsilon_n)^{-2} \gamma_k = O(\epsilon_n^{-(\alpha+1)/\alpha})$ under Assumption \ref{ass-alpha} according to {Lemma \ref{lem-alpha}}. By Chebyshev's inequality, we have
\begin{align}\label{eq-term5-rate}
    \| \tilde{\Sigma}_{\Ddot{X}Z|Z}(\Sigma_{ZZ}+\epsilon_nI)^{-1} \|_{\mathrm{OP}} \le \| \tilde{\Sigma}_{\Ddot{X}Z|Z}(\Sigma_{ZZ}+\epsilon_nI)^{-1} \|_{\mathrm{HS}} = O_P(n^{-1/2} \epsilon_n^{-(\alpha+1)/(2\alpha)}).
\end{align}

Regarding the fourth term of \eqref{eq-decomp-bhat-b}, we have
\begin{align*}
    &\| \tilde{\Sigma}_{\Ddot{X}Z|Z}[(\hat{\Sigma}_{ZZ}+\epsilon_nI)^{-1} -  (\Sigma_{ZZ}+\epsilon_nI)^{-1}] \|_{\mathrm{OP}} \\
    &= \| \tilde{\Sigma}_{\Ddot{X}Z|Z}(\Sigma_{ZZ}+\epsilon_nI)^{-1} ( \Sigma_{ZZ}-  \hat{\Sigma}_{ZZ}) (\hat{\Sigma}_{ZZ}+\epsilon_nI)^{-1}\|_{\mathrm{OP}} \\
    &\le \| \tilde{\Sigma}_{\Ddot{X}Z|Z}(\Sigma_{ZZ}+\epsilon_nI)^{-1} \|_{\mathrm{OP}} \| \Sigma_{ZZ}-  \hat{\Sigma}_{ZZ}\|_{\mathrm{OP}} \|(\hat{\Sigma}_{ZZ}+\epsilon_nI)^{-1}\|_{\mathrm{OP}}.
\end{align*}
Note that the first part $\| \tilde{\Sigma}_{\Ddot{X}Z|Z}(\Sigma_{ZZ}+\epsilon_nI)^{-1} \|_{\mathrm{OP}}$ is exactly the fifth term, and its rate is given by \eqref{eq-term5-rate}. Same as the arguments in the second term of \eqref{eq-decomp-bhat-b}, we have $\|\hat{\Sigma}_{ZZ}-\Sigma_{ZZ}\|_{\mathrm{OP}} = O_P(n^{-1/2})$ and $\|(\hat{\Sigma}_{ZZ}+\epsilon_nI)^{-1}\|_{\mathrm{OP}} =O_P( \epsilon_n^{-1})$. Therefore, the fourth term has the rate
\begin{align}\label{eq-term4-rate}
\begin{split}
    \| \tilde{\Sigma}_{\Ddot{X}Z|Z}[(\hat{\Sigma}_{ZZ}+\epsilon_nI)^{-1} -  (\Sigma_{ZZ}+\epsilon_nI)^{-1}] \|_{\mathrm{OP}} 
    =& O_P(n^{-1/2} \epsilon_n^{-(\alpha+1)/(2\alpha)}) O_P(n^{-1/2}) O_P( \epsilon_n^{-1})\\
    =& O_P(n^{-1} \epsilon_n^{-(3\alpha+1)/(2\alpha)}).
\end{split}
\end{align}

Substituting \eqref{eq-term1-rate}, \eqref{eq-term2-rate}, \eqref{eq-term3-rate}, \eqref{eq-term4-rate}, \eqref{eq-term5-rate} back into \eqref{eq-decomp-bhat-b}, we have the desired result {for $\|\hat{B}_{\Ddot{X}|Z} - B_{\Ddot{X}|Z} - R \|_{\mathrm{OP}}$}.

{
For the rate of $\|\hat{B}_{Y|Z}-B_{Y|Z}\|_{\mathrm{OP}}$, using a similar triangle inequality argument, \eqref{eq-decomp-bhat-b} becomes
\begin{align*}
\begin{split}
    &\|\hat{B}_{Y|Z}-B_{Y|Z}\|_{\mathrm{OP}}
    =\|\hat{\Sigma}_{YZ}(\hat{\Sigma}_{ZZ}+\epsilon_nI)^{-1}-\Sigma_{YZ}\Sigma_{ZZ}^{\dagger}\|_{\mathrm{OP}}\\
    &\le\|\Sigma_{YZ} (\Sigma_{ZZ}+\epsilon_nI)^{-1} - \Sigma_{YZ} \Sigma_{ZZ}^{\dagger}\|_{\mathrm{OP}}\\
    & \quad + \|[\Sigma_{YZ}\Sigma_{ZZ}^{\dagger}\hat{\Sigma}_{ZZ}(\hat{\Sigma}_{ZZ}+\epsilon_nI)^{-1} - \Sigma_{YZ} (\Sigma_{ZZ}+\epsilon_nI)^{-1}\|_{\mathrm{OP}} \\
    & \quad + \| (\hat{\Sigma}_{YZ} -\Sigma_{YZ}\Sigma_{ZZ}^{\dagger}\hat{\Sigma}_{ZZ} -\tilde{\Sigma}_{YZ|Z})(\hat{\Sigma}_{ZZ}+\epsilon_nI)^{-1} \|_{\mathrm{OP}}\\ 
    & \quad + \| \tilde{\Sigma}_{YZ|Z}[(\hat{\Sigma}_{ZZ}+\epsilon_nI)^{-1} -  (\Sigma_{ZZ}+\epsilon_nI)^{-1}] \|_{\mathrm{OP}} \\
    & \quad + \| \tilde{\Sigma}_{YZ|Z}  (\Sigma_{ZZ}+\epsilon_nI)^{-1} \|_{\mathrm{OP}},
\end{split}
\end{align*}
where
\begin{align*}
    \tilde{\Sigma}_{YZ|Z} = E_n[\tilde{\kappa}_{Y|Z}(\cdot,Y|Z) \otimes \tilde{\kappa}_Z(\cdot,Z)],
\end{align*}
and $\tilde{\kappa}_{Y|Z}$ is defined by \eqref{eq-k-tilde-xz}. The remaining arguments are same as in $\|\hat{B}_{\Ddot{X}|Z} - B_{\Ddot{X}|Z} - R \|_{\mathrm{OP}}$, with $\Ddot{X}$ replaced by $Y$ and $\Lambda_{\Ddot{X}Z}$ replaced by $\Sigma_{YZ}$.
}
\end{proof}

\subsection{Proof of Theorem \ref{thm-clt-sigmahat}}\label{thm-clt-sigmahat-proof}

\def\hii#1{^{(#1)}}

\begin{proof}
We first define two intermediate linear operators:
\begin{align*}
    \tilde{\Sigma}_{\Ddot{X}Y|Z}\hii 1= & E_n\left\{[\tilde{\kappa}_{\Ddot{X}}(\cdot,\Ddot{X})-B_{\Ddot{X}|Z}\tilde{\kappa}_Z(\cdot,Z)]\otimes[\tilde{\kappa}_{Y}(\cdot,Y)-B_{Y|Z}\tilde{\kappa}_Z(\cdot,Z)]\right\},  \\
    \tilde{\Sigma}_{\Ddot{X}Y|Z}\hii 2= & E_n\left\{[\tilde{\kappa}_{\Ddot{X}}(\cdot,\Ddot{X})-\hat{B}_{\Ddot{X}|Z}\tilde{\kappa}_Z(\cdot,Z)]\otimes[\tilde{\kappa}_{Y}(\cdot,Y)-\hat{B}_{Y|Z}\tilde{\kappa}_Z(\cdot,Z)]\right\}.
\end{align*}
Using these   operators we make the following decomposition: 
\begin{align*}
    \hat  {\Sigma}_{\Ddot{X}Y|Z} -  {\Sigma}_{\Ddot{X}Y|Z} =     \hat  {\Sigma}_{\Ddot{X}Y|Z} -     \tilde  {\Sigma}_{\Ddot{X}Y|Z} \hii 2 +   \tilde  {\Sigma}_{\Ddot{X}Y|Z} \hii 2 -   \tilde  {\Sigma}_{\Ddot{X}Y|Z} \hii 1 +   \tilde  {\Sigma}_{\Ddot{X}Y|Z} \hii 1 -  {\Sigma}_{\Ddot{X}Y|Z} 
\end{align*}
By the central limit theorem for Hilbert-space-valued random elements,
\begin{align}
    \sqrt{n}(\tilde{\Sigma}_{\Ddot{X}Y|Z} \hii 1-\Sigma_{\Ddot{X}Y|Z})\xrightarrow{\mathcal{D}}N(0,\Gamma_{\Ddot{X}Y|Z}).\label{eq-res-ktilde-k}
\end{align}
So we need to show that 
\begin{align}\label{eq:hat Sigma ddot}
&        \hat  {\Sigma}_{\Ddot{X}Y|Z} -     \tilde  {\Sigma}_{\Ddot{X}Y|Z} \hii 2  = o \lo p ( n \hi {-1/2}), \quad 
  \tilde  {\Sigma}_{\Ddot{X}Y|Z} \hii 2 -   \tilde  {\Sigma}_{\Ddot{X}Y|Z} \hii 1 = o \lo p ( n \hi {-1/2}). 
\end{align}

To show the first equation in (\ref{eq:hat Sigma ddot}), we begin with the decomposition
\begin{align}
\begin{split}
&\sqrt{n}(\hat{\Sigma}_{\Ddot{X}Y|Z}-\tilde{\Sigma}_{\Ddot{X}Y|Z} \hii 2)\\
&=\sqrt{n}\left(E_n\left\{[\hat{\kappa}_{\Ddot{X}}(\cdot,\Ddot{X})-\hat{B}_{\Ddot{X}|Z}\hat{\kappa}_Z(\cdot,Z)]\otimes[\hat{\kappa}_{Y}(\cdot,Y)-\hat{B}_{Y|Z}\hat{\kappa}_Z(\cdot,Z)]\right\}\right.\\
&\quad\quad\quad \left.-E_n\left\{[\tilde{\kappa}_{\Ddot{X}}(\cdot,\Ddot{X})-\hat{B}_{\Ddot{X}|Z}\tilde{\kappa}_Z(\cdot,Z)]\otimes[\tilde{\kappa}_{Y}(\cdot,Y)-\hat{B}_{Y|Z}\tilde{\kappa}_Z(\cdot,Z)]\right\}\right)\\
&=\sqrt{n}\big\{E_n[\hat{\kappa}_{\Ddot{X}}(\cdot,\Ddot{X})\otimes\hat{\kappa}_{Y}(\cdot,Y)-\tilde{\kappa}_{\Ddot{X}}(\cdot,\Ddot{X})\otimes\tilde{\kappa}_{Y}(\cdot,Y)]\\
&\quad\quad\quad -\hat{B}_{\Ddot{X}|Z}E_n[\hat{\kappa}_Z(\cdot,Z)\otimes\hat{\kappa}_{Y}(\cdot,Y)-\tilde{\kappa}_Z(\cdot,Z)\otimes\tilde{\kappa}_{Y}(\cdot,Y)]\\
&\quad\quad\quad-E_n[\hat{\kappa}_{\Ddot{X}}(\cdot,\Ddot{X})\otimes\hat{\kappa}_Z(\cdot,Z)-\tilde{\kappa}_{\Ddot{X}}(\cdot,\Ddot{X})\otimes\tilde{\kappa}_Z(\cdot,Z)]\hat{B}_{Y|Z}^*\\
&\quad\quad\quad+\hat{B}_{\Ddot{X}|Z}E_n[\hat{\kappa}_Z(\cdot,Z)\otimes\hat{\kappa}_Z(\cdot,Z)-\tilde{\kappa}_Z(\cdot,Z)\otimes\tilde{\kappa}_Z(\cdot,Z)]\hat{B}_{Y|Z}^*\big\}.
\end{split}\label{eq-decomp-khat-kbar}
\end{align}
We now consider each term on the right-hand side of \eqref{eq-decomp-khat-kbar} separately. By construction, the first term is 
\begin{align*}
& \sqrt{n}E_n[\hat{\kappa}_{\Ddot{X}}(\cdot,\Ddot{X})\otimes\hat{\kappa}_{Y}(\cdot,Y)-\tilde{\kappa}_{\Ddot{X}}(\cdot,\Ddot{X})\otimes\tilde{\kappa}_{Y}(\cdot,Y)]\\
& = \frac{1}{\sqrt{n}}\sum_{i=1}^n\left\{[\kappa_{\Ddot{X}}(\cdot,\Ddot{X}_i)-\hat{\mu}_{\Ddot{X}}]\otimes[\kappa_{Y}(\cdot,Y_i)-\hat{\mu}_{Y}]-[\kappa_{\Ddot{X}}(\cdot,\Ddot{X}_i)-\mu_{\Ddot{X}}]\otimes[\kappa_{Y}(\cdot,Y_i)-\mu_{Y}]\right\}. 
\end{align*}
The right-hand side can be decomposed into eight terms, among which the two terms involving $\ka \lo {\ddot X} ( \cdot, \ddot X \lo i) \otimes \ka \lo {Y} ( \cdot, Y \lo i) $ and $-\ka \lo {\ddot X} ( \cdot, \ddot X \lo i) \otimes \ka \lo {Y} ( \cdot, Y \lo i) $ cancel out, leaving  the following six terms: 
\begin{align*}
&\frac{1}{\sqrt{n}}\sum_{i=1}^n\big[-\kappa_{\Ddot{X}}(\cdot,\Ddot{X}_i)\otimes\hat{\mu}_{Y}-\hat{\mu}_{\Ddot{X}}\otimes\kappa_{Y}(\cdot,Y_i)+\hat{\mu}_{\Ddot{X}}\otimes\hat{\mu}_{Y}
\\
&\qquad\qquad\ +\kappa_{\Ddot{X}}(\cdot,\Ddot{X}_i)\otimes\mu_{Y}+\mu_{\Ddot{X}}\otimes\kappa_{Y}(\cdot,Y_i)-\mu_{\Ddot{X}}\otimes\mu_{Y}\big]. 
\end{align*}
Next, applying the relations $n \inv \sum \lo {i=1} \hi n \ka \lo {\ddot X} ( \cdot, \ddot X \lo i ) = \hat \mu \lo {\ddot X}$ and $n \inv \sum \lo {i=1} \hi n \ka \lo {Y} ( \cdot, Y \lo i ) = \hat \mu \lo {Y}$, we further reduce the right-hand side above to 
\begin{align*}
\begin{split}
&\sqrt{n}(-\hat{\mu}_{\Ddot{X}}\otimes\hat{\mu}_{Y}-\hat{\mu}_{\Ddot{X}}\otimes\hat{\mu}_{Y}+\hat{\mu}_{\Ddot{X}}\otimes\hat{\mu}_{Y}+\hat{\mu}_{\Ddot{X}}\otimes\mu_{Y}+\mu_{\Ddot{X}}\otimes\hat{\mu}_{Y}-\mu_{\Ddot{X}}\otimes\mu_{Y})\\
=&\sqrt{n}(-\hat{\mu}_{\Ddot{X}}\otimes\hat{\mu}_{Y}+\hat{\mu}_{\Ddot{X}}\otimes\mu_{Y}+\mu_{\Ddot{X}}\otimes\hat{\mu}_{Y}-\mu_{\Ddot{X}}\otimes\mu_{Y})\\
=&-\sqrt{n}(\hat{\mu}_{\Ddot{X}}-\mu_{\Ddot{X}})\otimes\left(\hat{\mu}_{Y}-\mu_{Y}\right).
\end{split}
\end{align*}
By Chebychev's inequality, we have $\hat \mu \lo {\ddot X } - \mu \lo {\ddot X} = O \lo P ( n \hi {-1/2})$ and $\hat \mu \lo {Y } - \mu \lo {Y} = O \lo P ( n \hi {-1/2})$. Hence the right-hand side is of the order $O \lo P (n \hi {-1/2})$. Thus we have proved 
\begin{align*}
& \sqrt{n}E_n[\hat{\kappa}_{\Ddot{X}}(\cdot,\Ddot{X})\otimes\hat{\kappa}_{Y}(\cdot,Y)-\tilde{\kappa}_{\Ddot{X}}(\cdot,\Ddot{X})\otimes\tilde{\kappa}_{Y}(\cdot,Y)]
= O \lo P ( n \hi {-1/2}). 
\end{align*}
By the similar argument, we can show that 
\begin{align*}
&\sqrt n  E_n[\hat{\kappa}_Z(\cdot,Z)\otimes\hat{\kappa}_{Y}(\cdot,Y)-\tilde{\kappa}_Z(\cdot,Z)\otimes\tilde{\kappa}_{Y} (\cdot,Y)]= O \lo P ( n \hi {-1/2}) \\
&\sqrt n E_n[\hat{\kappa}_{\Ddot{X}}(\cdot,\Ddot{X})\otimes\hat{\kappa}_Z(\cdot,Z)-\tilde{\kappa}_{\Ddot{X}}(\cdot,\Ddot{X})\otimes\tilde{\kappa}_Z(\cdot,Z)]= O \lo P ( n \hi {-1/2})  \\
&\sqrt n  E_n[\hat{\kappa}_Z(\cdot,Z)\otimes\hat{\kappa}_Z(\cdot,Z)-\tilde{\kappa}_Z(\cdot,Z)\otimes\tilde{\kappa}_Z(\cdot,Z)]= O \lo P ( n \hi {-1/2}).
\end{align*}
{Note that \eqref{eq-r-def} and Assumption \ref{ass-beta} imply that 
\begin{align*}
B_{\Ddot{X}|Z}= \Sigma_{\Ddot{X}Z}\Sigma_{ZZ}^\dagger = \Lambda_{\Ddot{X}Z}\Sigma_{ZZ}^\dagger + \mu_X \Sigma_{ZZ}\Sigma_{ZZ}^\dagger = S_{\Ddot{X}Z} \Sigma_{ZZ}^\beta + \mu_XI
\end{align*}
is bounded and $B_{Y|Z} = \Sigma_{YZ}\Sigma_{ZZ}^\dagger = S_{YZ}\Sigma_{ZZ}^\beta$ is also bounded.
By Corollary \ref{cor-opt-eps} and 
Lemma \ref{lemma-regul-conv}}, 
 we have $\hat B \lo {\ddot X|Z} = O \lo P (1)$ and $\hat B \lo {Y|Z} = O \lo P (1)$. Hence 
\begin{align*}
& \sqrt{n}(\hat{\Sigma}_{\Ddot{X}Y|Z}-\tilde{\Sigma}_{\Ddot{X}Y|Z} \hii 2)\\
& =O_P(n^{-1/2})+O_P(1)O_P(n^{-1/2})+O_P(n^{-1/2})O_P(1)+O_P(1)O_P(n^{-1/2})O_P(1)\\
& =O_P(n^{-1/2})=o_P(1).
\end{align*}
This proves the first equation in (\ref{eq:hat Sigma ddot}).

To show  the second equation in (\ref{eq:hat Sigma ddot}), we begin by rewriting $\sqrt{n}(\tilde{\Sigma}_{\Ddot{X}Y|Z} \hii 2-\tilde{\Sigma}_{\Ddot{X}Y|Z} \hii 1)$ as 
\begin{align}\label{eq:sqrt n (}
\begin{split}
&\sqrt{n}\Big(E_n\left\{[\tilde{\kappa}_{\Ddot{X}}(\cdot,\Ddot{X})-\hat{B}_{\Ddot{X}|Z}\tilde{\kappa}_Z(\cdot,Z)]\otimes[\tilde{\kappa}_{Y}(\cdot,Y)-\hat{B}_{Y|Z}\tilde{\kappa}_Z(\cdot,Z)]\right\}\\
&\qquad -E_n\left\{[\tilde{\kappa}_{\Ddot{X}}(\cdot,\Ddot{X})-B_{\Ddot{X}|Z}\tilde{\kappa}_Z(\cdot,Z)]\otimes[\tilde{\kappa}_{Y}(\cdot,Y)-B_{Y|Z}\tilde{\kappa}_Z(\cdot,Z)]\right\}\Big).
\end{split}
\end{align}
Define random elements in $\epsilon \lo {\ddot X} \in \ca G \lo {\ddot X}$ and $\epsilon \lo Y \in \ca G \lo {Y}$ as follows
\begin{align*}
    \epsilon_{\Ddot{X}}  = \tilde{\kappa}_{\Ddot{X}}(\cdot,\Ddot{X})-B_{\Ddot{X}|Z}\tilde{\kappa}_Z(\cdot,Z), \quad \epsilon_{Y} = \tilde{\kappa}_{Y}(\cdot,Y)-B_{Y|Z}\tilde{\kappa}_Z(\cdot,Z).
\end{align*}
Then, 
\begin{align*}
&  \tilde{\kappa}_{\Ddot{X}}(\cdot,\Ddot{X})-\hat{B}_{\Ddot{X}|Z}\tilde{\kappa}_Z(\cdot,Z)  =    \epsilon_{\Ddot{X}} -(\hat{B}_{\Ddot{X}|Z} - {B}_{\Ddot{X}|Z}) \tilde{\kappa}_Z(\cdot,Z)  \\
& 
 \tilde{\kappa}_{Y}(\cdot,Y)-\hat{B}_{Y|Z}\tilde{\kappa}_Z(\cdot,Z) = \epsilon_{Y}  -(\hat{B}_{Y|Z}-{B}_{Y|Z})\tilde{\kappa}_Z(\cdot,Z) 
\end{align*}
Substituting these into (\ref{eq:sqrt n (}), we have 
\begin{align} 
\begin{split}
&\sqrt{n}(B_{\Ddot{X}|Z}-\hat{B}_{\Ddot{X}|Z})E_n[\tilde{\kappa}_Z(\cdot,Z)\otimes\epsilon_{Y}]+\sqrt{n}E_n[\epsilon_{\Ddot{X}}\otimes\tilde{\kappa}_Z(\cdot,Z)](B_{Y|Z}-\hat{B}_{Y|Z})^*\\
&\qquad +\sqrt{n}(B_{\Ddot{X}|Z}-\hat{B}_{\Ddot{X}|Z})E_n[\tilde{\kappa}_Z(\cdot,Z)\otimes\tilde{\kappa}_Z(\cdot,Z)](B_{Y|Z}-\hat{B}_{Y|Z})^*. \label{eq:sqrt n ( En}
\end{split}
\end{align} 

{
We now consider the three terms in \eqref{eq:sqrt n ( En} one by one. 
For the first term of \eqref{eq:sqrt n ( En}, take $\eta$ to be a number such that 
\begin{align*}
\max \left\{ 0, \frac{\alpha-1}{2\alpha} - (\beta \wedge 1) \right\} < \eta < \frac{\alpha-1}{2\alpha}.
\end{align*} 
Then,
\begin{align*}
&\sqrt{n}(B_{\Ddot{X}|Z}-\hat{B}_{\Ddot{X}|Z})E_n[\tilde{\kappa}_Z(\cdot,Z)\otimes\epsilon_Y]
= (B_{\Ddot{X}|Z}-\hat{B}_{\Ddot{X}|Z})\Sigma_{ZZ}^\eta \sqrt{n}E_n\left[\Sigma_{ZZ}^{-\eta}\tilde{\kappa}_Z(\cdot,Z)\otimes\epsilon_Y\right]\\
&= - (\hat{B}_{\Ddot{X}|Z}-B_{\Ddot{X}|Z}-R)\Sigma_{ZZ}^\eta \sqrt{n}E_n\left[\Sigma_{ZZ}^{-\eta}\tilde{\kappa}_Z(\cdot,Z)\otimes\epsilon_Y\right] - R\Sigma_{ZZ}^\eta \sqrt{n}E_n\left[\Sigma_{ZZ}^{-\eta}\tilde{\kappa}_Z(\cdot,Z)\otimes\epsilon_Y\right]
\end{align*}}
By Corollary \ref{cor-reg-form},  we have  
\begin{align*}
    E( \epsilon_Y |Z ) = E[\tilde{\kappa}_Y(\cdot,Y)|Z]-B_{Y|Z}\tilde{\kappa}_Z(\cdot,Z) = 0. 
\end{align*}
{Hence 
\begin{align*}
E\left[\Sigma_{ZZ}^{-\eta}\tilde{\kappa}_Z(\cdot,Z)\otimes\epsilon_Y \right]
=E\left[\Sigma_{ZZ}^{-\eta}\tilde{\kappa}_Z(\cdot,Z)\otimes E(\epsilon_Y|Z) \right]=0. 
\end{align*}
Since
\begin{align*}
\Sigma_{ZZ}^{-\eta} \tilde{\kappa}_Z(\cdot,Z)
= \left( \sum_{j=1}^\infty \gamma_j^{-\eta} \phi_j \otimes \phi_j\right) \left(\sum_{k=1}^\infty \xi_k \phi_k\right)
= \sum_{j=1}^\infty \gamma_j^{-\eta} \xi_j \phi_j,
\end{align*}
where $\xi_j$ are uncorrelated with mean zero and variance $\gamma_j$, we have, by Assumption \ref{ass-alpha},
\begin{align*}
E \| \Sigma_{ZZ}^{-\eta} \tilde{\kappa}_Z(\cdot,Z) \|_{\mathcal G_Z}^2 
= \sum_{j=1}^\infty \gamma_j^{-2\eta} E(\xi_j^2)
= \sum_{j=1}^\infty \gamma_j^{1-2\eta}
\preceq \sum_{j=1}^\infty j^{-\alpha(1-2\eta)}
< \infty,
\end{align*}
where the sum is finite because $\alpha(1-2\eta)>1$.
Moreover, Assumptions \ref{ass-moment2} and \ref{ass-beta} guarantee that $E\|\epsilon_Y\|_{\mathcal{G}_Y}^2<\infty$. 
By the central limit theorem for i.i.d. random elements in Hilbert spaces, we have
\begin{align*}
    \sqrt{n} E_n\left[\Sigma_{ZZ}^{-\eta}\tilde{\kappa}_Z(\cdot,Z)\otimes\epsilon_Y(\cdot)\right] \xrightarrow{\mathcal{D}}N(0,\Xi_{ZY}), 
\end{align*}
where
$ \Xi_{ZY}=E\left\{[\Sigma_{ZZ}^{-\eta}\tilde{\kappa}_Z(\cdot,Z)\otimes\epsilon_Y(\cdot)]\otimes[\Sigma_{ZZ}^{-\eta}\tilde{\kappa}_Z(\cdot,Z)\otimes\epsilon_Y(\cdot)]\right\}$. 
Hence
\begin{align*}
E_n\left[\Sigma_{ZZ}^{-\eta}\tilde{\kappa}_Z(\cdot,Z)\otimes\epsilon_Y(\cdot)\right]=O_P(n^{-1/2}).
\end{align*}
By Corollary \ref{cor-opt-eps}, the conditions given for $\alpha$ and $\beta$ guarantee that $\hat{B}_{\Ddot{X}|Z}-B_{\Ddot{X}|Z}-R = o_P(1)$.
Also, we know that $\Sigma_{ZZ}^\eta$ is bounded. Thus, 
\begin{align*}
(\hat{B}_{\Ddot{X}|Z}-B_{\Ddot{X}|Z}-R)\Sigma_{ZZ}^\eta \sqrt{n}E_n\left[\Sigma_{ZZ}^{-\eta}\tilde{\kappa}_Z(\cdot,Z)\otimes\epsilon_Y\right] = o_P(1) \sqrt{n} O_P(n^{-1/2}) = o_P(1).
\end{align*}
}

{
Taking the $\theta$ in \eqref{eq-r-sigma-theta} to be $\eta$, we have
\begin{align}\label{eq:2 OP}
\|R \Sigma_{ZZ}^\eta\|_{\mathrm{OP}}
= O_P(n^{-1/2}\epsilon_n^{\eta\wedge 1 -1}) 
+O_P( \epsilon_n^{\eta\wedge 1}).
\end{align}
Since $\epsilon_n = o(1)$ and $\eta>0$, the second term  on the right is $o \lo P (1)$. Substituting $\epsilon_n \asymp n^{-\frac{\alpha}{2\alpha(\beta\wedge 1)+\alpha+1}}$ into the first term on the right, it becomes $n \hi c$ where 
\begin{align*}
c = -\frac{1}{2} - \frac{\alpha}{2\alpha(\beta\wedge 1)+\alpha+1} (\eta-1)
< -\frac{1}{2} - \frac{\alpha}{2\alpha(\beta\wedge 1)+\alpha+1} \left\{\frac{\alpha-1}{2\alpha} - (\beta \wedge 1)-1\right\} = 0.
\end{align*}
Hence the first term on the right-hand side of (\ref{eq:2 OP}) is also $o \lo P(1)$, implying
\begin{align}\label{eq-r-sigma-zz-eta}
R\Sigma_{ZZ}^\eta \sqrt{n}E_n\left[\Sigma_{ZZ}^{-\eta}\tilde{\kappa}_Z(\cdot,Z)\otimes\epsilon_Y\right]=o_P(1) O_P(1) = o_P(1).
\end{align}
Thus we have proved that the first term in (\ref{eq:sqrt n ( En}) is of the order $o \lo P(1)$; that is, 
\begin{align}\label{eq-term1-small}
\sqrt{n}(B_{\Ddot{X}|Z}-\hat{B}_{\Ddot{X}|Z})E_n[\tilde{\kappa}_Z(\cdot,Z)\otimes\epsilon_Y]
= o_P(1).
\end{align}
}

{For the second term of \eqref{eq:sqrt n ( En}, similarly, by Corollary \ref{cor-reg-form},  we have  
\begin{align*}
    E( \epsilon_{\Ddot{X}} |Z ) = E[\tilde{\kappa}_{\Ddot{X}}(\cdot,\Ddot{X})|Z]-B_{\Ddot{X}|Z}\tilde{\kappa}_Z(\cdot,Z) = 0. 
\end{align*}
Hence
\begin{align*}
E\left[\epsilon_{\Ddot{X}} \otimes \tilde{\kappa}_Z(\cdot,Z) \right]
=E\left[ E(\epsilon_{\Ddot{X}} |Z) \otimes \tilde{\kappa}_Z(\cdot,Z) \right]=0. 
\end{align*}
Moreover, Assumptions \ref{ass-moment2} and \ref{ass-beta} guarantee that $E\|\epsilon_{\Ddot{X}}\|_{\mathcal{H}_{\Ddot{X}}}^2<\infty$ and $E\|\tilde{\kappa}_Z(\cdot,Z)\|_{\mathcal H_Z}^2 < \infty$. 
By the central limit theorem for i.i.d. random elements in Hilbert spaces, we have
\begin{align*}
    \sqrt{n} E_n\left[\epsilon_{\Ddot{X}}(\cdot) \otimes \tilde{\kappa}_Z(\cdot,Z)\right] \xrightarrow{\mathcal{D}}N(0,\Xi_{\Ddot{X}Z}), 
\end{align*}
where
$ \Xi_{\Ddot{X}Z}=E\left\{[\epsilon_{\Ddot{X}}(\cdot)\otimes\tilde{\kappa}_Z(\cdot,Z)]\otimes[\epsilon_{\Ddot{X}}(\cdot)\otimes\tilde{\kappa}_Z(\cdot,Z)]\right\}$.} Hence
\begin{align*}
E_n\left[\epsilon_{\Ddot{X}}(\cdot)\otimes\tilde{\kappa}_Z(\cdot,Z)\right]=O_P(n^{-1/2}).
\end{align*}
Furthermore, by Corollary \ref{cor-opt-eps}, the conditions given for $\alpha$ and $\beta$ guarantee that 
{$\hat{B}_{Y|Z}-B_{Y|Z} = o_P(1)$.
Hence the second term of \eqref{eq:sqrt n ( En} becomes
\begin{align}\label{eq-term2-small}
\sqrt{n}E_n\left[\epsilon_{\Ddot{X}}\otimes\tilde{\kappa}_Z(\cdot,Z)\right](B_{Y|Z}-\hat{B}_{Y|Z})^*
= \sqrt{n} O_P(n^{-1/2}) o_P(1) = o_P(1).
\end{align}}

{We now consider the last term of \eqref{eq:sqrt n ( En}. Note that
\begin{align}\label{eq:term3-decomp}
\begin{split}
&\sqrt{n}(B_{\Ddot{X}|Z}-\hat{B}_{\Ddot{X}|Z})E_n[\tilde{\kappa}_Z(\cdot,Z)\otimes\tilde{\kappa}_Z(\cdot,Z)](B_{Y|Z}-\hat{B}_{Y|Z})^*\\
&= (B_{\Ddot{X}|Z}-\hat{B}_{\Ddot{X}|Z})\sqrt{n}\left\{ E_n[\tilde{\kappa}_Z(\cdot,Z)\otimes\tilde{\kappa}_Z(\cdot,Z)] - \Sigma_{ZZ} \right\} (B_{Y|Z}-\hat{B}_{Y|Z})^*\\
& \quad - \sqrt{n} (\hat{B}_{\Ddot{X}|Z} - B_{\Ddot{X}|Z} -R ) \Sigma_{ZZ} (B_{Y|Z}-\hat{B}_{Y|Z})^*
- \sqrt{n} R \Sigma_{ZZ} (B_{Y|Z}-\hat{B}_{Y|Z})^*.
\end{split}
\end{align}
For the first term in \eqref{eq:term3-decomp}, since $\hat{B}_{\Ddot{X}|Z} - B_{\Ddot{X}|Z} - R = o_P(1)$ by Corollary \ref{cor-opt-eps}  and $R$ is bounded by Lemma \ref{lem-R-rate}, we have  $\hat{B}_{\Ddot{X}|Z} - B_{\Ddot{X}|Z} = O_p(1)$. Furthermore, by the central limit theorem in Hilbert spaces, we have,  under Assumption \ref{ass-moment2},
\begin{align*}
\sqrt{n}\left\{ E_n[\tilde{\kappa}_Z(\cdot,Z)\otimes\tilde{\kappa}_Z(\cdot,Z)] - \Sigma_{ZZ} \right\}\xrightarrow{\mathcal{D}}N(0,\Xi_{ZZ}), 
\end{align*}
where
$ \Xi_{ZZ}=E\left\{[\tilde{\kappa}_Z(\cdot,Z)\otimes\tilde{\kappa}_Z(\cdot,Z)]\otimes[\tilde{\kappa}_Z(\cdot,Z)\otimes\tilde{\kappa}_Z(\cdot,Z)]\right\}$. Hence
\begin{align*}
E_n[\tilde{\kappa}_Z(\cdot,Z)\otimes\tilde{\kappa}_Z(\cdot,Z)] - \Sigma_{ZZ} = O_P(n^{-1/2}).
\end{align*}
Also, by Corollary \ref{cor-opt-eps}, the conditions given for $\alpha$ and $\beta$ guarantee that 
$\hat{B}_{Y|Z}-B_{Y|Z} = o_P(1)$. 
Therefore, the first term in \eqref{eq:term3-decomp} is of the order
\begin{align*}
&(B_{\Ddot{X}|Z}-\hat{B}_{\Ddot{X}|Z})\sqrt{n}\left\{ E_n[\tilde{\kappa}_Z(\cdot,Z)\otimes\tilde{\kappa}_Z(\cdot,Z)] - \Sigma_{ZZ} \right\} (B_{Y|Z}-\hat{B}_{Y|Z})^*\\
&= O_P(1) \sqrt{n} O_P(n^{-1/2}) o_P(1) = o_P(1).
\end{align*}
For the second term in \eqref{eq:term3-decomp},  since the conditions given for $\alpha$ and $\beta$ guarantee that 
\begin{align*}
\hat{B}_{\Ddot{X}|Z}-B_{\Ddot{X}|Z} - R = o_P(n^{-1/4}),\qquad
\hat{B}_{Y|Z}-B_{Y|Z} = o_P(n^{-1/4}),
\end{align*}
we have
\begin{align}\label{eq-app-bhat-quarter}
\sqrt{n} (\hat{B}_{\Ddot{X}|Z} - B_{\Ddot{X}|Z} -R ) \Sigma_{ZZ} (B_{Y|Z}-\hat{B}_{Y|Z})^* 
= \sqrt{n} o_P(n^{-1/4}) O(1) o_p(n^{-1/4}) = o_P(1).
\end{align}
We now consider the third term in \eqref{eq:term3-decomp}. 
{By \eqref{eq-r-sigma-hs}, we know that
\begin{align*}
\|R \Sigma_{ZZ}\|_{\mathrm{HS}}
= O_P(n^{-1/2}) + O_P(\epsilon_n^{1-1/(2\alpha)}).
\end{align*}}
Also, by Corollary \ref{cor-opt-eps}, the conditions given for $\alpha$ and $\beta$ guarantee that 
\begin{align*}
\hat{B}_{Y|Z}-B_{Y|Z} = O_P(n^{-\frac{\alpha(\beta\wedge 1)}{2\alpha(\beta\wedge 1)+\alpha+1}}) = o_P(1), \quad \text{and} \quad 
\epsilon_n \asymp n^{-\frac{\alpha}{2\alpha(\beta\wedge 1)+\alpha+1}}.
\end{align*}
Thus, 
\begin{align}\label{eq-r-sigma-zz-bhat}
\sqrt{n} R \Sigma_{ZZ} (\hat{B}_{Y|Z}-B_{Y|Z}) 
= \sqrt{n} O_P(n^{-1/2}) o_P(1) + \sqrt{n} O_P(n^{-\frac{\alpha(\beta\wedge 1)}{2\alpha(\beta\wedge 1)+\alpha+1}}) {O(n^{-\frac{\alpha-1/2}{2\alpha(\beta\wedge 1)+\alpha+1}})}.
\end{align}
The first term on the right is obviously $o \lo P(1)$. The second term is of the form $n \hi c$ where 
\begin{align*}
c &= \frac{1}{2}-\frac{\alpha(\beta\wedge 1)+\alpha{-1/2}}{2\alpha(\beta\wedge 1)+\alpha+1}
= \frac{2\alpha(\beta\wedge 1)+\alpha+1 - 2\alpha(\beta\wedge 1) - 2\alpha{+1}}{2[2\alpha(\beta\wedge 1)+\alpha+1]}\\
&= \frac{{2-\alpha}}{2[2\alpha(\beta\wedge 1)+\alpha+1]}<0
\end{align*}
since {$\alpha>2$}. Thus all  three parts of   the last term in  \eqref{eq:sqrt n ( En} are of the order $o \lo P(1)$; that is, 
\begin{align}\label{eq-term3-small}
\sqrt{n}(B_{\Ddot{X}|Z}-\hat{B}_{\Ddot{X}|Z})E_n[\tilde{\kappa}_Z(\cdot,Z)\otimes\tilde{\kappa}_Z(\cdot,Z)](B_{Y|Z}-\hat{B}_{Y|Z})^*
= o_P(1).
\end{align}}
Substituting \eqref{eq-term1-small}, \eqref{eq-term2-small}, and \eqref{eq-term3-small} into \eqref{eq:sqrt n ( En}, we have{
\begin{align}
\sqrt{n}(\tilde{\Sigma}_{\Ddot{X}Y|Z} \hii 2-\tilde{\Sigma}_{\Ddot{X}Y|Z} \hii 1)
=o_P(1).
\label{eq-res-kbar-ktilde}
\end{align}}
This proves the second equation in (\ref{eq:hat Sigma ddot}), completing the proof of the theorem.
\end{proof}

\subsection{Proof of Theorem \ref{thm:local power}}\label{thm:local power-proof}

\begin{proof}
{Following the argument in the proof of Theorem 3 of \cite{tang2024nonparametric}, we have the expansion
\begin{align*}
n\|\hat{\Sigma} \lo {\Ddot X Y | Z} \| \lo {\mathrm{HS}}  \hi 2 
=  n\|\hat{\Sigma} \lo {\Ddot X Y | Z} - \Sigma \lo {\Ddot X Y | Z}\| \lo {\mathrm{HS}}  \hi 2 + 2n \langle \hat{\Sigma} \lo {\Ddot X Y | Z} - \Sigma \lo {\Ddot X Y | Z}, \Sigma \lo {\Ddot X Y | Z} \rangle \lo {\mathrm{HS}} + n\|\Sigma \lo {\Ddot X Y | Z}\| \lo {\mathrm{HS}}  \hi 2.
\end{align*}
By Theorem \ref{thm-clt-sigmahat},  
$\sqrt{n}(\hat{\Sigma} \lo {\Ddot X Y | Z} - \Sigma \lo {\Ddot X Y | Z}) \xrightarrow{\mathcal{D}} G$ where $G \sim N(0,\Gamma \lo {\Ddot X Y | Z})$. By the continuous mapping theorem,  
\begin{align*}
n\|\hat{\Sigma} \lo {\Ddot X Y | Z} \| \lo {\mathrm{HS}}  \hi 2 
 \xrightarrow{\mathcal{D}} \| G \| \lo {\mathrm{HS}}  \hi 2 + 2 \sqrt{n} \langle G, \Sigma \lo {\Ddot X Y | Z} \rangle \lo {\mathrm{HS}} + n\|\Sigma \lo {\Ddot X Y | Z}\| \lo {\mathrm{HS}}  \hi 2.
\end{align*}
 Thus, under $H \lo 1 \hi {(n)}$,
\begin{align}\label{eq:local alternative}
n\|\breve{\Sigma} \lo {\Ddot X Y | Z} \| \lo {\mathrm{HS}}  \hi 2  \xrightarrow{\mathcal{D}}   \| G \| \lo {\mathrm{HS}}  \hi 2 + 2 \langle G, \Sigma \lo 1 \rangle \lo {\mathrm{HS}}+ c \hi 2.
\end{align}
Since $G \sim N(0,\Gamma \lo {\Ddot X Y | Z})$ ane $(\lambda \lo 1, v \lo 1), (\lambda \lo 2, v \lo 2), \dots$ are eigenvalue-eigenfunction pairs of $\Gamma \lo {\Ddot X Y | Z}$, we see that $\langle G, v \lo j \rangle \lo {\mathrm{HS}} \sim N(0, \lambda \lo j)$ for $j=1,2,\dots$. Therefore, the Karhunen-Lo\'eve expansion for $G$ is
\begin{align*}
G = \sum \lo {j=1} \hi {\infty} \sqrt{\lambda \lo j} Z \lo j v \lo j,
\end{align*}
where $Z \lo 1, Z \lo 2, \ldots$ are independent standard normal random variables. See Example 11.4.1 of \cite{kokoszka2017introduction} or Section 7.3 of \cite{hsing2015theoretical} for details.
Also, since $\Sigma \lo 1 = \sum \lo {j=1} \hi \infty \sigma \lo j  v \lo j$, we have $\sum \lo {j=1} \hi {\infty} \sigma \lo j \hi 2= \|\Sigma \lo 1\| \lo {\mathrm{HS}} \hi 2  = c \hi 2$. In addition, since  
\begin{align*}
\| G \| \lo {\mathrm{HS}} \hi 2 =  \sum \lo {j=1} \hi {\infty} \lambda \lo j Z \lo j \hi 2, \quad\langle G, \Sigma \lo 1 \rangle \lo {\mathrm{HS}} = \sum \lo {j=1} \hi {\infty} \sqrt{\lambda \lo j } Z \lo j \sigma \lo j, 
\end{align*}
we can rewrite the right-hand side of (\ref{eq:local alternative}) as 
\begin{align*}
\| G \| \lo {\mathrm{HS}}  \hi 2 + 2 \langle G, \Sigma \lo 1 \rangle \lo {\mathrm{HS}}+ c \hi 2 
&=  \sum \lo {j=1} \hi {\infty} \lambda \lo j Z \lo j \hi 2 + 2 \sum \lo {j=1} \hi {\infty} \sqrt{\lambda \lo j } Z \lo j \sigma \lo j + \sum \lo {j=1} \hi {\infty} \sigma \lo j \hi 2 \\
&=  \sum \lo {j=1} \hi {\infty} \lambda \lo j ( Z \lo j + \sigma \lo j / \sqrt{\lambda \lo j} ) \hi 2 
=:  \sum \lo {j=1} \hi {\infty} \lambda \lo j \tilde Z \lo j \hi 2,
\end{align*}
where $\tilde Z \lo 1, \tilde Z \lo 2, \ldots$ are independent $N(\sigma \lo j  / \sqrt{\lambda \lo j}, 1)$, as desired. The local power result follows from the definition of convergence in distribution.}
\end{proof}

\clearpage

\section{Additional Plots for Section \ref{sec-application}}\label{sec-app-plot}

\begin{figure}[!htb] 
\centering 
\includegraphics[width=\linewidth]{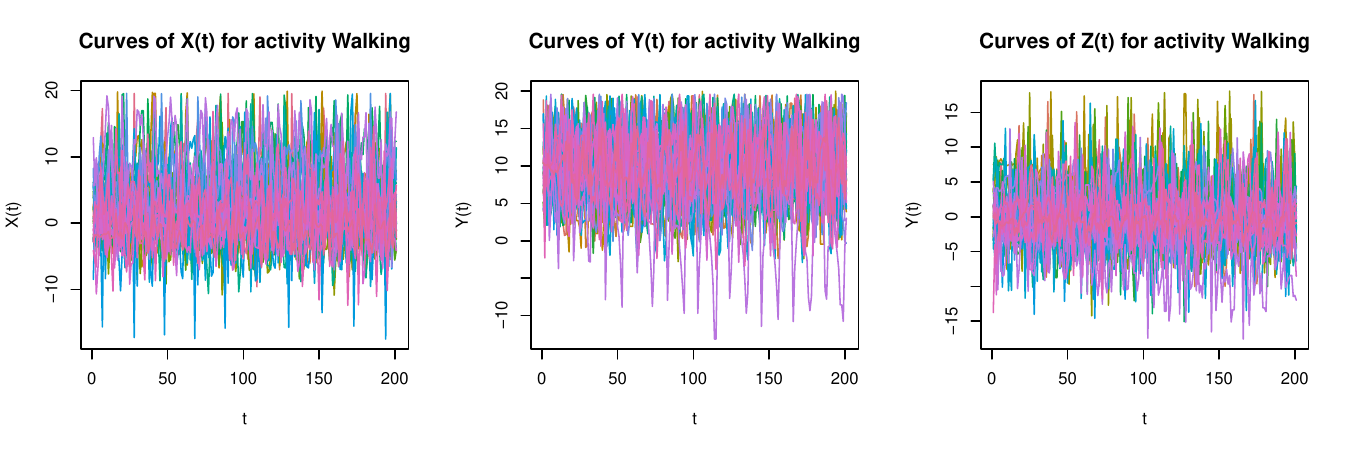}
\caption{Acceleration curves for Walking based on $n=36$ users and $m=200$ observations per user. The left, middle, and right panels show $X(t)$, $Y(t)$, and $Z(t)$, respectively, and each colored curve represents one user. The observations cover the first 10 seconds of data recorded at 20 Hz.}
\label{fig-curve-1}
\end{figure}

\begin{figure}[!htb] 
\centering 
\includegraphics[width=\linewidth]{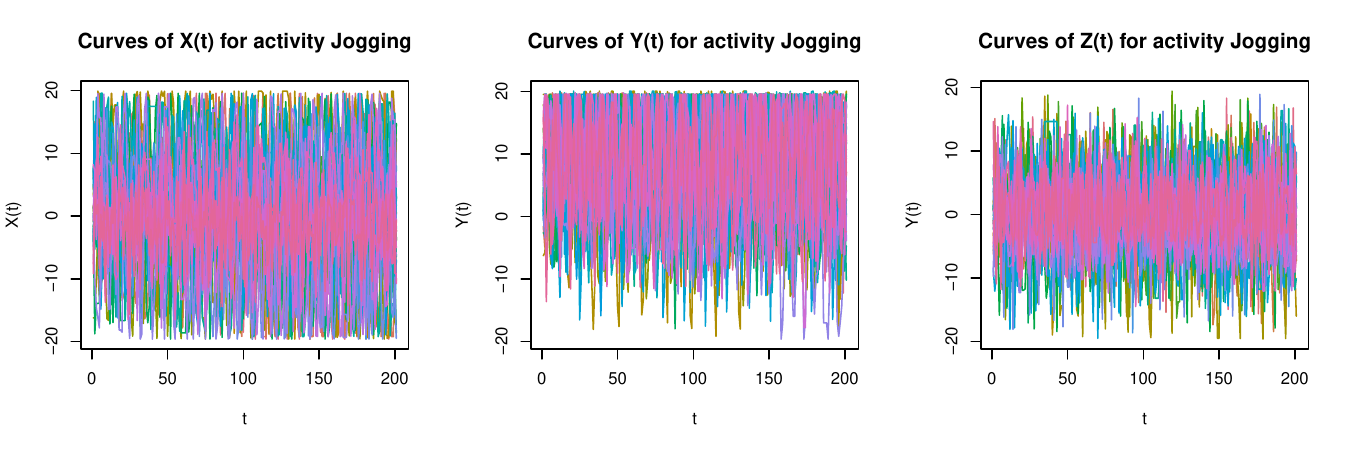}
\caption{Acceleration curves for Jogging based on $n=32$ users. The panel arrangement and observation window are the same as in Figure \ref{fig-curve-1}.}
\label{fig-curve-2}
\end{figure}

\begin{figure}[!htb] 
\centering 
\includegraphics[width=\linewidth]{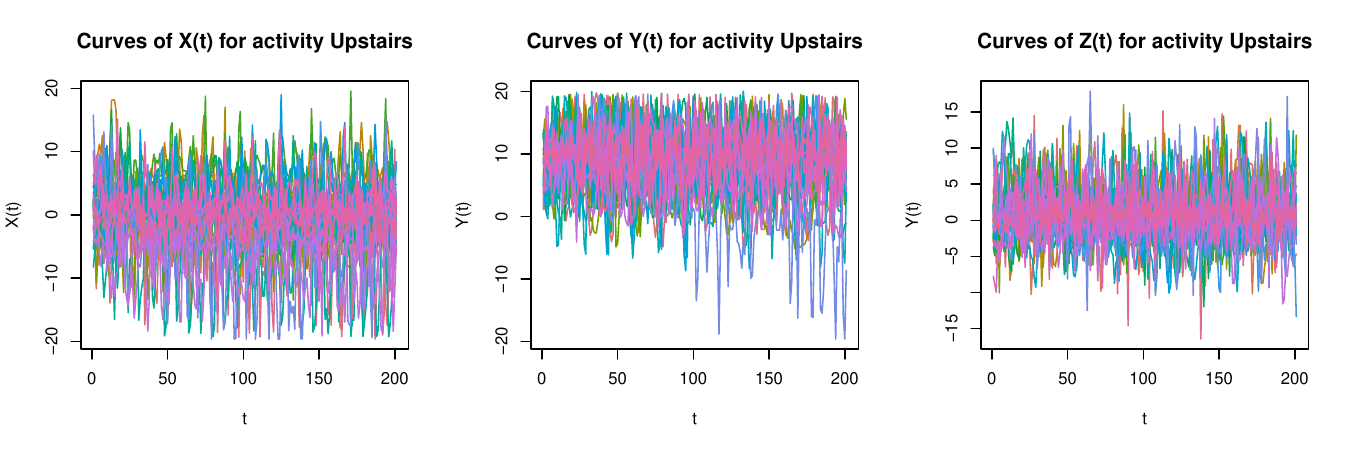}
\caption{Acceleration curves for Upstairs based on $n=32$ users. The panel arrangement and observation window are the same as in Figure \ref{fig-curve-1}.}
\label{fig-curve-3}
\end{figure}

\begin{figure}[!htb] 
\centering 
\includegraphics[width=\linewidth]{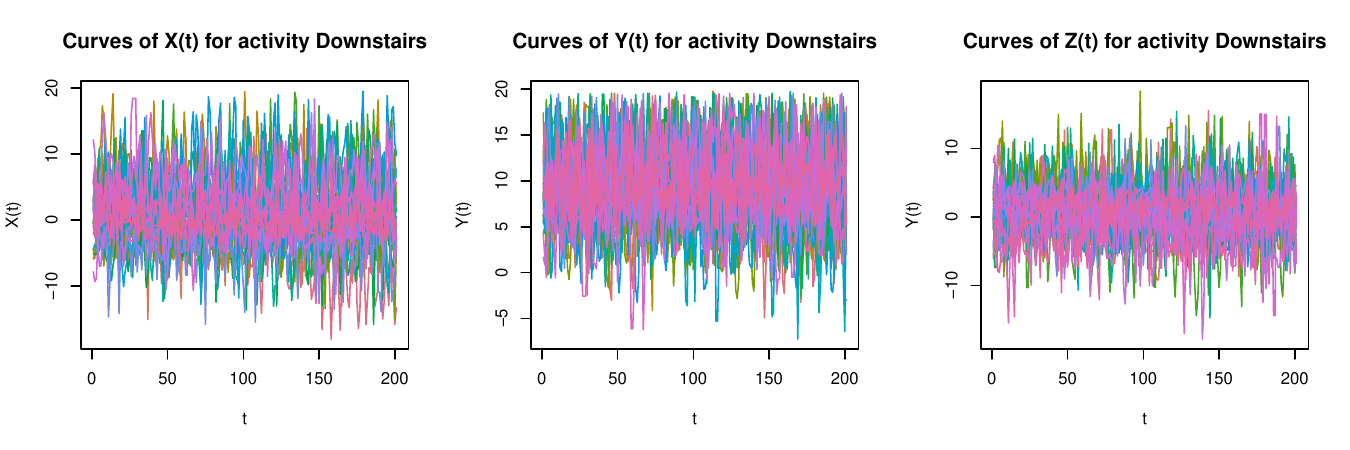}
\caption{Acceleration curves for Downstairs based on $n=32$ users. The panel arrangement and observation window are the same as in Figure \ref{fig-curve-1}.}
\label{fig-curve-4}
\end{figure}

\begin{figure}[!htb] 
\centering 
\includegraphics[width=\linewidth]{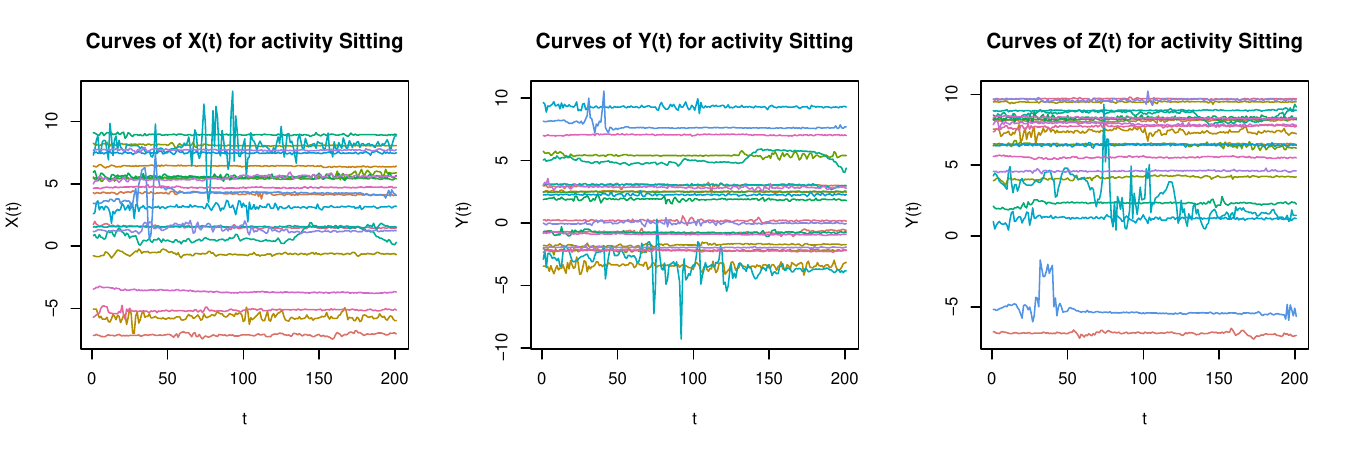}
\caption{Acceleration curves for Sitting based on $n=23$ users. The panel arrangement and observation window are the same as in Figure \ref{fig-curve-1}.}
\label{fig-curve-5}
\end{figure}

\begin{figure}[!htb] 
\centering 
\includegraphics[width=\linewidth]{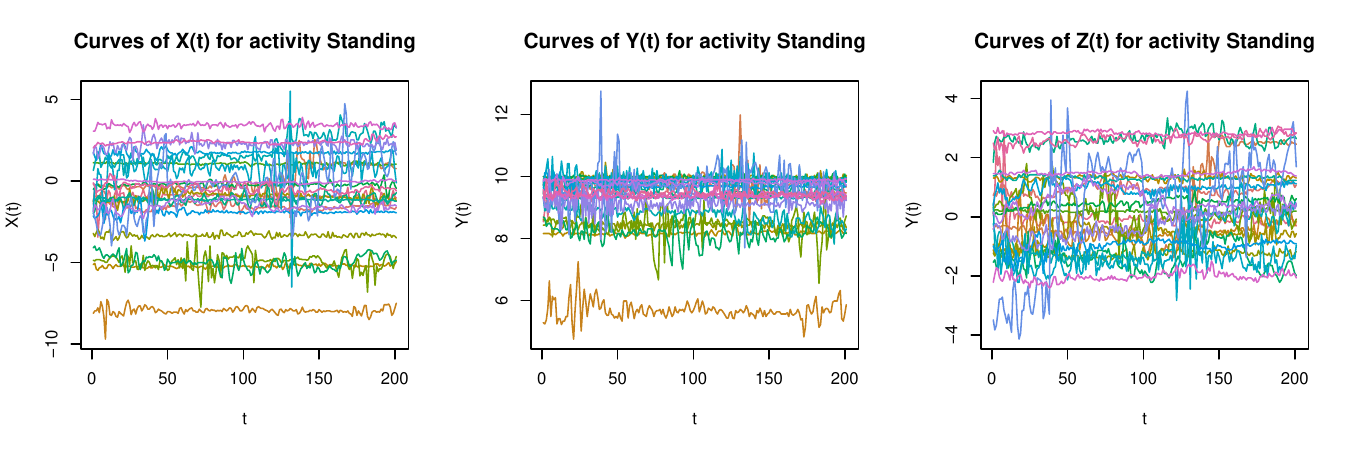}
\caption{Acceleration curves for Standing based on $n=24$ users. The panel arrangement and observation window are the same as in Figure \ref{fig-curve-1}.}
\label{fig-curve-6}
\end{figure}

\begin{figure}[!htbp] 
\centering 
\includegraphics[width=\linewidth]{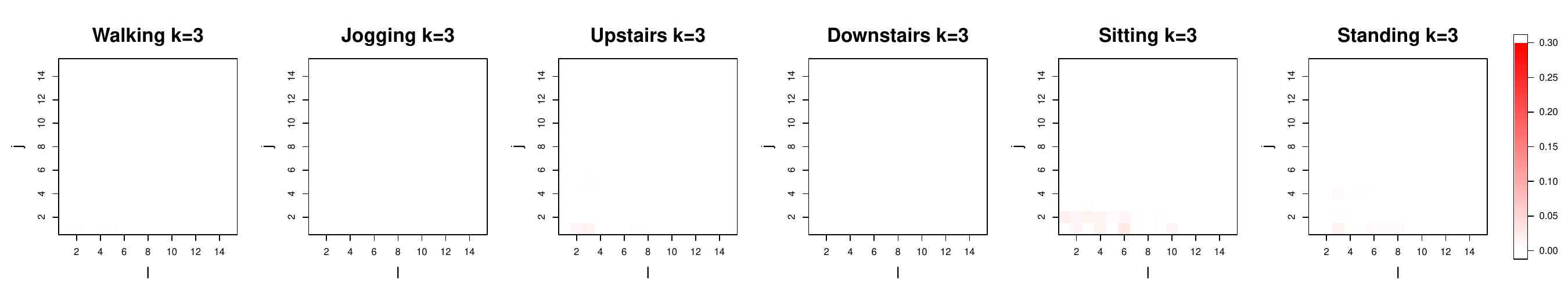}
\includegraphics[width=\linewidth]{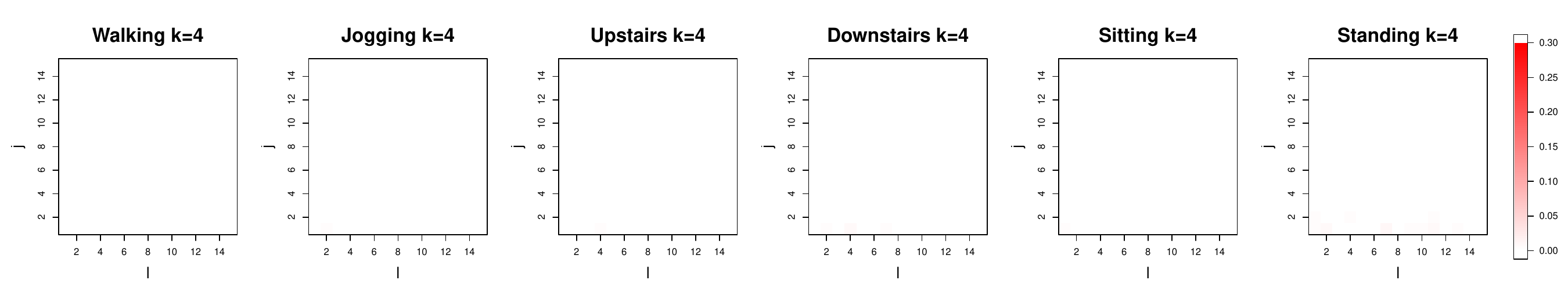}
\includegraphics[width=\linewidth]{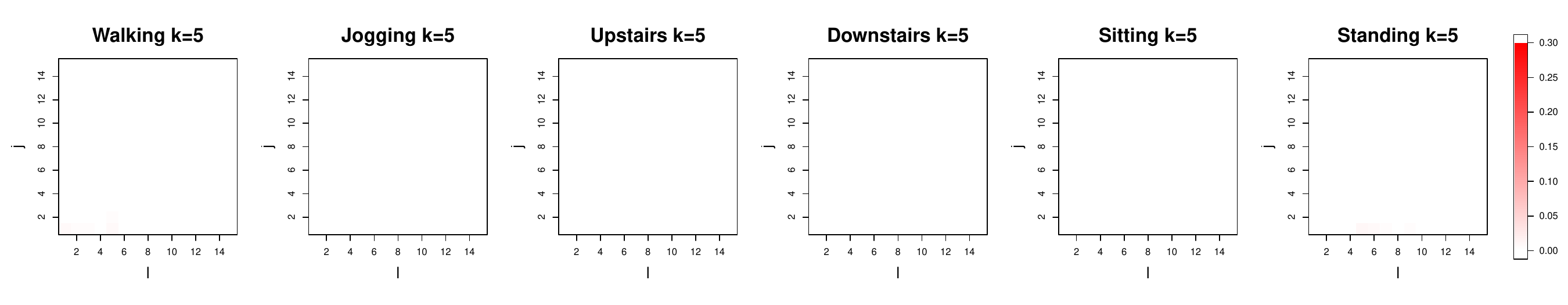}
\caption{Heatmaps of $s_{jk\ell}^{(3)}$ for $j,\ell\in\{1,\ldots,15\}$, which measures a scaled third-order association involving the $j$-th component on the $\ddot{X}$ side and the $k$-th and $\ell$-th components of $Z$. The three rows correspond to $k=3$, $4$, and $5$. All panels use the same color scale from $0$ to $0.3$, with darker colors indicating larger values. Most values are close to zero.}
\label{fig-corr-xz-2} 
\end{figure}

\begin{table}[!htbp]
    \centering
\begin{tabular}{rcccccc}
\hline
 $k$ & Walking & Jogging & Upstairs & Downstairs & Sitting & Standing\\
\hline
1 & 0.0154 & 0.0042 & 0.0112 & 0.0096 & 0.2636 & 0.1413\\
2 & 0.0113 & 0.0038 & 0.0158 & 0.0080 & 0.1371 & 0.0790\\
3 & 0.0044 & 0.0027 & 0.0172 & 0.0029 & 0.0258 & 0.0174\\
4 & 0.0020 & 0.0030 & 0.0041 & 0.0075 & 0.0039 & 0.0096\\
5 & 0.0074 & 0.0025 & 0.0025 & 0.0027 & 0.0019 & 0.0108\\
6 & 0.0017 & 0.0014 & 0.0063 & 0.0026 & 0.0031 & 0.0038\\
7 & 0.0018 & 0.0016 & 0.0045 & 0.0027 & 0.0002 & 0.0010\\
8 & 0.0011 & 0.0013 & 0.0022 & 0.0012 & 0.0000 & 0.0015\\
9 & 0.0011 & 0.0016 & 0.0028 & 0.0009 & 0.0000 & 0.0013\\
10 & 0.0016 & 0.0012 & 0.0017 & 0.0013 & 0.0000 & 0.0005\\
11 & 0.0012 & 0.0020 & 0.0025 & 0.0017 & 0.0000 & 0.0007\\
12 & 0.0016 & 0.0014 & 0.0013 & 0.0010 & 0.0000 & 0.0002\\
13 & 0.0007 & 0.0012 & 0.0014 & 0.0013 & 0.0000 & 0.0002\\
14 & 0.0008 & 0.0006 & 0.0016 & 0.0015 & 0.0000 & 0.0002\\
15 & 0.0009 & 0.0012 & 0.0014 & 0.0015 & 0.0000 & 0.0000\\
\hline
\end{tabular}
    \caption{Maximum values of $s_{jk\ell}^{(3)}$ over $j,\ell\in{1,\ldots,15}$ for each activity and $k\in{1,\ldots,15}$. For $k\geq6$, all reported values are below $0.01$.}
    \label{tab:maximum-corr-xz}
\end{table}

\newpage

\section{Application to WDI Dataset}
We now apply our test to several indicators from the World Development Indicators
(WDI) database \citep{WDI} to investigate whether pairs of development
indicators are conditionally independent given a country's GDP.\footnote{The dataset
is available at \url{https://datacatalog.worldbank.org/search/dataset/0037712/World-Development-Indicators},
version 123 (metadata last updated on October 9, 2025). We downloaded the CSV
archive and used the file \texttt{WDICSV.csv}.}
In other words, we examine whether the relationship between two socioeconomic
factors can be explained solely by the level of economic development.
The conditioning variable is \texttt{NY.GDP.MKTP.CD}, which represents GDP
(current US\$). The candidate factors considered in our analysis are listed in
Table~\ref{tab:factors}.

\begin{table}[!htb]
    \centering
\begin{tabular}{lll}
\hline
 & Variable & Explanation \\
\hline
1 & \texttt{AG.LND.AGRI.ZS} & Agricultural land (\% of land area) \\
2 & \texttt{NV.AGR.TOTL.ZS} & Agriculture, forestry, and fishing, value added (\% of GDP) \\
3 & \texttt{SP.POP.GROW} & Population growth (annual \%) \\
4 & \texttt{FP.CPI.TOTL.ZG} & Inflation, consumer prices (annual \%) \\
5 & \texttt{FR.INR.RINR} & Real interest rate (\%) \\
6 & \texttt{SL.UEM.TOTL.NE.ZS} & Unemployment, total (\% of total labor force) (national \\
& & estimate) \\
7 & \texttt{SP.DYN.LE00.IN} & Life expectancy at birth, total (years)\\
\hline
\end{tabular}
    \caption{WDI factors included in the conditional independence analysis. The table gives the WDI variable code and description for each factor. GDP in current US dollars is used as the conditioning variable.}
    \label{tab:factors}
\end{table}

We use data from 2000 to 2023, so that $m=23$,  and treat the observations
as a balanced design. For each pair of factors together with GDP, we
remove all countries with missing values. The resulting sample sizes for
all tests are reported in Table~\ref{tab:sample-size}. 
For each pair among the seven selected factors, we conduct our CCCO-based test for conditional independence given GDP at the significance level
$\alpha=0.05$. The resulting conditional dependency structure among the
factors is shown in Figure~\ref{fig-app-plot}, where an edge between two
factors indicates that they are not conditionally independent given GDP.
All corresponding $p$-values are reported in
Table~\ref{tab:pval-ccco}.

\begin{table}[!htb]
    \centering
\begin{tabular}{l|rrrrrrr}
\hline
 & 1 & 2 & 3 & 4 & 5 & 6 & 7\\
\hline
1 &  & 163 & 188 & 58 & 70 & 57 & 189\\
2 & 163 &  & 166 & 53 & 65 & 56 & 167\\
3 & 188 & 166 &  & 59 & 70 & 58 & 194\\
4 & 58 & 53 & 59 &  & 19 & 41 & 59\\
5 & 70 & 65 & 70 & 19 &  & 19 & 71\\
6 & 57 & 56 & 58 & 41 & 19 &  & 58\\
7 & 189 & 167 & 194 & 59 & 71 & 58 & \\
\hline
\end{tabular}
    \captionof{table}{Sample sizes for the pairwise CCCO tests after countries with missing observations are removed. The row and column indices correspond to the factors in Table  \ref{tab:factors}. Sample sizes range from $19$ to $194$, reflecting differences in data availability across factor pairs.}
    \label{tab:sample-size}
\end{table}

\begin{figure}[!htb]
\centering
\begin{subfigure}{0.48\textwidth}
\centering
\includegraphics[width=\linewidth]{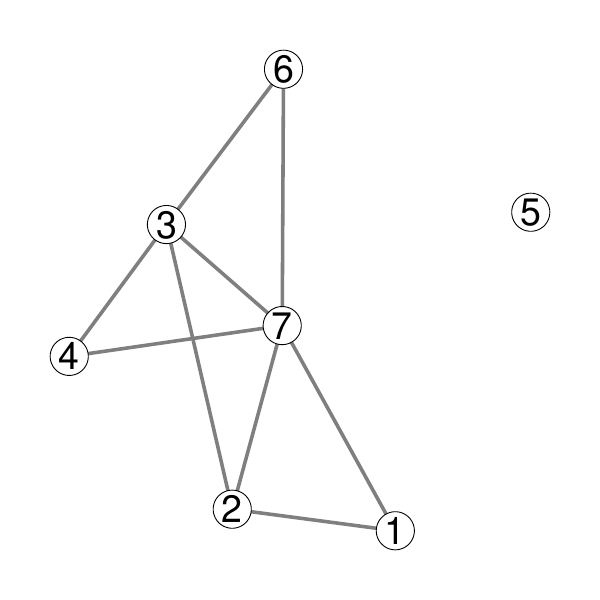}
\caption{Results using weighted sum of $\chi^2_{(1)}$ distributions.}
\label{fig-app-plot-raw} 
\end{subfigure}
\hfill
\begin{subfigure}{0.48\textwidth}
\centering
\includegraphics[width=\linewidth]{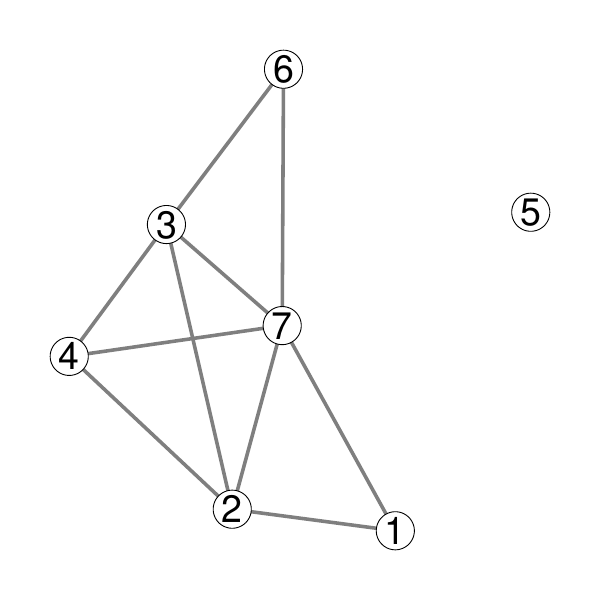}
\caption{Results with Welch-Satterthwaite approximation.}
\label{fig-app-plot-sw} 
\end{subfigure}
\caption{Conditional dependency graphs for the seven WDI factors in Table \ref{tab:factors} given GDP. An edge indicates rejection of conditional independence at $\alpha=0.05$. Panels (a) and (b) use the weighted sum of $\chi^2_{(1)}$ distributions and the Welch-Satterthwaite approximation, respectively. The graphs differ only in the edge between Factors 2 and 4, while Factor 5 is isolated in both panels.}
\label{fig-app-plot} 
\end{figure}

\begin{table}[!htb]
    \centering
\begin{tabular}{l|rrrrrrr}
\hline
& 1 & 2 & 3 & 4 & 5 & 6 & 7\\
\hline
1&  & 0.0013 & 0.1893 & 0.2925 & 0.1553 & 0.6596 & 0.0036\\
2& 0.0013 &  & 0.0000 & 0.1029 & 0.2221 & 0.3105 & 0.0000\\
3& 0.1893 & 0.0000 &  & 0.0469 & 0.1373 & 0.0132 & 0.0000\\
4& 0.2925 & 0.1029 & 0.0469 &  & 0.2475 & 0.1898 & 0.0428\\
5& 0.1553 & 0.2221 & 0.1373 & 0.2475 &  & 0.3454 & 0.2258\\
6& 0.6596 & 0.3105 & 0.0132 & 0.1898 & 0.3454 &  & 0.0242\\
7& 0.0036 & 0.0000 & 0.0000 & 0.0428 & 0.2258 & 0.0242 & \\
\hline
\end{tabular}
    \caption{CCCO $p$-values computed using the weighted sum of $\chi^2_{(1)}$ distributions. The row and column indices correspond to the factors in Table \ref{tab:factors}. Values below $0.05$ correspond to the edges in Figure \ref{fig-app-plot-raw}.}
    \label{tab:pval-ccco-raw}
\end{table}

\begin{table}[!htb]
    \centering
\begin{tabular}{l|rrrrrrr}
\hline
& 1 & 2 & 3 & 4 & 5 & 6 & 7\\
\hline
1&  & 0.0000 & 0.1484 & 0.2795 & 0.0962 & 0.7570 & 0.0000\\
2& 0.0000 &  & 0.0000 & 0.0479 & 0.1740 & 0.3002 & 0.0000\\
3& 0.1484 & 0.0000 &  & 0.0097 & 0.0733 & 0.0006 & 0.0000\\
4& 0.2795 & 0.0479 & 0.0097 &  & 0.2148 & 0.1392 & 0.0084\\
5& 0.0962 & 0.1740 & 0.0733 & 0.2148 &  & 0.3512 & 0.1786\\
6& 0.7570 & 0.3002 & 0.0006 & 0.1392 & 0.3512 &  & 0.0021\\
7& 0.0000 & 0.0000 & 0.0000 & 0.0084 & 0.1786 & 0.0021 & \\
\hline
\end{tabular}
    \caption{CCCO $p$-values computed using the Welch–Satterthwaite approximation. The row and column indices correspond to the factors in Table \ref{tab:factors}. Values below $0.05$ correspond to the edges in Figure \ref{fig-app-plot-sw}.}
    \label{tab:pval-ccco}
\end{table}

As shown in Figure~\ref{fig-app-plot}, life expectancy at birth (Factor~7)
is conditionally dependent on each of the other factors except the real
interest rate (Factor~5) given GDP. This suggests that the relationships
between life expectancy and these factors cannot be explained solely by
GDP. In contrast, the real interest rate (Factor~5) appears to be
conditionally independent of all other factors given GDP, indicating
that its associations with the remaining variables are largely mediated
through GDP. 
{Note that, comparing Figures~\ref{fig-app-plot-raw} and \ref{fig-app-plot-sw}, the dependence results are only different between Factors~2 and 4, but the corresponding p-value with Welch-Satterthwaite approximation is 0.0479, which is very close to $\alpha = 0.05$. }
We should mention  that the sample sizes associated with Factors~4, 5, and 6
are relatively small. Therefore, the conclusions should be interpreted
with caution. In particular, the power of the asymptotic test may be
limited when the sample size is small, and the presence of missing data
may introduce additional bias. Further investigation may therefore be
warranted to better understand the relationships among these variables.

\section*{Acknowledgments}
We thank two referees for their insightful comments and suggestions, which helped us greatly in improving this work. We thank Wenxi Tan {and Haoning Chen} for pointing out {two errors} in the proof in {earlier versions} of the manuscript. Bing Li's research is partly supported by NIH 1 R01 GM152812-01. 
We thank the University of Kentucky Center for Computational Sciences and Information Technology Services Research Computing for their support and use of the Morgan Compute Cluster and associated research computing resources.

\bibliographystyle{asa}
\bibliography{biblio}

\end{document}